%% file: JournalCCbA.tex
\documentclass[smallextended]{svjour3}       
\def\makeheadbox{\relax}
\usepackage{graphicx}
\usepackage{amsmath}

\usepackage{amssymb}
\usepackage{amsthm}
\usepackage{dsfont}
\usepackage{subfigure}
\usepackage{tabularx}
\usepackage{algorithm}
\usepackage{color}
\usepackage{epsfig}
\usepackage{url}
\usepackage{mathtools}
\usepackage{algcompatible}
\DeclarePairedDelimiter{\ceil}{\lceil}{\rceil}
\DeclarePairedDelimiter{\floor}{\lfloor}{\rfloor}

\newtheorem{Corollary}{Corollary}
\newtheorem{Definition}{Definition}

\date{}

\usepackage{tabularx}
\graphicspath{ {figures/} }

\begin{document}
	%
	
	\title{Network Calculus-based Timing Analysis of AFDX networks incorporating multiple TSN/BLS traffic classes}
	
	\author{Ana\"\i s FINZI, Ahlem MIFDAOUI, Fabrice FRANCES, Emmanuel LOCHIN\\
	\email{firstname.lastname@isae-supaero.fr}	\\}
	
	\institute {University of Toulouse,ISAE-SUPAERO, France}

\authorrunning{A. FINZI et al.}
\titlerunning{Timing analysis of AFDX networks with multiple TSN/BLS}

\setcounter{tocdepth}{3}
\maketitle

\begin{abstract}
	We propose a formal timing analysis of an extension of the AFDX standard, incorporating the TSN/BLS shaper, to homogenize the avionics communication architecture, and enable the interconnection of different avionics domains with mixed-criticality levels, e.g., current AFDX traffic, Flight Control and In-Flight Entertainment. Existing Network Calculus models are limited to three classes, but applications with heterogeneous traffic require additional classes. Hence, we propose to generalize an existing Network Calculus model to do a worst-case timing analysis of an architecture  with multiple BLS on  multi-hop networks, to infer real-time bounds. Then, we conduct the performance analysis of such a proposal. First  we evaluate the model on a simple 3-classes single-hop network to assess the sensitivity and tightness of the model, and compare it to existing models (CPA and Network Calculus). Secondly, we study a realistic AFDX configuration with six classes and two BLS. Finally, we compute a real use-case to add A350 flight control traffic to the AFDX. Results show the good properties of the generalized Network Calculus model compared to the CPA model and the efficiency of the extended AFDX to noticeably enhance the medium priority level delay bounds, while respecting the higher priority level constraints, in comparison with the current AFDX standard. 
\end{abstract}

\newpage
\tableofcontents
\newpage
\begin{table}[h!]

	\caption{Notations}
	\begin{center}
		\begin{tabularx}{\linewidth} {l X}
			\hline
			$C$ & Link speed\\
			$MFS_i$ & Maximum Frame Size of flow or class $i$\\
			$J_f,Dl_f$,$BAG_f$ & Jitter, deadline and BAG of flow $f$\\				
			$L_M^k, L_R^k$ & BLS maximum and resume credit levels of class $k$\\	
			$L_M, L_R$ & BLS maximum and resume credit levels when considering a single shaped class\\
			$BW^k$ & BLS reserved bandwidth of class $k$\\
			$BW$ & BLS reserved bandwidth when considering a single shaped class\\
			$I_{idle}^k,I_{send}^k$ & BLS idle and sending slopes of class $k$, defined in Eq. (\ref{iidle}) and Eq. (\ref{isend})\\
			$I_{idle},I_{send}$ & BLS idle and sending slopes  when considering a single shaped class\\
			$p(k)$ & Priority level of a class $k$ \\ 
			$p_H(k)$ & BLS high priority of class $k$\\
			$p_L(k)$ & BLS low priority of class $k$\\
			$HC(k)$ & Set of flows or classes with a priority strictly higher than $p_H(k)$, i.e., $\forall j$ such as: $p_H(k)>p(j)$\\
			$LC(k)$ & Set of flows or classes with a priority strictly lower than $p_L(k)$, i.e., $\forall j$ such as: $p_L(k)<p(j)$\\
			$MC(k)$ & Set of flows or classes with a priority strictly between $p_L(k)$ and $p_H(k)$\\
			$UR_{k}$& Utilisation rate of a class $k$ at the input of a output port\\
			$UR^{bn}_{k}$& The bottleneck network utilisation rate of a class $k$ \\
			$n_{k}^{es}$ & Number of flows of class $k$ generated per node $es$\\			
			$\gamma_{k}^{n}$ & Maximum service curve guaranteed for the traffic class $k$ within node $n$\\
			$\gamma_k^{bls,fluid}$ & Maximum service curve of class $k$ in the BLS node when considering fluid traffics\\
			$\beta_{k,f}^{n}$ & Strict minimum service curve guaranteed to a flow $f$ of class $k$  in a node $n  \in \{es, mux\}$\\
			$\beta_{k}^{n}$ & Strict minimum service curve guaranteed for the traffic class $k$ in a node $n  \in \{es, mux \}$ (end-system or output multiplexer) or component $n \in \{bls, sp \}$\\			
			$\beta_{k\in BLS, p}^{sp}$ & Strict minimum service curve guaranteed to BLS class $k$ when having the priority level $p$ in a $sp$ component \\
			$\beta_k^{bls,fluid}$ & Strict minimum service curve of class $k$ in the BLS node when considering fluid traffics\\		
			$\alpha_{k,f}^{n}$ & Input arrival curve of the flow $f$ of class $k$ in the node $n  \in \{es, mux \}$ or component $n  \in \{bls, sp \}$ in its path \\
			$\alpha_k^n$ & Input arrival curve of the aggregated flows of class $k$ in a node $n  \in \{es, mux\}$ or component $n  \in \{bls, sp \}$ \\
			$\alpha_{k,f}^{*,n}$ & Output arrival curve of the flow f of class $k$ from the node $n  \in \{es, mux\}$ or component $n  \in \{bls, sp \}$ in its path\\
			$\alpha_k^{*,n}$ & Output arrival curve of the aggregated flows of class $k$ from a node $n  \in \{es, mux \}$\ or a component $n  \in \{bls, sp \}$\\
			$Deadline_{k,f}^{end2end}$ & End-to-end  deadline of flow $f$ of class $k$\\
			$delay_{k,f}^{end2end}$ & End-to-end delay of flow $f$ of class $k$\\
			$delay_{k,f}^{n}$ & Delay of flow $f$ of class $k$ in a node $n\in \{es,sw,mux\}$\\
			$delay_{k,f}^{prop}$ & Propagation delay of flow $f$ of class $k$\\
			\hline												
		\end{tabularx}
	\end{center}
	\label{notations1}
\end{table}

 \newpage
\section{Introduction}
\label{intro}
The growing number of interconnected end-systems and the expansion of exchanged data in avionics have led to an increase in complexity of the communication architecture. To cope with this trend, a first communication solution based on a high rate backbone network, i.e., the AFDX (Avionics Full Duplex Switched Ethernet) \cite{ARINC664}, has been implemented by Airbus in the A380, to interconnect critical subsystems. Moreover, some low rate data buses, e.g., CAN \cite{CAN}, are still used to handle some specific avionics domains, such as the I/O process and the Flight Control Management. Although this architecture reduces the time to market, it conjointly leads to inherent heterogeneity and new challenges to guarantee the real-time requirements.
	
To cope with these emerging issues, with the maturity and reliability progress of the AFDX after a decade of successful use, a homogeneous avionic communication architecture based on such a technology to interconnect different avionics domains may bring significant advantages, such as quick installation and maintenance and reduced weight and costs. 

Furthermore, this new communication architecture needs to support, in addition to the current AFDX traffic profile, called Rate Constrained (RC) traffic, at least two extra profiles. The first, denoted by Safety-Critical Traffic (SCT), is specified to support flows with hard real-time constraints and the highest criticality level, e.g., flight control data; whereas the second is for Best-Effort (BE) flows with no delivery constraint and the lowest criticality level, e.g., In-Flight Entertainment traffic. 

Various fair solutions exist to solve this issues, for instance implementing well-known scheduling schemes, e.g., Deficit-Round-Robin (DRR) \cite{hua2012scheduling}, Weighted-Round-Robin (WRR) \cite{tianran2012design}. However, these solutions are notably hard to tune.

More recently, the Audio Video Bridging (AVB) Task Group proposed a new shaper: the Credit-Based Shaper (CBS) \cite{diemer2012formal}. However, it has a severe limitation: it is a blocking shaper, which may cause undue delays. To fix this issue, the Time Sensitive Networking (TSN) task group has proposed the Burst Limiting Shaper (BLS). 

As these are promising solutions, they have gathered interest from the automotive \cite{thiele2016formal}, avionics \cite{schneele2012comparison}  and internet communities \cite{globefinzi17}. In this last work, simulations showed the interest of the BLS as it offers a new, more predictable service. However, for time sensitive traffic in automotive, avionics or satellite applications, a formal analysis to compute worst-case latencies is needed. Additionally, for an avionics use, worst-case bounds are required to obtain the certification. A first work in \cite{erts2finzi17} proposed a detailed analysis of these mixed-criticality solutions.

Additionally, Urgency Based Scheduler (UBS) \cite{specht2016urgency} is a novel fair scheduler with good modularity and predictability, that had not been considered in \cite{erts2finzi17}. In \cite{gavrilut2017fault}\cite{specht2016urgency}, they conclude that the implementation complexity is low, in part because they assume the queue selection process is already implemented in the switches thanks to the standardization of 802.1Qci-Per-Stream Filtering and Policing. But while implementing it in higher layer is simple, implementing at the hardware level for avionics is much {more complex}.

Finally, in \cite{erts2finzi17}, we concluded that in the avionics case, the BLS is the most promising solution offering high modularity, fairness, predictability and low complexity.

There are several existing works of BLS formal analysis, the most prominent one being a CPA analysis for automotive applications \cite{thiele2016formal}. However while very complete, this modelisation is complex to implement and requires computation power. Moreover, as we will show, the model can be optimistic. Additionally,  a formal analysis based on the Network Calculus is more scalable compared to one based on CPA \cite{Perathoner08}, and  Network Calculus has already been used to certified the AFDX \cite{grieu2004analyse}. Thus, we have proposed two Network Calculus models in \cite{Finzi-sies-18}\cite{Finzi-wfcs-18}. Both these models are limited to three classes. Hence, in this paper we propose an generalization of the tightest of these two Network Calculus models to multiple classes, multiple BLS and multi-hop networks.

Therefore, our main contributions in this paper are two-fold: (i) \textbf{first}  an appropriate system modeling and timing analysis, based on the Network Calculus framework, generalizing the Continuous Credit-based Approach to multiple BLS in multi-hop networks; (ii) \textbf{second} the evaluation of the proposed model, including a use-case with two BLS and a used-case with the concrete challenge of adding the A350 flight control traffic on the AFDX.

First in Section~\ref{RelatedWork} we present the background and related work on the BLS. Then, in Section~\ref{NCframework} we present the Network Calculus framework. Next, in Sections~\ref{systModel} and~\ref{WFA} we define the system model and the associated worst-case timing analysis. Finally, in Section~\ref{PA} we present a performance analysis of the Extended AFDX.

\section{Background and Related Work}
	\label{RelatedWork}
First in this section, we detail the Burst Limiting Shaper (BLS). Then, we present the existing worst-case timing analysis of this solution.

\subsection{Burst Limiting Shaper}
\label{basicconcepts}
The BLS belongs to the credit-based shaper class. Each shaped queue is associated to a class $k$ and has been defined in \cite{Gotz2012} by an upper threshold $L_M^k$, a lower threshold $L_R^k$, such as $0\leqslant L_R^k < L_M^k$, and a reserved bandwidth $BW^k$. Additionally, the priority of a class $k$ shaped by BLS, denoted $p(k)$, can vary between a high and a low value , denoted $p_H(k)$ and $p_L(k)$ (with priority 0 the highest priority, and $p_L(k)>p_H(k)$). The low value is usually below the lowest priority of the unshaped traffic. In the avionic context, to guarantee the safety isolation level between the different traffic profiles, the low value associated to the SCT is set to be lower than the RC priority level, but higher than the BE priority. Therefore, when considering one class for each traffic type, i.e. three classes, SCT queue priority oscillates between 0 (the highest) and 2 (see Fig. \ref{fig:BLSshaper}), RC priority is 1 (see Fig. \ref{fig:BLSshaper}) and BE has the priority 3 (the lowest, see Fig. \ref{fig:BLSshaper}). Thus, when SCT traffic is enqueued, BE traffic can never be sent no matter the state of BLS. In this case, RC is the only traffic that can be sent and this only happens when the SCT priority is 2. As a consequence, BE traffic is isolated from SCT and RC traffics.

\
The credit counter varies as follows:
\begin{itemize} \label{fctBLS}
	\item	 initially, the credit counter starts at 0 and the priority of the queue of the burst limited flows is high (\#0);
	\item the main feature of the BLS is the change of priority $p(k)$ of the  queue of the shaped class, which occurs in two contexts: 1) if $p(k)$ is high and  credit reaches $L_M^k$;  2) if $p(k)$ is low and credit reaches $L_R^k$; \\
	\item when a frame is transmitted, the credit increases (is consumed) with a rate of $I_{send}^k$, else the credit decreases (is gained) with a rate of $I_{idle}^k$;
	\item when the credit reaches $L_M^k$, it stays at this level until the end of the transmission of the current frame;	
	\item when the credit reaches $0$, it stays at this level until the end of the transmission of the current frame (if any). The credit remains at 0 until a new BLS frame is transmitted.
\end{itemize}

\begin{figure}[h]
	\centering
	\includegraphics[width=0.6\columnwidth]{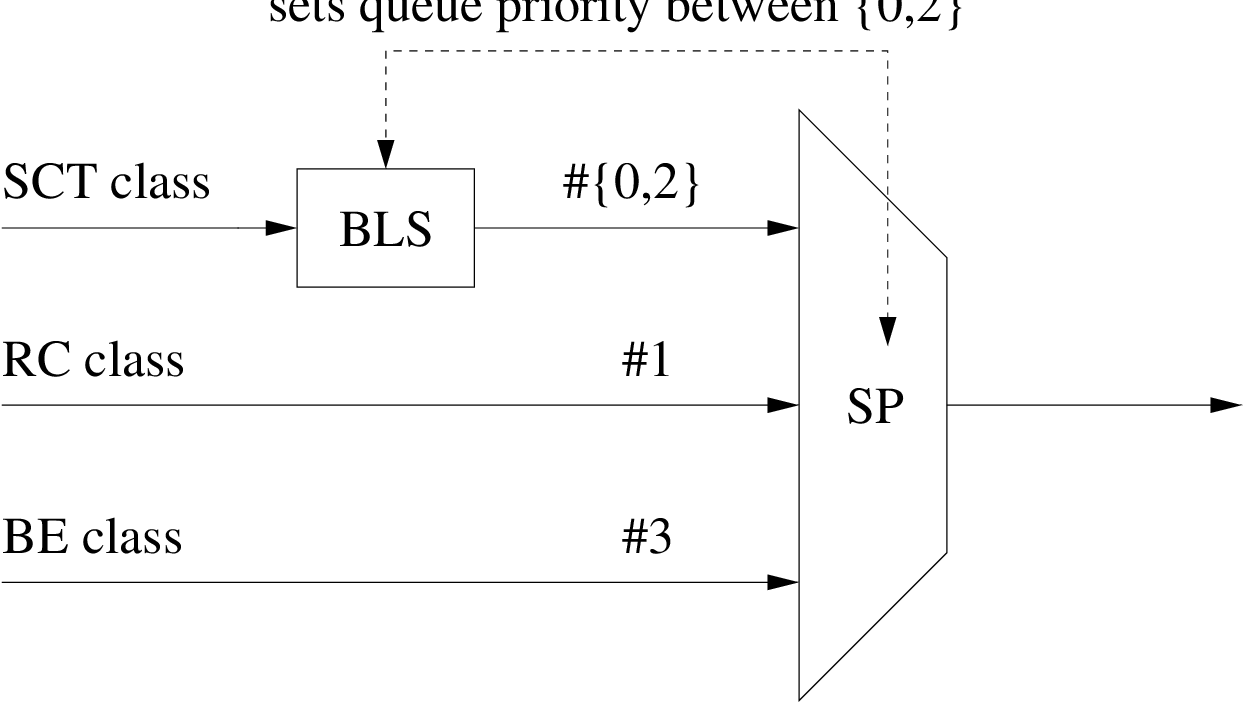}
	\caption{Burst Limiting Shaper on top of NP-SP at the output port with 3 classes }
	\label{fig:BLSshaper}
\end{figure}


The behavior of the BLS is illustrated in Fig. \ref{fig:BLScredit}. As shown, the credit is always between 0 and $L_M^k$. The credit rates of the BLS shaper are defined as follows:\\
\begin{itemize}
	\item the decreasing rate is: 
	\begin{equation}\label{iidle}
	I_{idle}^k = BW^k\cdot C
	\end{equation}
	where $C$ is the link speed and $BW^k$ is the percentage of bandwidth reserved for BLS frames.\\
	\item the increasing rate is: 
	\begin{equation}\label{isend}
	I_{send}^k = C - I_{idle}^k
	\end{equation}
	
\end{itemize}

\begin{figure}[h]
	\centering
		\includegraphics[width=0.7\columnwidth]{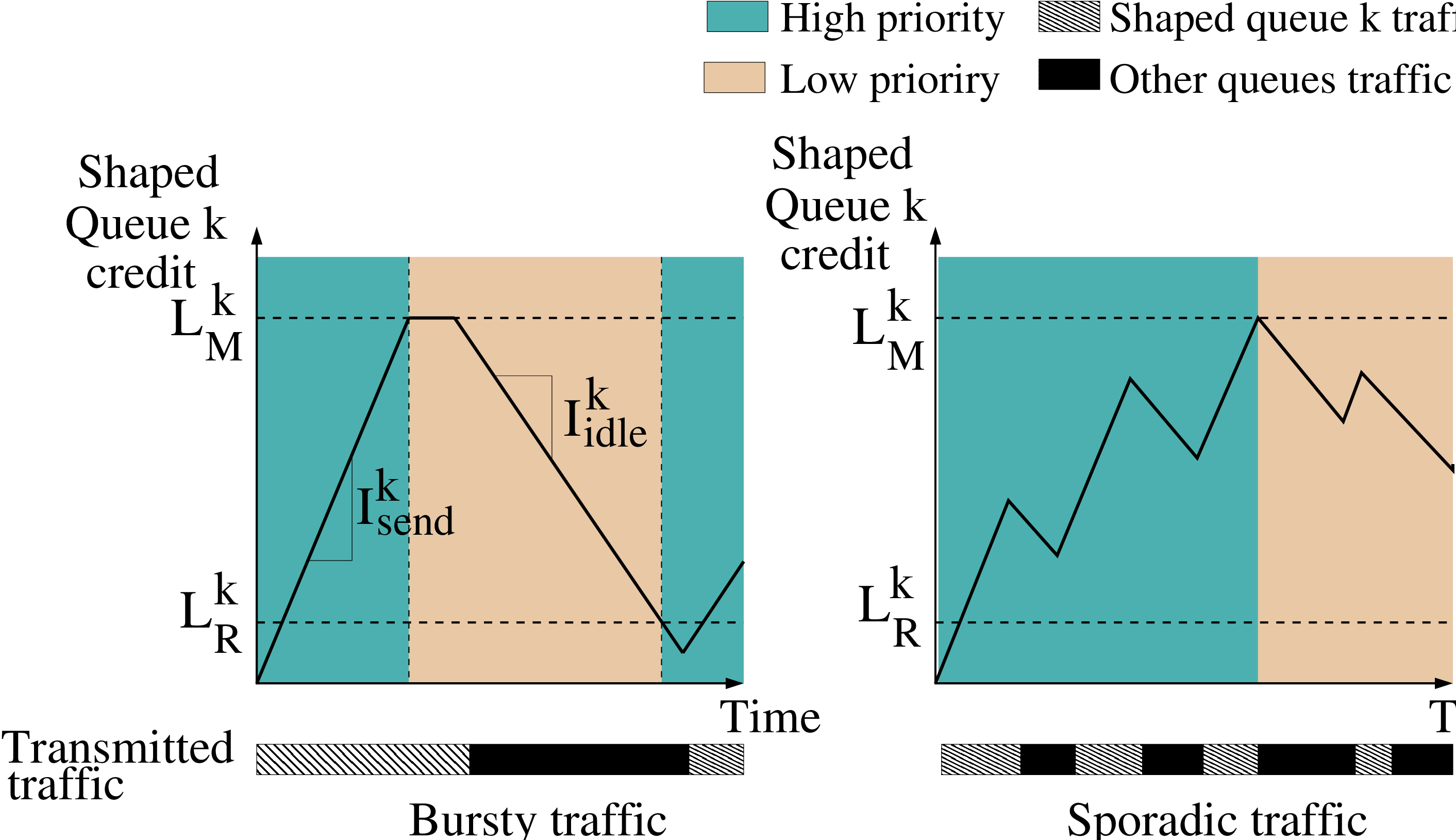}
	\footnotesize \caption{BLS credit evolution}
	\label{fig:BLScredit}
\end{figure}

It is worth noting that with the BLS, both the priority of the queue of the shaped class and the state of all the queues, i.e., empty or not, define whether the credit is gained or lost. This aspect is depicted in Fig. \ref{fig:BLScredit} for two arrival scenarios. The first one (left figure) shows the case of a bursty traffic, where the maximum of traffic shaped by the BLS is sent when its priority is the highest. Consequently, the other priorities send as much traffic as possible when the priority of the  BLS class has the low value. The second one (right figure) is for sporadic traffic, where we can see that when the shaped-class priority is highest but no frame is available, then the credit is regained. However, when the priority is at the low value and the other queues are empty, then shaped-class frames can be transmitted and the credit is consumed.

\subsection{Existing Worst-case Timing Analyses of TSN/BLS Shaper}
\label{WCTA}
In this section, we present the existing work on the BLS formal analysis. Then, we detail the limitations of the main ones, the Compositional Performance Analysis (CPA)  and Network Calculus (NC) models.
There are some interesting approaches in the literature concerning the worst-case timing analysis of TSN network, and more particularly BLS shaper. The first and seminal one in \cite{Kerschbaum2013} introduces a first service curve model to deduce worst-case delay computation. However, this presentation published by the TSN task group has never been extended in a formal paper. The second one has detailed a more formal worst-case timing analysis in \cite{thangamuthu2015analysis}, which also has some limitations. Basically, the proposed model does not take into account the impact of either the same priority flows or the higher ones, which will clearly induce optimistic worst-case delays. Then, a formal analysis of TSN/BLS shaper, based on a CPA method has been proposed in \cite{thiele2016formal}. This approach has handled the main limitations of the model presented in \cite{thangamuthu2015analysis}; and interesting results for an automotive case study have been detailed. However, this method necessitates extensive computation power to solve two maximization problems, an Integral Linear Programming (ILP) problem and a fixed point problem.  
Additionally, contrary the definition in \cite{Gotz2012}, the CPA model considers the BLS as a blocking shaper, i.e. no BLS frame can be sent if their associated credit has reached $L_M$, until the credit has decreased to $L_R$. We will show that this fact may provide optimistic delay bounds and thus false guarantees for messages that will actually miss their deadline in the worst-case.

Finally, two Network Calculus models have been proposed: a Window-based approach (WbA) modelisation in \cite{Finzi-sies-18}, and a Continuous Credit-based Approach (CCbA) modelisation in \cite{Finzi-wfcs-18}. They are both limited to a specific 3-classes architecture. Next, we compare them to identify the tightest one, before generalizing it in Section\ref{WFA}.


\hfill\\
\textbf{CPA model of BLS}

The CPA model \cite{thiele2016formal} computes the impact of the other flows by dividing them in four categories: the lower-priority blocking, the same-priority blocking, the higher priority blocking, and the BLS shaper blocking. 
The latter is defined as follows for a flow of class $I$: $$I^{SB}_i(\delta t)=\lceil\frac{\delta t}{t^{S-}_I}\rceil \cdot t^{R+}_{I}$$ with:
\begin{itemize}		
	\item $L_R^I$, $L_M^I$ and $I_{idle}^I$ BLS parameters of class $I$;
	\item $t^{R+}_I=\lceil \frac{L_M^I-L_R^I}{I_{idle}^I}\rceil+\max_{j\in lp(I)}\frac{MFS_{j}}{C}$, with $MFS_{j}$ the Maximum Frame Size of flow $j$, $lp(I)$ the streams with a priority lower than $I$: the maximum blocking time, called the replenishment interval
	\item $t^{S-}_{I}=\max\Big\{\floor[\bigg]{\frac{L_M^I-L_R^I}{I_{send}^I}},\max_{j\in I}\frac{MFS_{j}}{C}\Big\}$ the shortest service interval for class $I$.
\end{itemize} 

We have identified three main limitations in the CPA model, which may lead to over-pessimistic delay bounds, or worse, optimistic delay bounds, when considering the definition of the BLS proposed in \cite{Gotz2012}.
The first limitation causes pessimism and  concerns the maximum replenishment interval $t^{R+}_I$. The additional frame transmission $\max_{j\in lp(I)}\frac{MFS_{j}}{C}$ considers all the priorities lower than $I$. This computation means two implicit assumptions, which are not necessarily fulfilled in the general case. The first implicit hypothesis is to consider that the priority for I is the BLS high priority. The second implicit hypothesis is the fact that the low BLS priority is the lowest one. The delay caused by $\max_{j\in lp(I)}\frac{MFS_{j}}{C}$ is due to the transmission of a frame while the BLS priority is low, just before the credit reaches the resume level. But only classes with a priority higher than the low BLS priority can be transmitted while BLS frames are enqueued, thanks to the Static Priority Scheduler. Thus,  CPA model considers that all the flows are in $lp(I)$ and the BLS low priority is the lowest one. This may not be the case, especially when multiple BLS are considered. As a consequence, the shaper blocking effect may be overestimated, depending on the maximum frame sizes.

The second limitation also causes pessimism and concerns again the replenishment interval $t^{R+}_I$. The definition of $t^{R+}_I$ is completely independent from the lower priority traffic rates and bursts. As a consequence, if the replenishment intervals are too large in comparison to the traffic load, the shaper blocking is again overestimated: when no lower priority traffic is available, the BLS flows can be sent no matter the state of the credit. Contrary to the hypothesis set in \cite{thiele2016formal}, the BLS \cite{Gotz2012} is actually a non-blocking shaper: only the state of the queues and their respective priorities matter.

Finally, the third limitation is due to the blocking shaper hypothesis as a whole, stating that BLS frames cannot be sent after their associated credit has reached $L_M$, until the credit has decreased to $L_R$. We will show now that this can lead to optimistic bounds when considering the definition of the BLS from \cite{Gotz2012}.


To assess the CPA model optimism,  we consider herein a 3-classes case study illustrated in Fig. \ref{fig:BLSshaper}, where the SCT class is shaped by a BLS with the priorities switching between 0 and 2, RC has the priority 1 and BE the priority 3.

To compute the worst-case delay, an usual assumption is to consider all the traffics are backlogged. In the case of the BLS however, we  show that this may lead to optimistic bounds.

To compute the worst-case delay, we first detail the case where all classes are backlogged. The resulting credit evolution and the SCT output traffic are visible in Fig. \ref{fig:SCTbeta1} in plain line (1). We can see that this is equivalent to the hypothesis in \cite{thiele2016formal} stating that the BLS blocks the BLS-frames while the credit is decreasing between $L_M$ and $L_R$.  

\begin{figure}[h]
	\centering	
	\includegraphics[width=0.65\columnwidth]{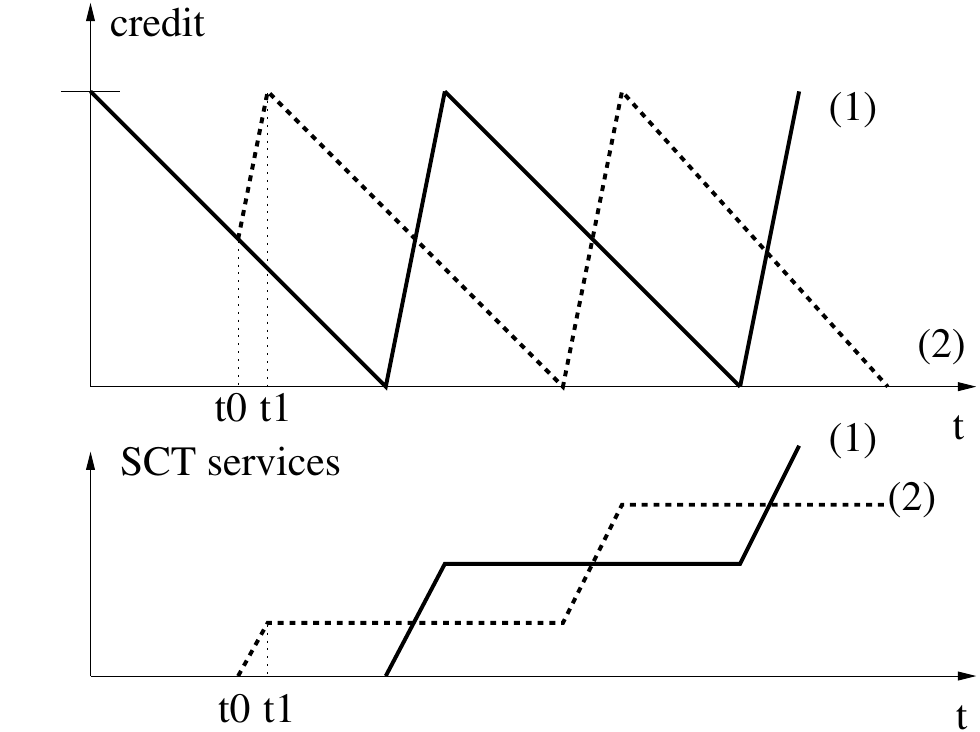}
	\caption{Two examples of worst-case  BLS behaviour}
	\label{fig:SCTbeta1}
\end{figure}

Then, we consider the case where RC traffic is not backlogged between two times $t0$ and $t1$ (dotted lines (2) in Fig. \ref{fig:SCTbeta1}):  
\begin{itemize}
	\item the credit starts at $L_M$ at $ti$ and decreases until it reaches $\frac{L_M}{2}$ at $t0$;	
	\item then it increases until $t1$ when the credit reaches $L_M$;	
	\item finally it decreases until reaching $L_R$ at $t2$.
\end{itemize}
We see in Fig. \ref{fig:SCTbeta1} that in this particular case, the SCT output corresponding to the dotted line (2) can be below the one corresponding to the plain line (1).  
\textbf{This shows that the most intuitive worst-case SCT output, i.e., all traffic are backlogged, is not actually the worst-case SCT output. As a consequence, the shaper blocking hypothesis of the CPA model causes optimism when considering the BLS as defined in \cite{Gotz2012}.}

From both scenarios presented in Fig. \ref{fig:SCTbeta1}, we have computed in Appendix~\ref{AWCs} two \textbf{\textit{Achievable Worst-Case}} delay bounds for SCT. In Section~\ref{Use-case1}, we show the optimism of the CPA model in reference to these achievable worst-case delays.

\hfill\\
\textbf{Existing 3-classes Network Calculus models}

The inherent idea of the WbA \cite{Finzi-sies-18} is based on the different possible combinations of idle and sending BLS windows to model the minimum and maximum service curves, i.e. availability of the traversed node. However, when taking a closer look at the credit behavior of the BLS covering the worst-case scenario of the minimum service curve, $\beta_{k}^{bls}(t)$, we found a pessimism inherent to the WbA model. We illustrate this behavior in Fig. \ref{fig:discontinuities2}, where the minimum sending window $\Delta_{send}^{k,min}$ and the maximum idle window $\Delta_{idle}^{k,max}$ introduce a credit discontinuity, which is not a realistic behavior. Moreover, we have also noticed a similar discontinuity when studying the best-case scenario of the maximum service curve.

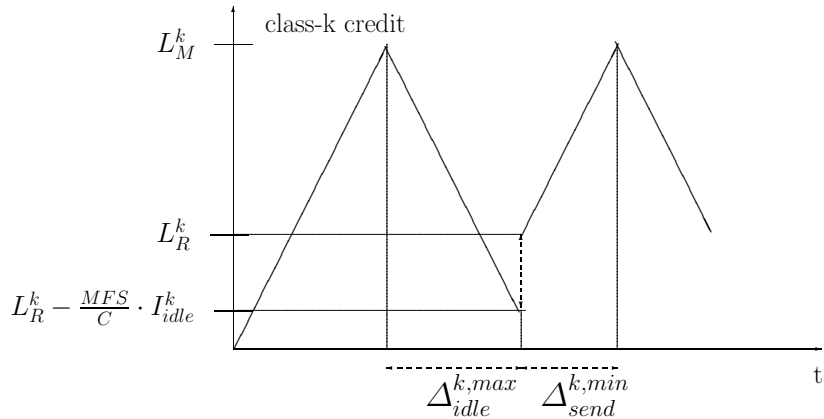
\begin{figure}[h]
	\centering	
	\resizebox{0.9\columnwidth}{!}{\input{figures/discontinuity2.latex}}
	\caption{WbA \cite{Finzi-sies-18}: discontinuities with $\beta_{k}^{bls}(t)$ windows }
	\label{fig:discontinuities2}
\end{figure}

We notice that the discontinuities of the BLS credit happen between the end of the idle window and the start of the sending window for both the minimum and maximum service curves. This issue is situated around $L_R^k$. This highlights the fact that $L_R^k$ is not taken into account in an accurate way by the WbA model. This led  to the second model, CCbA \cite{Finzi-wfcs-18}, which is based on the continuity of the credit. The results presented in \cite{Finzi-wfcs-18} show the tightness of the CCbA.

Hence, in this paper, we have selected the CCbA model to be generalized.

\section{Network Calculus framework}

\label{NCframework}
The timing analysis used here is based on Network Calculus theory \cite{leboudecthiran12} providing upper bounds on delays and backlogs. Delay bounds depend on the traffic arrival described by the so called \textit{arrival curve} $\alpha$, and on the availability of the traversed node described by the so called minimum \textit{service curve} $\beta$. The definitions of these curves are explained as following.
\begin{Definition}[Arrival Curve]
	\label{def:arrivalCurve}
	\cite{leboudecthiran12} A function $\alpha(t)$ is an arrival curve for a data flow with an input cumulative function $R(t)$,i.e., the number of bits received until time $t$, iff:
	\begin{displaymath}
	\forall t, R(t) \leq  R \otimes\footnote{$f \otimes g (t) = \inf_{0 \leq s \leq t}\{f(t-s) + g(s)\}$} \alpha(t)
	\end{displaymath}
\end{Definition}
\begin{Definition}[Strict minimum service curve]
	\label{def:strict-min-service-curve}
	\cite{leboudecthiran12} The function $\beta $ is the minimum \emph{strict} service curve for a data flow with an output cumulative function $R^*$, if for any backlogged period $]s,t]$\footnote{ $]s,t]$ is called backlogged period if $R(\tau) -R^*(\tau) >0, \forall \tau \in ]s,t] $}, $R^*(t) - R^*(s) \geq \beta(t-s)$.
\end{Definition}
\begin{Definition}[Maximum service curve]\label{def:max-service-curve}
	\cite{leboudecthiran12} The function $\gamma(t)$ is the maximum  service curve for a data flow with an input cumulative function $R(t)$ and output cumulative function $R^*(t)$ iff:
	\begin{displaymath}
	\forall t,	R^*(t) \leq R \otimes \gamma (t)
	\end{displaymath}
\end{Definition}

To compute end-to-end delay bounds of individual traffic flows, we need the following Theorem.

\begin{theorem}
	\label{th:blind2flows}
	(Blind Multiplex of two flows) \cite{bouillard2009service} Consider two flows $f_1$,$f_2$ crossing a system $n$ with the strict minimum service $\beta(t)$, and with the flows $f_j$ $\alpha_j$-constrained, $j\in\{1,2\}$. Then, the residual minimum service curve offered to $f_1$ is:
	\begin{displaymath}
	\beta_{1}^{n}(t)= (\beta(t) - \alpha_2(t))_\uparrow
	\end{displaymath}
	
\end{theorem}

Then, to compute the main performance metrics, we need the following results. 

\begin{Corollary}(Left-over service curve - NP-SP Multiplex)\cite{bouillard2009service}
	\label{cor:residual-service-curve}
	Consider a system with the strict service $\beta(t)$ and $m$ flows crossing it, $f_1$,$f_2$,..,$f_m$. The maximum frame size of $f_i$ is $MFS_{i}$, its priority is $p(i)$, $f_i$ is $\alpha_i$-constrained, and $\forall i$, the priorities strictly higher than $p(i)$, are $\forall j$, $p(j)<p(i)$. The flows are scheduled by the NP-SP policy. For each $i \in \{1,..,m \}$,  the strict service curve offered to $f_i$ is given by\footnote{$g_{\uparrow}(t) = \max\{0,\sup_{0 \leqslant s \leqslant t} g(s)\}$}:
	$$\beta_i(t)= \left( \beta(t) - \sum_{\forall j,p(j) <p(i)} \alpha_j(t) - \max_{\forall l, p(l) \geqslant p(i)} MFS_{l}\right) _{\uparrow}$$
\end{Corollary}

\
\begin{theorem} [Performance Bounds]\cite{leboudecthiran12}
	\label{PerformanceBounds}
	Consider a flow $F$ constrained by an arrival curve $\alpha$ crossing a system $\mathcal{S}$ that offers a minimum service curve $\beta$ and a maximum service curve $\gamma$. The performance bounds obtained at any time $t$ are:\\
	Backlog\footnote{v: maximal vertical distance}: $ \forall~t:~q(t)\leq v(\alpha,\beta)$\\
	Delay\footnote{h: maximal horizontal distance}: 	$ \forall~t:~d(t)\leq h(\alpha,\beta)$\\
	Output arrival curve : $\alpha^*(t) =\alpha \oslash\footnote{$f \oslash g(t) = \sup_{s \geq 0}\{f(t+s) - g(s)\}$ } \beta (t)$\\
	Tight Output arrival curve: $\alpha^*(t) =\left( (\gamma\otimes\alpha) \oslash\beta\right)(t)$
\end{theorem}
The computation of these bounds is greatly simplified in the case of i)~leaky bucket arrival curve $\alpha(t) = b + rt$, with $b$ the maximal burst and $r$ the maximum rate, i.e., the flow is $(b,r)$-constrained, and ii)~Rate-Latency service curve $\beta_{R,T}(t) = R \cdot (t-T)^+$ ( $(x)^+$ is the maximum between $x$ and 0), with latency $T$ and rate $R$. In this case, the delay is bounded by $h(\alpha,\beta)=\frac{b}{R} + T$, and the backlog bound is $v(\alpha,\beta)=b + r\cdot T$. Moreover, the output arrival curve is $\alpha^*(t)=b+r (t + T)$.

In the case of a piecewise linear input arrival curve and a piecewise linear minimum service curve, we can compute the delay bound as follows:
\begin{Corollary} [Maximum delay bound under a piecewise arrival curve and piecewise minimum service curve]	\label{OutputArrivalcurve}
	Consider a flow $f$ constrained by a piecewise linear arrival curve $ \alpha$ such as: $\alpha_f(t)=\underset{i}{\min} (\alpha_{r_i,b_i}(t)) $, with  $\alpha_{r_i,b_i}(t)=r_i\cdot t+b_i$, $ i\in[1,n]$ and  minimum service curve such as: $\beta_f(t)=\underset{j}{\max} \left( \beta_{R_j,T_j}(t)\right)$, with  $\beta_{R_j,T_j}(t)=R_j\cdot(t-T_j)^+$. The maximum latency of flow f is: 
	\begin{eqnarray}   	
	delay^{max}_f=\underset{j}{\min}\Big(\frac{y_k}{R_j}+T_j-x_k\Big)\nonumber \textnormal{, with: }k=\min\{i | r_i\leqslant R_j\}
	\end{eqnarray}
	and:
	
	\begin{equation}
	\left\{ \begin{array}{ll}
	x_1=0, y_1=b_1\nonumber\\
	x_k=\frac{b_k-b_{k-1}}{r_{k-1}-r_k},y_k=b_k+r_{k}\cdot x_k\textnormal{, for }2\leqslant k \leqslant n\nonumber\\
	x_{n+1}=+\infty, y_{n+1}=\infty\nonumber\\	
	\end{array} \right.
	\end{equation} 
	
\end{Corollary}
\begin{proof}
	From Theorem~\ref{PerformanceBounds}, we know that the maximum delay bound of flow $f$ is the maximal horizontal distance between $\alpha_f(t)$ and $\beta_f(t)$. Moreover, from Lemma~1 in \cite{boyer2010halfmodeling}, we know that the maximum horizontal distance between $\alpha_f(t)=\underset{i}{\min} (\alpha_{r_i,b_i}(t)) $ and  $\beta_{R_j,T_j}(t)=R_j\cdot(t-T_j)^+$ is:  $$delay^{max}=\frac{y_k}{R_j}+T_j-x_k \textnormal{, with: }k=\min\{i | r_i\leq R_j\}$$
	
	\begin{equation}
	\textnormal{and with: }\left\{ \begin{array}{ll}
	x_1=0, y_1=b_1\nonumber\\
	\alpha_k(x_1)=\alpha_{k+1}(x_1)=y_k\textnormal{, for }1\leqslant k \leqslant n\nonumber\\
	x_{n+1}=+\infty, y_{n+1}=\infty\nonumber\\	
	\end{array} \right.
	\end{equation} 
	
	With a few algebraic considerations, we deduce that $x_k=\frac{b_k-b_{k-1}}{r_{k-1}-r_k},y_k=b_k+r_{k}\cdot\frac{b_k-b_{k-1}}{r_{k-1}-r_k}\textnormal{, for } 2\leqslant k \leqslant n$
	
	Finally, we consider a service curve  $\beta_f(t)=\underset{j}{\max} (\beta_{R_j,T_j}(t)) $.
	From Network Calculus concepts in \cite{leboudecthiran12}, we know that $delay^{max}_f=\inf_{t\geqslant 0} \left\lbrace   (\alpha_{f}\oslash\beta_{f})(-t)\leqslant0  \right\rbrace$. Thus, if we consider a piecewise service curve we have:
	
	\begin{eqnarray}
	(\alpha\oslash\beta)(t)&=&\sup_{s\geqslant0}\left\lbrace \alpha(t+s)-\beta(s) \right\rbrace\nonumber\\
	&=&\sup_{s\geqslant0}\left\lbrace \alpha(t+s)-\underset{j}{\max} (\beta_{R_j,T_j}(s)) \right\rbrace\nonumber\\
	&=&\sup_{s\geqslant0}\left\lbrace \underset{j}{\min} \left[ \alpha(t+s)-\beta_{R_j,T_j}(s)\right]  \right\rbrace\nonumber\\
	&=& \underset{j}{\min} \left[\sup_{s\geqslant0}\left\lbrace \alpha(t+s)-\beta_{R_j,T_j}(s)  \right\rbrace\right]\nonumber
	\end{eqnarray}
	Hence, to compute the maximum delay we can compute the maximum delay for every $\beta_{R_j,T_j}(t)$ and keep the minimum value.
\end{proof}

\begin{theorem} [Concatenation-Pay Bursts Only Once]\cite{leboudecthiran12}
	\label{ConcatenationOfNodes}
	Assume a flow crossing two servers with respective service curves  $\beta_1$ and $\beta_2$. The system consisting of the concatenation of the two servers offers a  service curve $\beta_1 \otimes \beta_2$.
\end{theorem}



\section{System Model}
\label{systModel}

In this section, we present the system model, with first the network model, then the traffic model.

\subsection{Network model}

We consider multi-hop networks, with traffics generated in End-systems and transmitted through one or several switches before reaching the destination end-systems, as illustrated in Fig. \ref{fig:multihopgene}. 

\begin{figure}[htbp]
	\centering
	\includegraphics[width=0.7\textwidth]{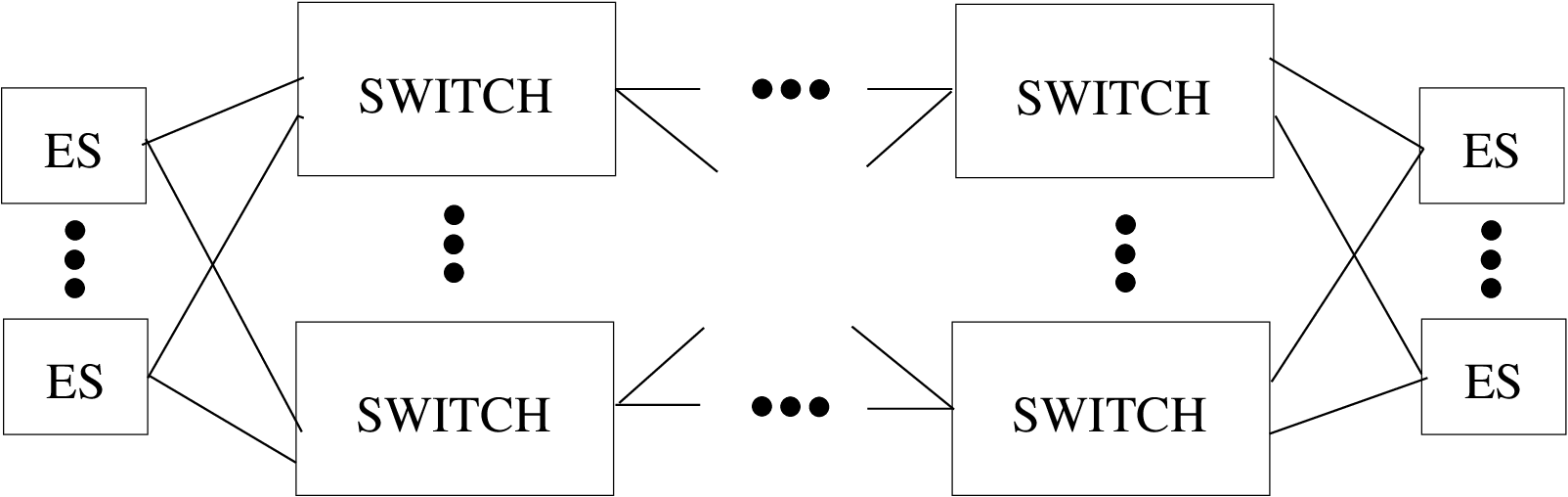}
	\footnotesize \caption{Multi-hop networks}
	\label{fig:multihopgene}
\end{figure}

Hence, in the networks we consider several types of nodes: the end-systems $es$, the output port multiplexers $mux$, composed of  BLS nodes $bls$ and a NP-SP node $sp$, as illustrated in Fig. \ref{fig:outputscheme}.

To assess the performance of the BLS, we use the delay bounds of SCT and RC as a metric, since they both have deadlines contrary to BE. To compute the delays bounds within each node $n \in \{es, bls, sp, mux\}$ we use Th.\ref{PerformanceBounds} under the following assumptions:

(i) leaky-bucket arrival curves for the traffic flows at the input of node $n$, i.e. $\alpha_{k}^n(t)=r_k^n\cdot t+b_k^n$, with $r_k^n$ is the rate and $b_k^n$ is the burst of flow $k$. 	

(ii) the offered service curve by node $n$ to the traffic class $k$ is a rate-latency curve: $\beta_k^n(t)=R_k^n\cdot (t-T_k^n)^+$. 

\begin{figure}[htbp]
	\centering
	\includegraphics[width=0.5\textwidth]{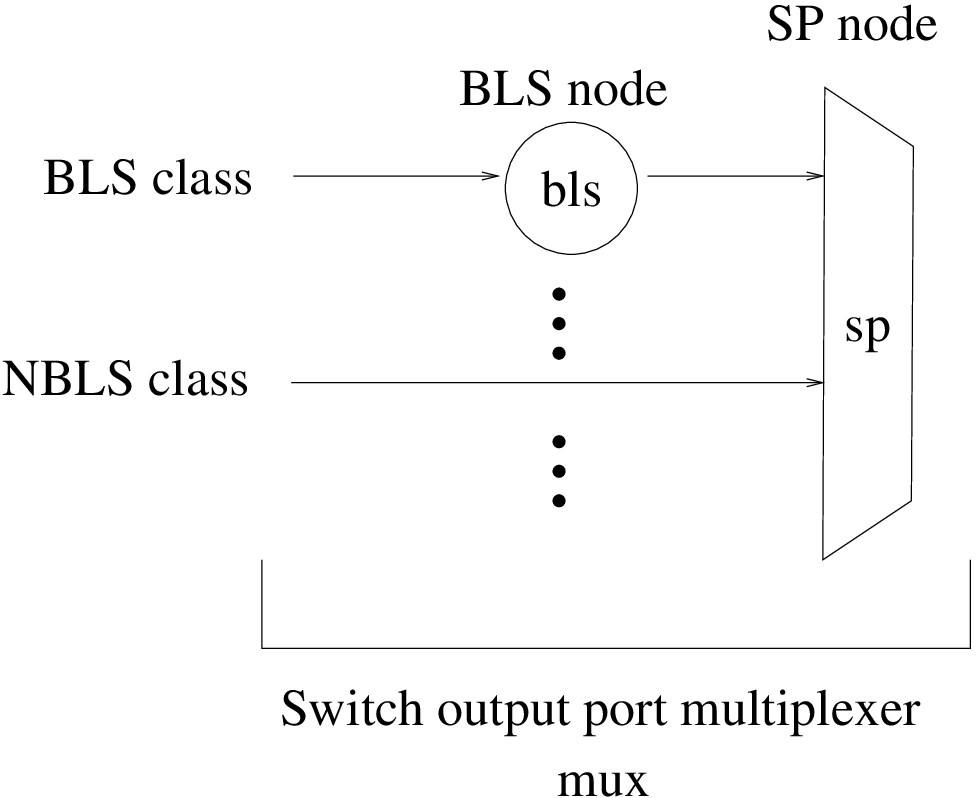}
	\footnotesize \caption{Output port multiplexer node nomenclature}
	\label{fig:outputscheme}
\end{figure}

\hfill\\
\textbf{End-System model}

For the end-systems, they are implementing a NP-SP scheduler. This scheduler has been already modeled in the literature \cite{bouillard2009service} through Corollary~\ref{cor:residual-service-curve}, and the defined strict minimum service curve guaranteed to a traffic class $k\in \{SCT,RC,BE\}$ within an end-system $es$ is as follows:	
$$\beta_{k}^{es}(t) =\left[ C\cdot t- \sum\limits_{ \forall (i,f), f \in i, p(i) < p(k)}\alpha_{i,f}^{es}(t) - \max\limits_{\forall (i,f), f \in i, p(i) \geqslant p(k)} MFS_f \right]_ \uparrow $$

\hfill\\
\textbf{Switch model}

The AFDX standard manages the exchanged data through the Virtual Link (VL) concept. This concept provides a way to reserve a guaranteed bandwidth for each traffic flow. Furthermore, the AFDX supports a NP-SP scheduler based on two priority levels within switches to enable the QoS features.

For the new  extended AFDX, we consider that depending on the constraints of the flows, the different traffics can be separated in several classes: $\{SCT_1, ... SCT_n \}$,$\{RC_1, ... RC_m \}$, and $\{BE_1, ... BE_l \}$, with $n$, $m$, $l$, the number of classes for each type of traffic. Additionally, also depending on the different constraints, any class can be shaped by a BLS at the output port as shown in Fig. \ref{fig:outputport} to reduce the impact of the considered class on lower priorities.

\begin{figure}[htbp]
	\centering
	\includegraphics[width=0.58\textwidth]{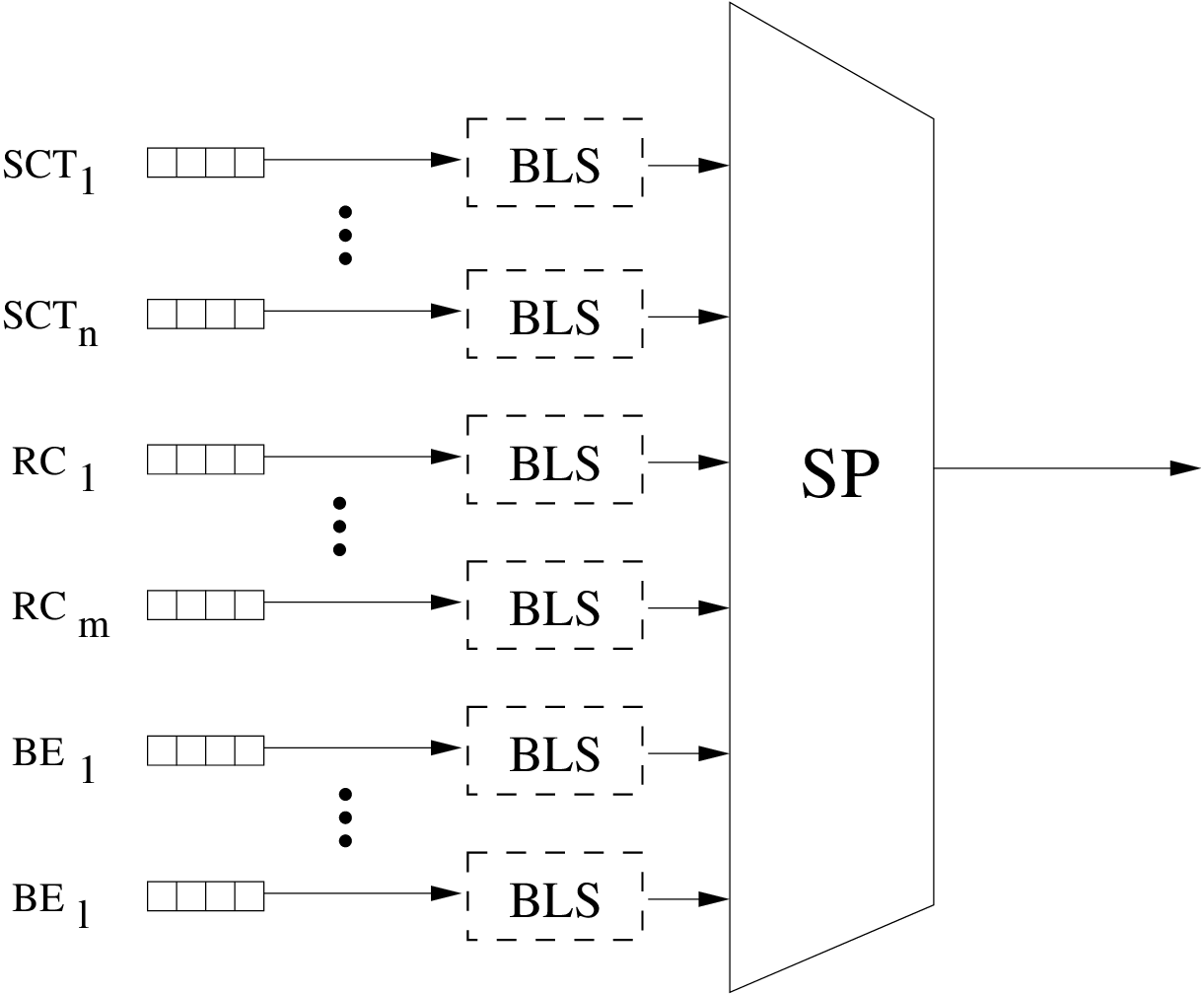}
	\caption{The output port of an extended AFDX switch}
	\label{fig:outputport}
\end{figure}

In Fig. \ref{fig:sw_archifgene}, we illustrate the architecture of our extended AFDX switch in the case of 3 classes. 
It consists of: (i) store and forward input ports to verify each frame correctness before sending it to the corresponding output port; (ii) a static configuration table to forward the received frames to the correct output port(s) based on their VL identifier; (iii) the output ports can handle $k=m+n+l$ priority queues, multiplexed with a NP-SP scheduler, as illustrated in Fig. \ref{fig:outputport}.

As a consequence, for each queue, we associate a class $k$, and we can set two different priorities: $p_L(k)$ and $p_H(k)$. The BLS is only activated if $p_L(k)>p_H(k)$ (because the priority increases when p(k) decreases, i.e., priority 0 is the highest priority). In this case, a credit counter monitoring the SP dequeuing process is attributed to this queue. The credit manages the selection of the priority of the queue viewed by the NP-SP as described in Fig. \ref{fig:outputport2}.

The resulting architecture is very flexible and offers many opportunities to manage each class as needed.
For example for homogeneous classes, we can only consider one queue by type of traffic and we can shape the SCT as proposed in Fig. \ref{fig:BLSshaper}. Or, for more heterogeneous classes, we can use two queues by type of traffic and only shape the second queue of both SCT and RC-type traffic classes, leaving the first ones for tighter deadlines, as illustrated in Fig. \ref{fig:outputport6}.

\begin{figure}[h]
	\centering
	\includegraphics[width=0.68\textwidth]{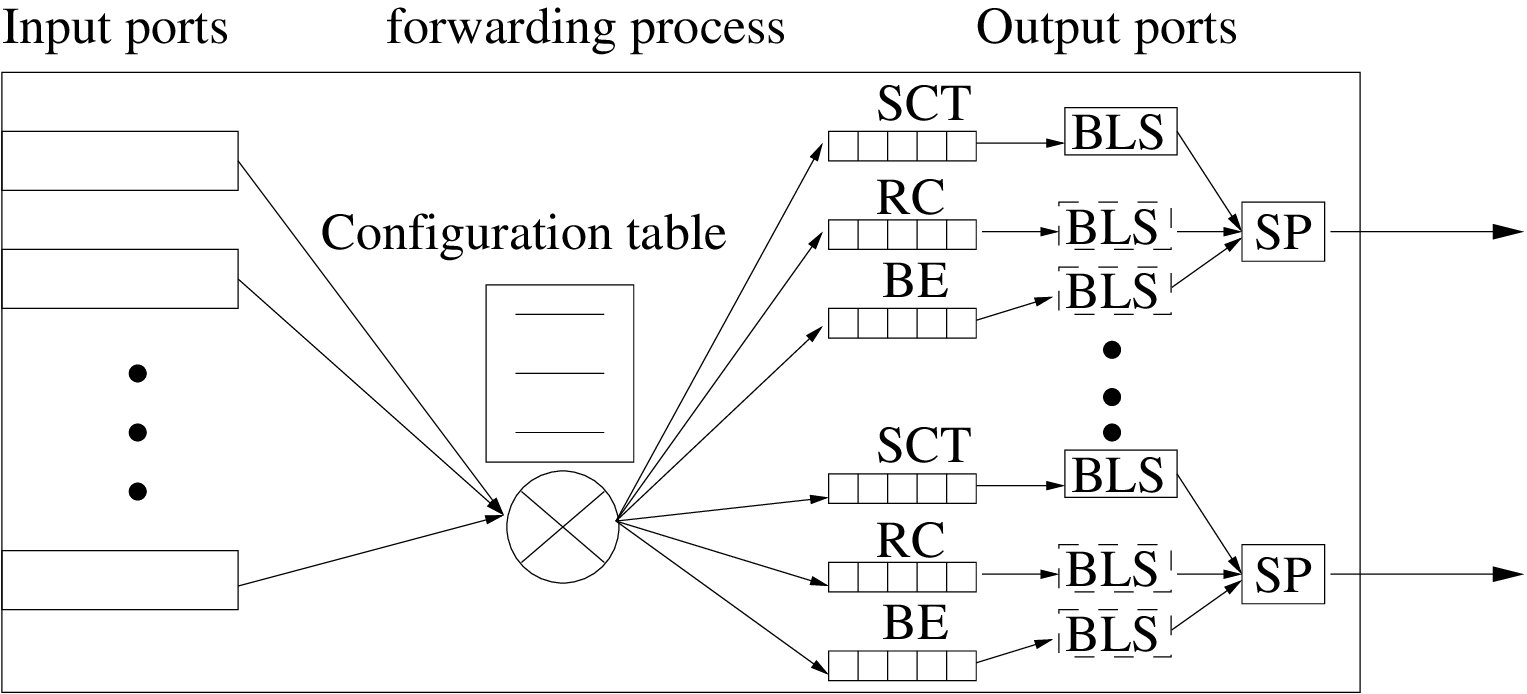}
	\footnotesize \caption{An extended AFDX switch architecture with 3 classes}
	\label{fig:sw_archifgene}
\end{figure}	

\begin{figure}[htbp]
\centering
\includegraphics[width=0.7\textwidth]{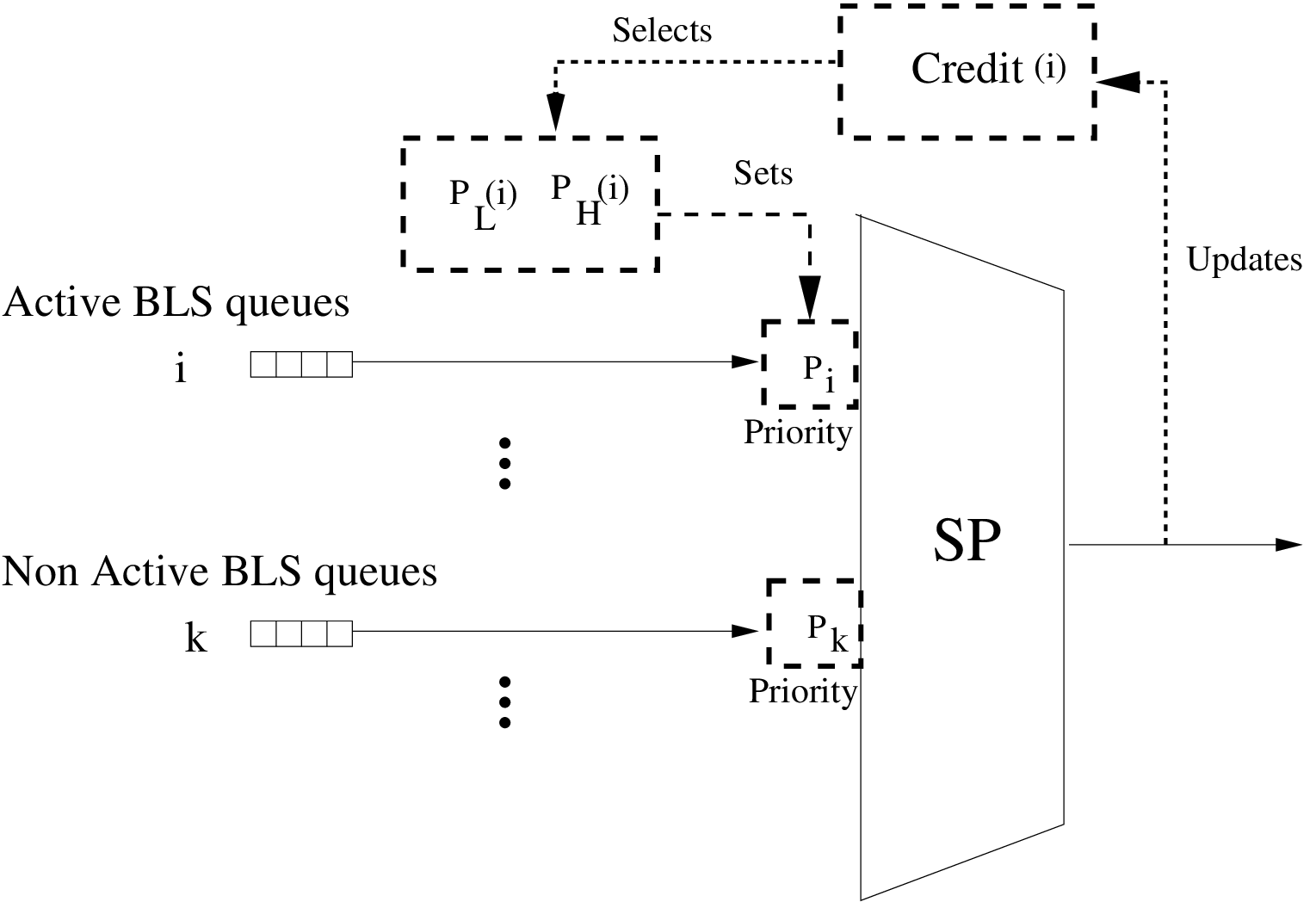}
\caption{BLS behaviour in an output port of an extended AFDX switch}
\label{fig:outputport2}
\end{figure}

\begin{figure}[htbp]
\centering
\includegraphics[width=0.5\textwidth]{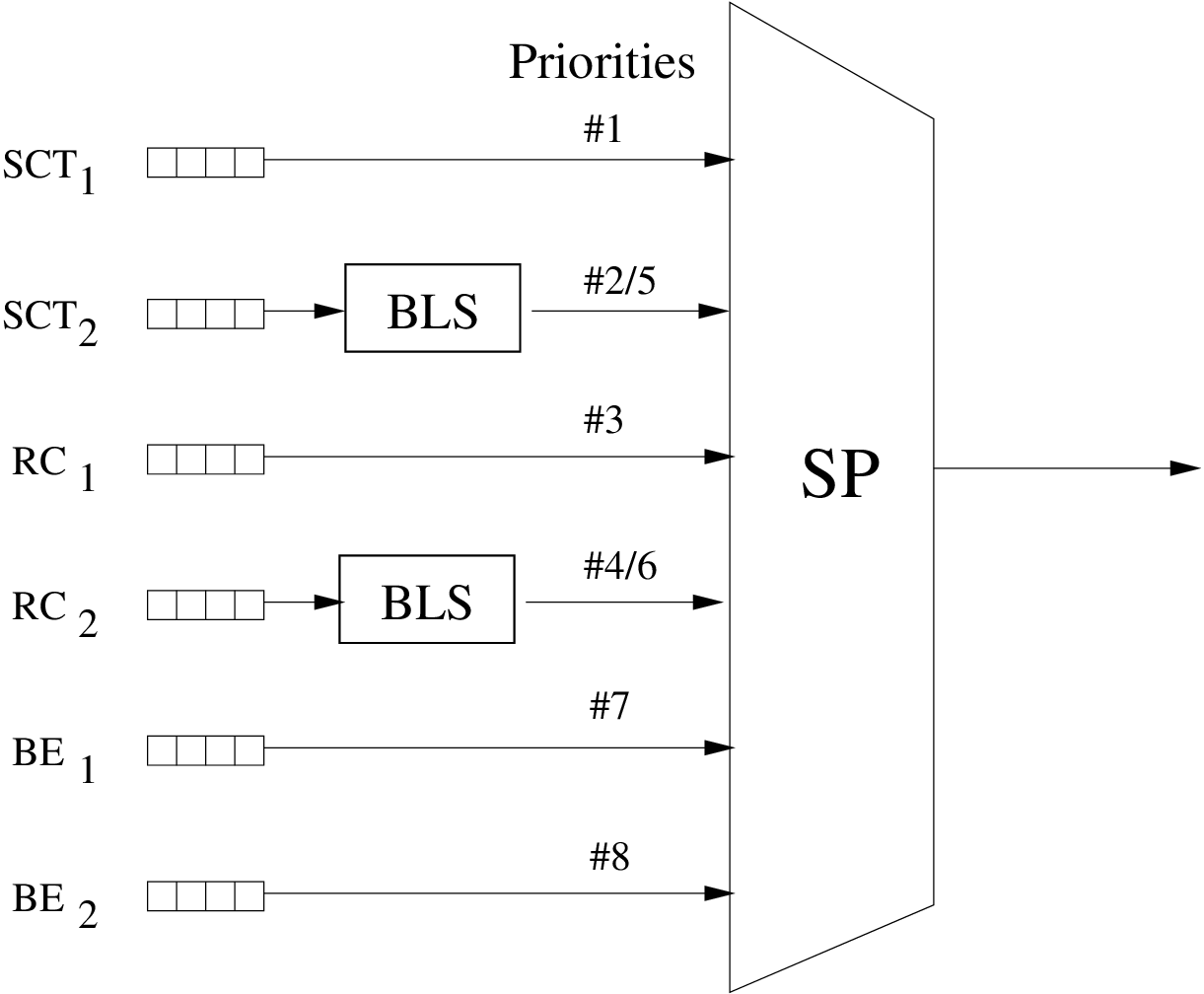}
\caption{Example of an output port of an extended AFDX switch}
\label{fig:outputport6}
\end{figure}

\subsection{Traffic Model}	

To compute upper bounds on end-to-end delays of different traffic classes using Network Calculus, we need to model each message flow to compute its maximum arrival curve.

The arrival curve of each flow $f$ in class $k$ at the input of the node $n \in \{es, sw\}$ or a component $n \in \{bls, sp\}$ along its path is a leaky-bucket curve with a burst $b_{k,f}^{n}$ and a rate $ r_{k,f}^{n}$: $$\alpha_{k,f}^{n}(t) = b_{k,f}^{n} + r_{k,f}^{n} \cdot t$$

Therefore, the arrival curve of the aggregate traffic in class $k$ at the input (resp. output) of the node $n \in \{es, sw\}$ or a component $n \in \{bls, sp\}$ is: $\alpha_{k}^{n}(t) = \sum\limits_{f \in k}\alpha_{k,f}^{n}(t)$ (resp. $\alpha_{k}^{*,n}(t) = \sum\limits_{f \in k}\alpha_{k,f}^{*,n}(t)$ based on Theorem~\ref{PerformanceBounds}).

Each traffic flow $f$ of class $k$, generated by an end-system, is characterized by $\left( BAG_{f}, MFS_{f}, J_{f}\right)$ for respectively the minimum inter-arrival time, the maximum frame size integrating the protocol overhead, and the jitter.

Hence, the arrival curve of traffic class $k$ in the end-system $es$, based on a leaky bucket model, is as follows:
\begin{eqnarray}
\alpha_{k}^{es}(t) &=& \sum\limits_{f \in k}\alpha_{k,f}^{es}(t)= \sum\limits_{f \in k}MFS_{f}+\frac{MFS_{f}}{BAG_{f}}\left( t+J_{f}\right)\nonumber\\
&=&  b_k+r_kt \text{ with }
\begin{cases}
b_k=\sum\limits_{f \in k}MFS_{f}+\frac{MFS_{f} }{BAG_{f}}J_{f}\\
r_k=\sum\limits_{f \in k}\frac{MFS_{f} }{BAG_{f}}
\end{cases}\nonumber
\end{eqnarray}

\section{Worst-case Timing Analysis}
\label{WFA}

In this section, we present the worst-case Timing Analysis done using the Network Calculus framework. First, we present the schedulability condition to identify the needed modelisations. Then, we detail the BLS node modelisation, followed by the output port multiplexer modelisation. Finally, we present the computation of the end-to-end delay bounds and we discuss the nature of the BLS as a shaper.

\subsection{Schedulability Conditions}

\label{scheduTest}
To infer the real-time guarantees of our proposed solution, we need first to define a \textbf{necessary schedulability condition}. This consists in respecting the stability condition within the network, where the sum of maximum arrival rates of the input traffic flows $i$ 
at any crossed node $n$ has to be lower than its minimum guaranteed service rate within the node $n$
. This constraint is denoted as \textbf{\textit{rate constraint}}:
$$\forall \text{node }n\in\text{ network, } \sum_{\forall i\ni n}r_i\leqslant R_n$$

Then, we define a \textbf{sufficient schedulability condition} to infer the traffic schedulability, which consists in comparing the upper bound on end-to-end delay of each traffic flow $f$ of a class $k$ to its deadline, denoted $Deadline_f^{end2end}$. This constraint is called \textbf{\textit{deadline constraint}}:	
$$\forall \text{class }k, \forall \text{flow }f\in k,delay_{k,f}^{end2end}\leqslant Deadline_{k,f}^{end2end}$$

For this sufficient schedulability condition, we detail the end-to-end delay expression of a flow $f$ in the class $k$, $delay_{k,f}^{end2end}$, along its path $path_f$ as follows:
\begin{equation}\label{Dend-2-end}
delay_{k,f}^{end2end}= delay_{k,f}^{es}+ delay_{k,f}^{prop}+\sum\limits_{sw\in path_f}delay_{k,f}^{sw}
\end{equation}

with $delay_{k,f}^{es}$ the delay within the source end-system $es$ to transmit the flow $f$ of class $k$ and $delay_{k,f}^{prop}$ the propagation delay along the path, which is generally negligible in an avionics network. 

The last delay $delay_{k,f}^{sw}$ represents the upper bound of \textbf{the delay within  each intermediate switch} along the flow path, and it consists of several parts as shown in Fig. \ref{fig:sw_archifgene}: 

\begin{itemize}
	\item \textbf{the store and forward delay} at the input port, equal to $\frac{MFS}{C}$, with $MFS$ the length of the frame and $C$ the capacity;
	\item \textbf{the technological latency} due to the forwarding process, upper-bounded by 1$\mu$s in the pre-specification of the AFDX next generation; 		
	\item \textbf{the output port multiplexer delay} due to the BLS and NP-SP scheduler, denoted $delay_{k,f}^{mux}$.
\end{itemize}

 Hence, the only unknown is the delays in the output port multiplexer of the switch. To enable the computation of upper bounds on these delays, we need to model the different parts of the multiplexer, and more particularly the BLS node.
 
\subsection{BLS node model: generalized Continuous Credit-based Approach (gCCbA)}
\label{SC-NC}

In this section, we describe our proposed BLS node model.

We detail here the computation of BLS service curves offered to a BLS class $k$. The main idea is to compute the consumed and gained credits. Knowing that the credit is continuous and always between 0 and $L_M^k$, we use the sum of the consumed and gained credits to compute the minimum and maximum service curves of the BLS node. The main difficulty consists in computing the traffic sent during saturation times, i.e., when the credit is neither gained nor consumed due to the minimum and maximum levels, 0 and $L_M^k$, respectively.

The strict minimum and maximum service curves offered to a BLS class k by a BLS node are defined in Theorem~\ref{Th:k-minGene} and in Theorem~\ref{Th:k-maxGene}, respectively.

\begin{theorem}[Strict Minimum Service Curve offered to a BLS class k by a BLS node] 
	\label{Th:k-minGene}
	Consider a server with a constant rate $C$, implementing BLS shapers. The traffic of class $k$ crosses this server and is shaped by the BLS. Class $k$ has a high priority denoted $p_H(k)$ and a low priority denoted $p_L(k)$ (with $p_L(k)>p_H(k)$).
	We call $HC(k)$ the traffic classes with a priority strictly higher than $p_H(k)$ and MC(k) the classes with a priority between $p_L(k)$ and $p_H(k)$.	 
	The strict minimum service curve guaranteed to the BLS class $k$ is as follows:
	
	\begin{equation*}
	\beta_{k}^{bls}(t) = \left( C-\sum\limits_{h\in HC(k)} r_h -\frac{MFS_{k}^{sat} }{\Delta^{k,\beta}_{inter}}\right) \cdot \frac{I_{idle}^k}{C} \cdot \left( t-\Delta_{idle}^{k,\beta}\right)^+ 
	\end{equation*}		
	
	where 
	\begin{eqnarray}
	MFS_{k}^{sat}&=&\max(\max_{j\in MC(k)}MFS_{j}-\frac{C}{I_{idle}^k}\cdot L_R^k,0)\nonumber\\	
	\Delta^{k,\beta}_{inter}&=&\frac{L_M^k-L_R^{k,min}}{I_{send}^k}+\frac{L_M^k-L_R^k}{I_{idle}^k} +\frac{\max_{j\in MC(k)}MFS_{j}}{C}\nonumber\\		
	L_R^{k,min}&=&\max\left( L_R^k-\frac{\max_{j\in MC(k)}MFS_{j}}{C}\cdot I_{idle}^k,0\right) \nonumber\\	
	\Delta_{idle}^{k,\beta} &=& \frac{L_M^k-L_R^k}{I_{idle}^k} +\frac{\max_{j\in MC(k)}MFS_{j}}{C}\nonumber
	\end{eqnarray}
	
\end{theorem}
\begin{proof} 
	\label{sketch1}
	We present here only a sketch of proof, the complete proof is available in Appendix~\ref{proofGenemin}.  We search a strict minimum service curve defined by a rate-latency curve, i.e.,  $\beta_{k}^{bls}(t)=\rho\cdot (t-\tau)^+$ with rate $\rho$ and latency $\tau$.  

	\begin{figure}[h]
		\centering
		\includegraphics[width=0.5\linewidth]{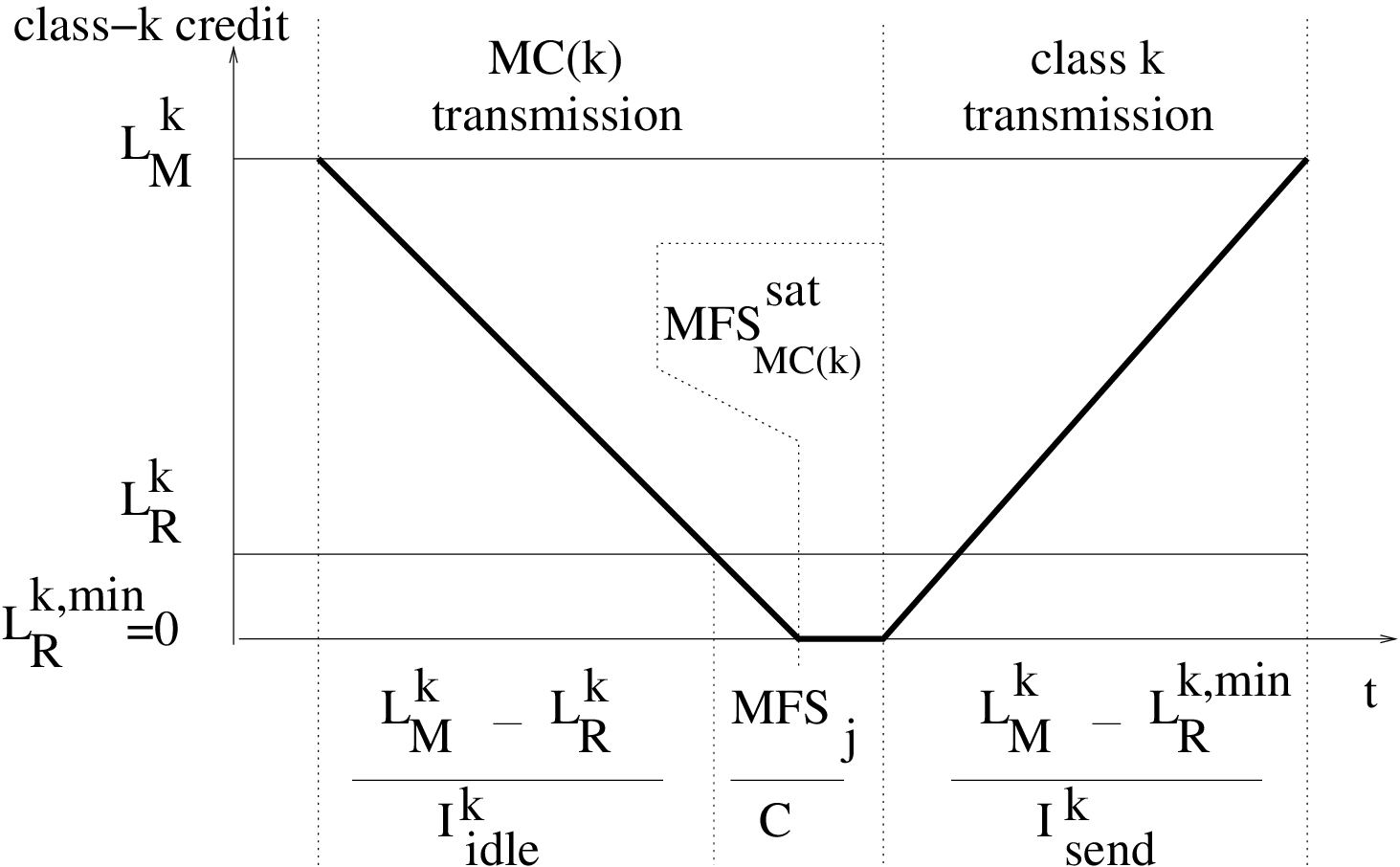}		
		\caption{Computing $\beta_{k}^{bls}(t)$}		
		\label{fig:minComputesketch}
	\end{figure}
	
	First, to compute $\tau$, we consider the maximum latency caused by the BLS.

	Secondly, to compute $\rho$, we consider the fact the credit is a continuous function with values between 0 and $L_M^k$. Consequently, the sum of the gained and consumed credits is upper bounded by $L_M^k$, and the credit can saturate at 0 or $L_M^k$. Thus, the credit consumed during a period $\delta$ is not simply the product of the credit increasing rate (denoted $I_{send}^k$)  and the transmission time of the output traffic (denoted $\Delta R^*(\delta)$). It is actually the product of $I_{send}^k$ and the  output traffic transmitted while the credit is not saturated. The same is true for the gained credit. Hence, the main issue of the proof is the computation of these saturation times. In particular, we compute the maximum saturation for the MC(k) and HC(k) classes, and the minimum saturation for the class $k$, as illustrated in Fig. \ref{fig:minComputesketch}. After this, we use the definition of $\beta(t)$ and the limit of $\frac{\Delta R^*(\delta)}{\delta}$ toward infinity to compute $\rho$.

\end{proof}

\begin{Corollary}[Strict Minimum Service Curve offered to SCT by a BLS node] 
	\label{cor:k-minCbA}
	Consider a server with a constant rate $C$, implementing a BLS shaping the SCT traffic. The SCT, RC and BE traffics cross this server with the following priorities: $p(SCT)\in\{p_H(SCT)=0,p_L(SCT)=2\}$, $p(RC)=1$, $p(BE)=3$, as illustrated in Fig. \ref{fig:sw_archifgene}. 
	
	The strict minimum service curve guaranteed to SCT class is as follows:
	
	\begin{equation*}
	\beta_{SCT}^{bls}(t) = \left( C-\frac{MFS_{SCT}^{sat} }{\Delta^{\beta}_{inter}}\right) \cdot \frac{I_{idle}}{C} \cdot \left( t-\Delta_{idle}^{\beta}\right)^+ 
	\end{equation*}		
	
	where 
	\begin{eqnarray}
	MFS_{SCT}^{sat}&=&\max(\max_{j\in RC}MFS_{j}-\frac{C}{I_{idle}}\cdot L_R,0)\nonumber\\	
	\Delta^{\beta}_{inter}&=&\frac{L_M-L_R^{min}}{I_{send}}+\frac{L_M-L_R}{I_{idle}} +\frac{\max_{j\in RC}MFS_{j}}{C}\nonumber\\		
	L_R^{min}&=&\max\left( L_R-\frac{\max_{j\in RC}MFS_{j}}{C}\cdot I_{idle},0\right) \nonumber\\	
	\Delta_{idle}^{\beta} &=& \frac{L_M-L_R}{I_{idle}} +\frac{\max_{j\in RC}MFS_{j}}{C}\nonumber
	\end{eqnarray}
	
\end{Corollary}
\begin{proof}
	We apply Theorem~\ref{Th:k-minGene} in the particular 3-classes case study presented in Fig. \ref{fig:BLSshaper}, with SCT as class $k$.; thus $HC(k)=\emptyset$ and MC(k)=RC.
\end{proof}

\begin{theorem}[Maximum Service Curve offered to a BLS class k by a BLS node] 
	\label{Th:k-maxGene}
	Consider a server with a constant rate $C$, implementing BLS shapers. The traffic of class $k$ crosses this server and is shaped by the BLS. Class $k$ has a high priority denoted $p_H(k)$ and a low priority denoted $p_L(k)$ (with $p_L(k)>p_H(k)$).
	We call MC(k) the classes with a priority between $p_L(k)$ and $p_H(k)$. 
	
	The maximum service curve offered to the class k traffic is as follows. In the absence of backlogged MC(k) traffic: $\gamma_{k}^{bls} (t) = C\cdot t$;	
	otherwise, during a backlogged period of MC(k):
	\begin{equation*}
	\gamma_{k}^{bls}(t) = \frac{\Delta^{k,\gamma}_{send}}{\Delta^{k,\gamma}_{inter}}\cdot C \cdot t + b^{max}_{k} \cdot\frac{\Delta^{k,\gamma}_{idle}}{\Delta^{k,\gamma}_{inter}}
	\end{equation*}
	
	
	where 
	$$b^{max}_{k}=\frac{C}{I_{send}^k}\cdot L_M^k+MFS_{k}$$
	$$\Delta^{k,\gamma}_{send}= \frac{MFS_{k}}{C}+\frac{L_M^k-L_R^k}{I_{send}^k}$$		
	$$\Delta^{k,\gamma}_{idle}=\frac{L_M^k-L_R^k}{I_{idle}^k}$$		
	and 
	$$\Delta^{k,\gamma}_{inter}=\Delta^{k,\gamma}_{send}+\Delta^{k,\gamma}_{idle}$$		
	
\end{theorem}
\begin{proof} 
	
	\label{sketch2}
	We present here only a sketch of proof, the complete proof is available in Appendix~\ref{proofGenemax}. We search a maximum service curve defined by a leaky-bucket curve, i.e.,  $\gamma_{k}^{bls}(t)=r\cdot t+b$ with rate $r$ and burst $b$.   First, for $r$ we use the fact that the sum of the gained and consumed credits is lower bounded by $-L_M^k$. Then, we calculate the bounds of saturation times during a period $\delta$. In particular, we calculate the minimum and maximum saturations, as illustrated in Fig. \ref{fig:maxComputesketch}. Finally, we use the definition of $\gamma(t)$ and the limit toward infinity of $\frac{\Delta R^*(\delta)}{\delta}$ to compute $r$.
	
	\begin{figure}[h]
		\centering
		\includegraphics[width=0.6\linewidth]{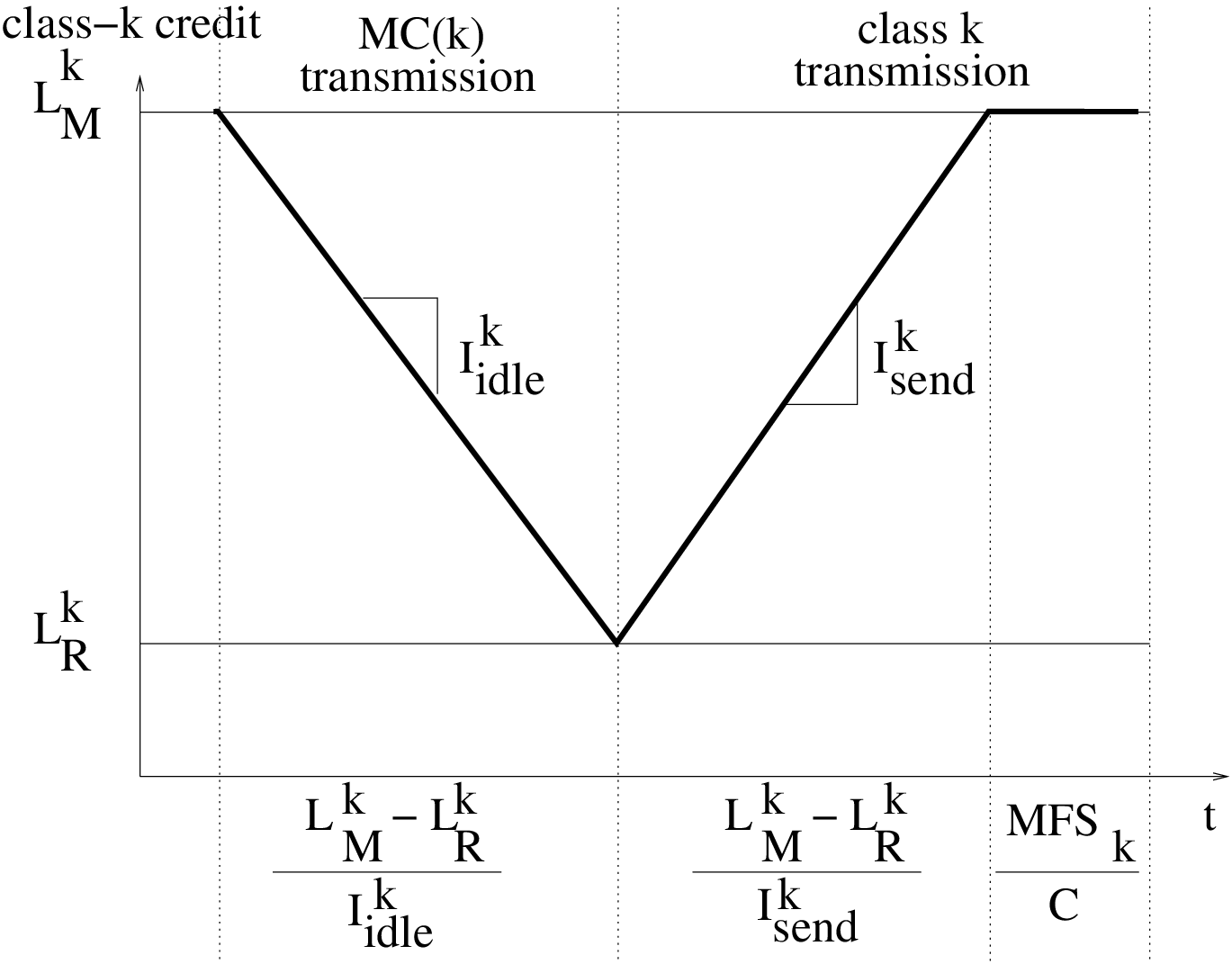}		
		\caption{Computing $\gamma_{k}^{bls}(t)$}		
		\label{fig:maxComputesketch}
	\end{figure}	
	
\end{proof}
\begin{Corollary}[Maximum Service Curve offered to SCT by a BLS node] 
	\label{cor:k-maxCCbA}
	Consider a server with a constant rate $C$, implementing a BLS shaping the SCT traffic. The SCT, RC and BE traffics cross this server with the following priorities: $p(SCT)\in\{p_H(SCT)=0,p_L(SCT)=2\}$, $p(RC)=1$, $p(BE)=3$, as illustrated in Fig. \ref{fig:sw_archifgene}. 
	
	The maximum service curve offered to the SCT traffic class is as follows. In the absence of backlogged RC traffic: $\gamma_{SCT}^{bls} (t) = C\cdot t$;	
	otherwise, during a backlogged period of RC:
	\begin{equation*}
	\gamma_{SCT}^{bls}(t) = \frac{\Delta^{\gamma}_{send}}{\Delta^{\gamma}_{inter}}\cdot C \cdot t + b^{max}_{SCT} \cdot\frac{\Delta^{\gamma}_{idle}}{\Delta^{\gamma}_{inter}}
	\end{equation*}
	
	
	where 
	$$b^{max}_{SCT}=\frac{C}{I_{send}}\cdot L_M+MFS_{SCT}$$
	$$\Delta^{\gamma}_{send}= \frac{MFS_{SCT}}{C}+\frac{L_M-L_R}{I_{send}}$$		
	$$\Delta^{\gamma}_{idle}=\frac{L_M-L_R}{I_{idle}}$$		
	and 
	$$\Delta^{\gamma}_{inter}=\Delta^{\gamma}_{send}+\Delta^{\gamma}_{idle}$$		
	
\end{Corollary}
\begin{proof}
	We apply Theorem~\ref{Th:k-maxGene} in the particular 3-classes case study presented in Fig. \ref{fig:BLSshaper}, with SCT as class $k$; thus $HC(k)=\emptyset$ and MC(k)=RC.
\end{proof}


The maximum output arrival curve of a BLS class $k$ is detailed in the following Corollary:

\begin{Corollary} [Maximum Output Arrival Curve of a BLS class]
	\label{cor:SCT-max-arrival}
	Consider a BLS class $k$ with a maximum leaky-bucket arrival curve $\alpha$ at the input of a BLS shaper, guaranteeing a minimum rate-latency service curve $\beta_{k}^{bls}$ and a maximum service curve $\gamma_{k}^{bls}$. The maximum output arrival curve is:
	\begin{equation*}
	\alpha^{*,bls}_{k} (t)=\min(\gamma_{k}^{bls}(t),\alpha \oslash \beta_{k}^{bls}(t) )
	\end{equation*}
\end{Corollary}

\begin{proof}
	To prove Corollary~\ref{cor:SCT-max-arrival}, we generalize herein the rule 13 in p. 123 in \cite{leboudecthiran12}, i.e., $(f\otimes g)\oslash g \leq f \otimes (g\oslash g) $, to the case of three functions $f$, $g$ and $h$ when $g\oslash h \in \mathcal{F}$, where $\mathcal{F}$ is the set of non negative and wide sense increasing functions:
	\begin{displaymath}
	\mathcal{F}=\{ f: \mathbb{R}^+ \rightarrow \mathbb{R}^+ \mid f(0)=0, \forall t \geq s: f(t) \geq f(s) \}
	\end{displaymath}
	According to Theorem~\ref{PerformanceBounds}, we have $\alpha^*(t)= (\gamma_{k}^{bls}\otimes\alpha)\oslash \beta_{k}^{bls}$. Moreover, in the particular case of a leaky-bucket arrival curve $\alpha$ and a rate-latency service curve $\beta_{k}^{bls}$, $\alpha\oslash \beta_{k}^{bls}$ is a leaky-bucket curve, which is in $\mathcal{F}$. Hence, we have the necessary condition to prove the following:		
	$$(\alpha\otimes\gamma)\oslash\beta(t)\leq  \gamma\otimes(\alpha\oslash\beta)(t)  \leq \min(\gamma(t),\alpha\oslash\beta(t))   $$
	\begin{eqnarray}
	&&(\alpha\otimes\gamma)\oslash\beta(t) \nonumber\\
	&& =\sup_{u\geq0}\left\lbrace (\gamma\otimes\alpha)(t+u)-\beta(u)\right\rbrace \nonumber\\
	&&=\sup_{u\geq0}\left\lbrace \inf_{-u\leq s'\leq t}\left\lbrace \gamma(t-s')+\alpha(s'+u)-\beta(u) \right\rbrace\right\rbrace \nonumber\\
	&& \leq \sup_{u\geq0}\left\lbrace \inf_{0\leq s'\leq t}\left\lbrace \gamma(t-s')+\alpha(s'+u)-\beta(u) \right\rbrace\right\rbrace \nonumber\\
	&& \leq \sup_{u\geq0}\left\lbrace \inf_{0\leq s'\leq t}\left\lbrace \gamma(t-s')+\sup_{v\geq0}\left\lbrace\alpha(s'+v)-\beta(v) \right\rbrace\right\rbrace\right\rbrace \nonumber\\
	&& = \gamma\otimes(\alpha\oslash\beta)(t)  \leq \min(\gamma(t),\alpha\oslash\beta(t)) \nonumber
	\end{eqnarray}
\end{proof}

\subsection{Output port multiplexer service curves}
\label{multiplexermodel}

In this section, we compute the strict minimum service curves offered by a switch output port multiplexer $mux$. Such a multiplexer $mux$ consists of $sp$ and $bls$ nodes as illustrated in Fig. \ref{fig:outputscheme}.

The strict minimum service curve offered by each output port multiplexer to a $BLS$ class  is defined in Theorem~\ref{Th:blssp}, and the strict minimum service curves offered by each output port multiplexer to a $NBLS$ class, i.e., not shaped by a BLS, is defined in Theorem~\ref{Th:nblssp}. 


\begin{theorem}[Strict Minimum Service Curve offered by an output port multiplexer to a BLS class] 
	\label{Th:blssp}
	Consider a system implementing a BLS with the strict minimum service $\beta$. 
	
	The strict minimum service curve offered to a BLS class $k$ by an  output port multiplexer is:		
	\begin{equation*}
	\label{betaRateLatency}
	\beta_{k\in BLS}^{mux}(t) = \max\left( \beta_{k\in BLS,p_L(k)}^{sp},\beta_{k\in BLS}^{bls}\otimes\beta_{k\in BLS,p_H(k)}^{sp}\right)  (t)
	\end{equation*}	
	
	with:
	\begin{itemize}
		\item $\beta_{k\in BLS,p_L(k)}^{sp}(t)=(\beta(t) - \sum_{j\in MC(k)\cup HC(k)} \alpha_j^{sp}(t) - \max_{j\in LC(k)\cup k}MFS_{j})_{\uparrow}$ the strict minimum service curve offered by the NP-SP when the class-k BLS priority is low;\\	
		\item $\beta_{k\in BLS}^{bls}(t)$ the strict minimum service curve offered by the BLS node to class k, defined in Theorem~\ref{Th:k-minGene};\\
		\item $\beta_{k\in BLS,p_H(k)}^{sp}(t)=(\beta(t) - \sum_{j\in HC(k)} \alpha_j^{sp}(t) - \max_{j\notin HC(k)}MFS_{j})_{\uparrow}$ the strict minimum service curve offered by the NP-SP when the class-k BLS priority is high;\\
		\item $\alpha_j^{sp}(t)$ the input arrival curve of flow $j$ at the $sp$ node, such as: \[\begin{cases}
		\text{if } j\in NBLS\text{, }\alpha_j^{sp}(t)=\alpha_j(t) \\
		\text{if } j\in BLS\text{, } \alpha_j^{sp}(t)=\alpha^{*,bls}_j(t)\\
		\quad  \quad  \quad  \quad  \quad \quad \quad \text{ }
		=\min(\gamma_j^{bls},\alpha_j^{bls}\oslash\beta_{j}^{bls})(t), \text{defined in Corollary~\ref{cor:SCT-max-arrival}}
		\end{cases}
		\]		
	\end{itemize}		
\end{theorem}
\begin{proof}		
	The idea is to model the impact of a BLS implemented on top of the NP-SP scheduler on $BLS$ class $k$. To achieve this aim, we distinguish two possible scenarios. The first one covers the particular case where the class-$k$ priority  remains low, i.e., the other queues are empty; whereas the second one covers the general case where the priority of class $k$ oscillates between $p_L(k)$ and $p_H(k)$, as explained in Section~\ref{basicconcepts}. Firstly, the minimum service curve guaranteed within $mux$ in the first scenario is due to the NP-SP scheduler and denoted $\beta^{sp}_{k\in BLS,p_L(k)}$, which is computed via Corollary~\ref{cor:residual-service-curve} when considering the class-k priority is $p_L(k)$. Secondly, the minimum service curve guaranteed within $mux$ in the second scenario is computed via Theorem~\ref{ConcatenationOfNodes}, through the concatenation of the service curves within the BLS node $\beta_{k\in BLS}^{bls}$ (computed in Theorem~\ref{Th:k-minGene}) and the NP-SP node $\beta_{k\in BLS,p_H(k)}^{sp}$ (computed via Corollary~\ref{cor:residual-service-curve} when class-$k$ priority is high).		
\end{proof}

\begin{Corollary}[Strict Minimum Service Curve offered by an output port multiplexer to SCT in the case of three traffic classes] 
	\label{cor:SCTblssp}
	Consider a server with a constant rate $C$, implementing a BLS shaping the SCT traffic. The SCT, RC and BE traffics cross this server with the following priorities: $p(SCT)\in\{p_H(SCT)=0,p_L(SCT)=2\}$, $p(RC)=1$, $p(BE)=3$, as illustrated in Fig. \ref{fig:sw_archifgene}. 	
	
	The strict minimum service curve offered to SCT by an output port multiplexer is:		
	\begin{equation*}
	\label{betaRateLatency}
	\beta_{SCT}^{mux}(t) = \max\left( \beta_{SCT,2}^{sp},\beta_{SCT}^{bls}\otimes\beta_{SCT,0}^{sp}\right)  (t)
	\end{equation*}	
	
	with:
	\begin{itemize}
		\item $\beta_{SCT,2}^{sp}(t)=(C\cdot t -\alpha_{RC}(t) - \max_{j\in BE\cup SCT}MFS_{j})_{\uparrow}$ the strict minimum service curve offered by the NP-SP when the class-k BLS priority is low;\\	
		\item $\beta_{SCT}^{bls}(t)$ the strict minimum service curve offered by the BLS node to SCT, defined in Corollary~\ref{cor:k-minCbA}
		\item $\beta_{SCT,0}^{sp}(t)=(C\cdot t - \max_{j\in\{SCT,RC,BE\}} MFS_{j})_{\uparrow}$ the strict minimum service curve offered by the NP-SP when the BLS priority is high;		
	\end{itemize}		
\end{Corollary}
\begin{proof}
	We apply Theorem~\ref{Th:blssp} in the particular 3-classes case study presented in Fig. \ref{fig:BLSshaper}, with SCT as class $k$. Thus $HC(k)=\emptyset$, MC(k)=RC, and $\beta(t)=C\cdot t$.
\end{proof}		
\
\begin{theorem}[Strict Minimum Service Curve offered to a NBLS class by an output port multiplexer] 
	\label{Th:nblssp}
	Consider a system implementing a BLS with the strict service $\beta$ and $m$ flows crossing it, $f_1$,$f_2$,..,$f_m$.		
	The strict minimum service curve offered to a NBLS class $k$ by an output port multiplexer is:
	$$\beta_{k\in NBLS}^{mux}(t)=\max\left( \beta_{k\in NBLS}^{sp},\beta_{k\in NBLS}^{bls}\right) (t) $$
	
	
	with:
	\begin{itemize}
		\item $\beta_{k\in NBLS}^{sp}=(\beta - \sum_{p_H(j)<p(k),j\in BLS} \alpha_j\oslash\beta_{j}^{bls} - \sum_{p(j)<p(k), j\in NBLS} \alpha_j- \max_{p(j) \geqslant p(k)}MFS_j)_{\uparrow}$;\newline	
		\item $\beta_{j}^{bls}$, with $j\in BLS$, the strict minimum service curve offered by the BLS node to class $j$, defined in Theorem~\ref{Th:k-minGene};
		\item $\beta_{k\in NBLS}^{bls}=(\beta - \sum_{p_H(j)<p(k),j\in BLS} \gamma_{j}^{bls} - \sum_{p(j) <p(k), j\in NBLS} \alpha_j\\
	\quad  \quad  \quad  \quad  \quad \quad \quad \text{ }
	\quad  \quad  \quad  \quad   \quad
	- \max_{p(j) \geqslant p(k)} MFS_j)_{\uparrow}$;\newline			
		\item $\gamma_{j}^{bls}$, with $j\in BLS$, the maximum service curve offered by the BLS node to class $j$, defined in Theorem~\ref{Th:k-maxGene}.
	\end{itemize}	
	
	
\end{theorem}

\begin{proof}
	The proof of Theorem~\ref{Th:nblssp} is straightforward.
	Theorem~\ref{Th:nblssp} is obtained through replacing within the equation of Corollary~\ref{cor:residual-service-curve} the arrival curve of higher priority traffic than class $k\in NBLS$ by the curves computed in Corollary~\ref{cor:SCT-max-arrival}.		
\end{proof}

\begin{Corollary}[Strict Minimum Service Curve offered to RC by an output port multiplexer in the case of three traffic classes] 
	\label{cor:RCblssp}
	Consider a server with a constant rate $C$, implementing a BLS shaping the SCT traffic. The SCT, RC and BE traffics cross this server with the following priorities: $p(SCT)\in\{p_H(SCT)=0,p_L(SCT)=2\}$, $p(RC)=1$, $p(BE)=3$, as illustrated in Fig. \ref{fig:sw_archifgene}. 
	
	The strict minimum service curve offered to RC by an output port multiplexer is:
	$$\beta_{RC}^{mux}(t)=\max\left( \beta_{RC}^{sp},\beta_{RC}^{bls}\right) (t) $$
	
	
	with:
	\begin{itemize}
		\item $\beta_{RC}^{sp}(t)=(C\cdot t - \alpha_{SCT}(t)\oslash\beta_{SCT}^{bls}(t) - \max_{j\in\{SCT,RC,BE\}}MFS_j)_{\uparrow}$;\newline	
		\item $\beta_{SCT}^{bls}(t)$ the strict minimum service curve offered by the BLS node to SCT, defined in Corollary~\ref{cor:k-minCbA};
		\item $\beta_{RC}^{bls}(t)=(C\cdot t -  \gamma_{SCT}^{bls}(t) - \max_{j\in\{SCT,RC,BE\}} MFS_j)_{\uparrow}$;\newline			
		\item $\gamma_{SCT}^{bls}(t)$ the maximum service curve offered by the BLS node to SCT, defined in Corollary~\ref{cor:k-maxCCbA}.
	\end{itemize}	
\end{Corollary}	
\begin{proof}
	We apply Theorem~\ref{Th:nblssp} to the particular case of the 3-classes case study presented in Fig. \ref{fig:BLSshaper}, with SCT as class $k$ and RC as MC(k).
\end{proof}

\subsection{Computing End-to-End Delay Bounds}

With the BLS  node and the output port modeled, we are now able to compute the end-to-end delay bounds. The computation of the end-to-end delay upper bounds follows four main steps:

\begin{enumerate}
	\item computing the strict minimum service curve guaranteed to each traffic class $k$ in each node $n$ of type $\{es,mux\}$, $\beta_{k}^{n}$, will infer the computation of the residual service curve, guaranteed to each individual flow $f$ of class $k$, $\beta_{k,f}^{n}$ with Theorem~\ref{th:blind2flows};
	\item knowing the residual service curve guaranteed to each flow within each crossed node allows the propagation of the arrival curves along the flow path, using Theorem~\ref{PerformanceBounds}. We can compute the output arrival curve, based on the input arrival curve and the minimum service curve, which will be in its turn the input of the next node;
	\item the computation of the minimum end-to-end service curve of each flow $f$ in class $k$, based on Theorem~\ref{ConcatenationOfNodes}, is simply the concatenation of the residual service curves, $\beta_{k,f}^{n}$, $\forall n$ along its path $path_f$;
	\item given the minimum end-to-end service curve of each flow $f$ in class $k$ along its $path_f$ and its maximum arrival curve at the initial source, the end-to-end delay upper bound $delay_{k,f}^{end2end}$ is the maximum horizontal distance between the two curves using Theorem~\ref{PerformanceBounds} and Corollary~\ref{OutputArrivalcurve}.
\end{enumerate}

Now that we have modeled the proposed network, we use this model to  answer the question whether "shaper" is the correct qualifier for the BLS.

\subsection{Discussion: is the BLS really a shaper?}

\label{isshaper}

The most common kind of shapers is the greedy shaper, which has been detailed in \cite{leboudecthiran12}.
According to \cite{leboudecthiran12}, a $shaper$ with a shaping curve $\sigma$ is a bit processing device  that forces its output to have $\sigma$ as an output arrival curve. A $greedy$ $shaper$ is a shaper that delays the input bits in a buffer, whenever sending a bit would violate the constraint $\sigma$, but outputs them as soon as possible. A consequence of this definition is that, for an input flow $R$, the output flow $R^*$ is defined by $R^*=R\otimes\sigma$. Moreover, as the service curve $\beta$ and maximum service curve $\gamma$ are defined by $R^*\geqslant R\otimes\beta$ and $R^*\leqslant R\otimes\gamma$, this means that $\sigma=\beta=\gamma$ in the case of a greedy shaper. Obviously, this is not the case for the BLS. Another property of the greedy shaper is that the difference between the fluid model and the packetized model is bounded by the maximum sized packet. 

From the definition of the BLS gCCbA model in Section~\ref{SC-NC}, we can easily compute the corresponding fluid (bit-per-bit) gCCbA model: we do not consider an additional frame due to non-preemption. As a consequence, the defined saturation times are null.

So, when considering the 3-classes case study presented in Fig. \ref{fig:BLSshaper}, we have: 	
$$\gamma_{SCT}^{bls, fluid}(t)=I_{idle}\cdot t+L_M$$

$$\beta_{SCT}^{bls, fluid}(t)=I_{idle}\cdot\left( t-\frac{L_M-L_R}{I_{idle}}\right) ^+$$

This shows that the difference between the fluid and packetized models, i.e., $\beta_{k}^{bls,fluid}$ and $\beta_{k}^{bls}$, is larger than a single maximum sized frame: an additional MC(k) frame is considered in every $\Delta_{inter}^{k,\beta}$, and an  additional  frame of class k is considered in each $\Delta_{inter}^{k,\gamma}$.

Finally, the BLS functioning itself shows that the BLS is not a greedy shaper: if a frame is enqueued and there is no higher priority frame enqueued, then the frame is dequeued no matter the state of the BLS credit. Hence, the BLS is non-blocking contrary to the definition of a greedy shaper. Moreover, if no higher priority traffic is present, then the BLS does not force the output to conform to a certain $\sigma$, unlike a shaper.

So, if the BLS is not a shaper, what is its nature? The BLS changes the priority of a queue through reordering the priority of the different queues, and it cannot be used without a Static Priority Scheduler. So trying to characterize it on its own is futile. Together with the NP-SP however, they are able to reorganize the output traffic according to the BLS parameters. Because of this, BLS+SP is much closer to schedulers such as Deficit Round Robin (DRR) than shapers.  

\section{Performance Analysis}	
\label{PA}

In this section, we start with a first 3-classes single-hop use-case to evaluate the tightness and sensitivity of our model, in reference to Achievable Worst-Cases (AWCs) described in Appendix~\ref{AWCs}. Next, we compare the CPA and NC models (WbA and CCbA) under different scenarios. In a second use-case, we consider an multiple-BLS multi-hop architecture to highlight the delay bound reductions and schedulability increases. We finish with a third used-case to add the A350 flight control traffic to the AFDX.

\subsection{Use-case 1: sensitivity, tightness and comparison with CPA}
\label{Use-case1}
In this section, we start by presenting a case study.
Then, we analyze our model by assessing, first the impact of the BLS parameters and the utilization rates of SCT and RC on the delay bounds, then its tightness in reference to AWC\footnote{Since there is no strict order between the two achievable worst-cases (see Appendix~\ref{AWCs}), we will use the maximum value, denoted AWC=max(AWC-1, AWC-2), as a reference to assess our model tightness.}.

\newpage
\noindent
\textbf{Case study}

Our first case study is based on a single-hop Gigabit network described in Fig. \ref{fig:singlehoparchi2}, with the 3-class output port presented in Fig. \ref{fig:BLSshaper2} and the traffic profiles presented in Table~\ref{table:alltrafficprofiles}.

	\begin{figure}[h]
		\centering
			\includegraphics[width=0.6\columnwidth]{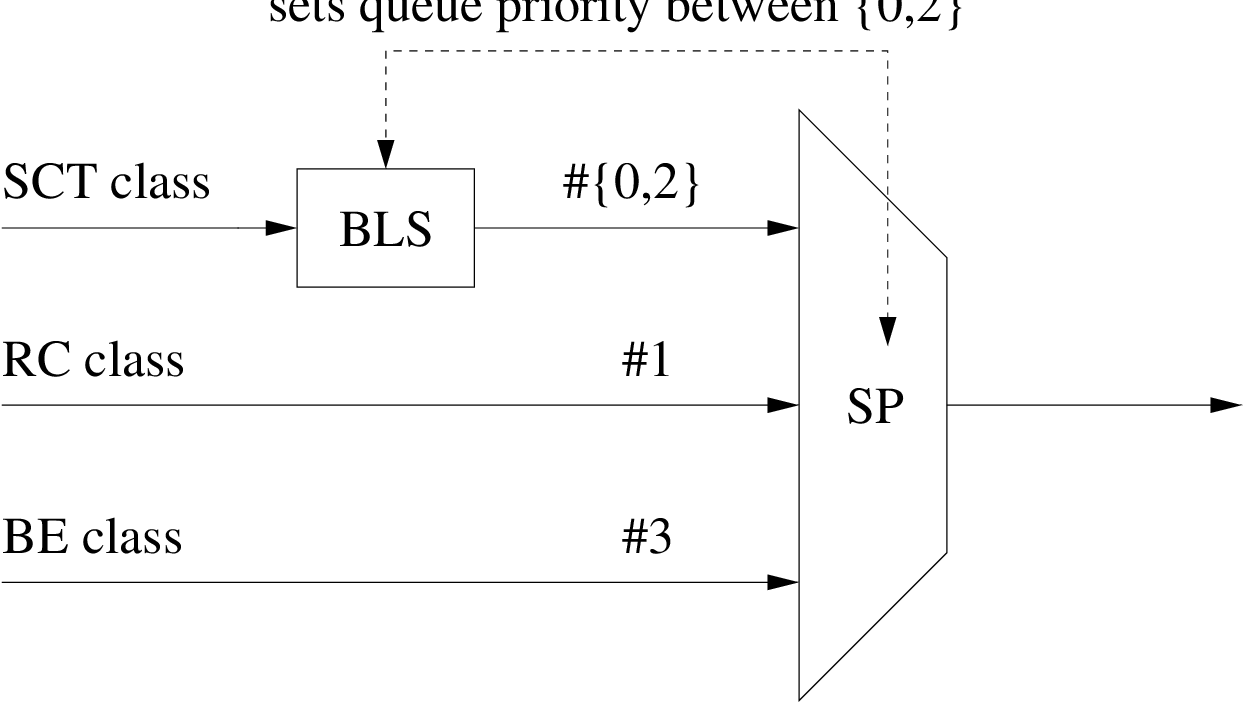}
		\caption{Burst Limiting Shaper on top of NP-SP at the output port with 3 classes }
		\label{fig:BLSshaper2}
	\end{figure}

\begin{table}[h!]
	\footnotesize
	\centering
	\begin{tabular}{|c|c|c|c|c|c|}
		\hline
		Priority & Traffic type & MFS & BAG & deadline& jitter  \\
		&   & (Bytes)          & (ms)&(ms) &(ms)\\
		\hline
		0/2 & SCT & 64 & 2 & 2 & 0 \\  
		\hline
		1 &  RC & 320 & 2 & 2 & 0\\
		\hline
		3 &	BE & 1024 & 8 & none & 0.5 \\ 
		\hline
	\end{tabular}
	\footnotesize \caption{Avionics flow Characteristics}
	\label{table:alltrafficprofiles}
\end{table}

As there is only one shaped class: SCT, we use  $k=\emptyset$ to simplify, for the non-ambiguous notations, such as $L_M^k$ or $I_{idle}^k$.

\begin{figure}[htbp]
	\centering
	\includegraphics[width=0.95\textwidth]{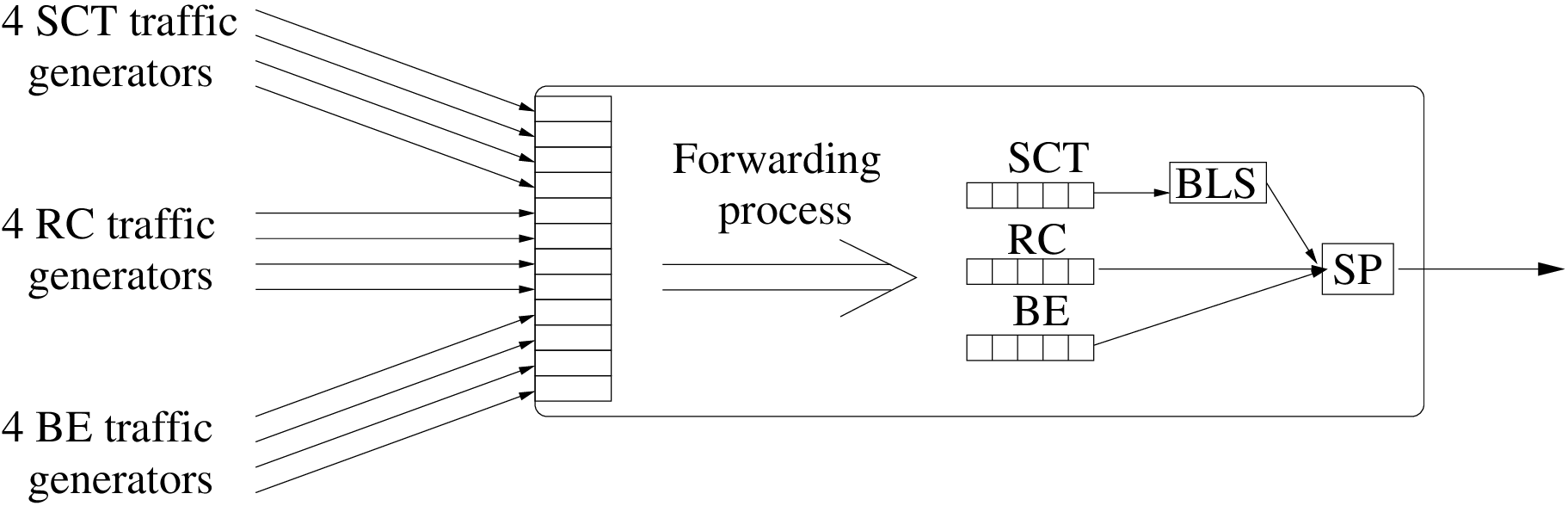}
	\footnotesize \caption{Considered extended AFDX network}
	\label{fig:singlehoparchi2}
\end{figure}

To evaluate our model, we conduct tightness and sensitivity analyses using different scenarios when varying the input rates of SCT and RC and the BLS parameters.
 	
The five scenarios are described by the following vectors:
\begin{eqnarray}
Scenario_{SCT}= &&( UR_{SCT}\in\left[0.1:78\right], UR_{RC}=20,L_M=22~118,L_R=0,\nonumber\\\nonumber &&BW=0.46)\\
Scenario_{RC}= &&(UR_{SCT}=20, UR_{RC}\in\left[0.5:72\right],L_M=22~118,L_R=0,\nonumber\\\nonumber && BW=0.46)\\
Scenario_{LM}=&&(UR_{SCT}=20, UR_{RC}=20,L_M\in\left[1382.4..216~830\right],\nonumber\\\nonumber && L_R=1177.6, BW=0.46)\\
Scenario_{LR}=&&(UR_{SCT}=20, UR_{RC}=20,L_M=22~118,\nonumber\\\nonumber && L_R\in\left[0..0.99\right]\cdot L_M,BW=0.46)\\
Scenario_{BW}=&&(UR_{SCT}=20, UR_{RC}=20,L_M=22~118,L_R=1177.6,\nonumber\\\nonumber &&  BW\in\left[0..0.99\right])
\end{eqnarray}

\hfill\\
\textbf{Sensitivity Analysis}
	
	In this section, we analyze the sensitivity of the BLS model when varying the BLS parameters and utilization rates, i.e., $UR_{SCT}$, $UR_{RC}$, $L_M$, $L_R$, $BW$. The results of the different scenarios are reported in Fig. \ref{fig:SCTImpactCM}, Fig. \ref{fig:RCImpactCM} Fig. \ref{fig:LMImpactCM}, Fig. \ref{fig:BWImpactCM}, and Fig. \ref{fig:LRImpactCM}.

	From our modelisation of the output port multiplexer and the BLS node, we notice that in  $\beta_{SCT}^{mux}(t)$ (see Corollary~\ref{cor:SCTblssp}) and $\beta_{RC}^{mux}(t)$ (see Corollary~\ref{cor:RCblssp}) the evolution of the strict minimum service curves of SCT and RC follows the maximum of two linear curves: one is the SP part, the other is the BLS part. Consequently, the delay bounds of SCT and RC evolve also following two parts under the different scenarios.
	
\hfill\\
\textit{Impact of $UR_{SCT}$}
	
	\label{modelbehavour}
	%
	%
	%
	%
	
	In Fig. \ref{fig:SCTImpactCM}(a), when the SCT utilization rate increases, we observe an increase of the SCT delay bounds starting close to 0 thanks to a low initial latency and high rate of the guaranteed minimum service curve due to the  BLS part. Then, at $UR_{SCT}=20\%$, the delay bounds due to the BLS part reaches the delay bounds due to the SP part. After this point, the delay follows the maximum rate according to Corollary~\ref{cor:SCTblssp}, i.e., the SP part.
	%
	Furthermore, in Fig. \ref{fig:SCTImpactCM}(b), when the SCT utilization rate increases, we observe a noticeable increase of the RC delay bounds following the increasing guaranteed rate of the service curve due to the SP part. Then, after $UR_{SCT}=18\%$, the delay bounds becomes constant since it depends on the strict minimum service curve due to the BLS part, which is constant when the RC utilization rate is constant.
	\begin{figure}[htbp]
		\centering
		\subfigure[]{\includegraphics[width=0.345\columnwidth, angle=270]{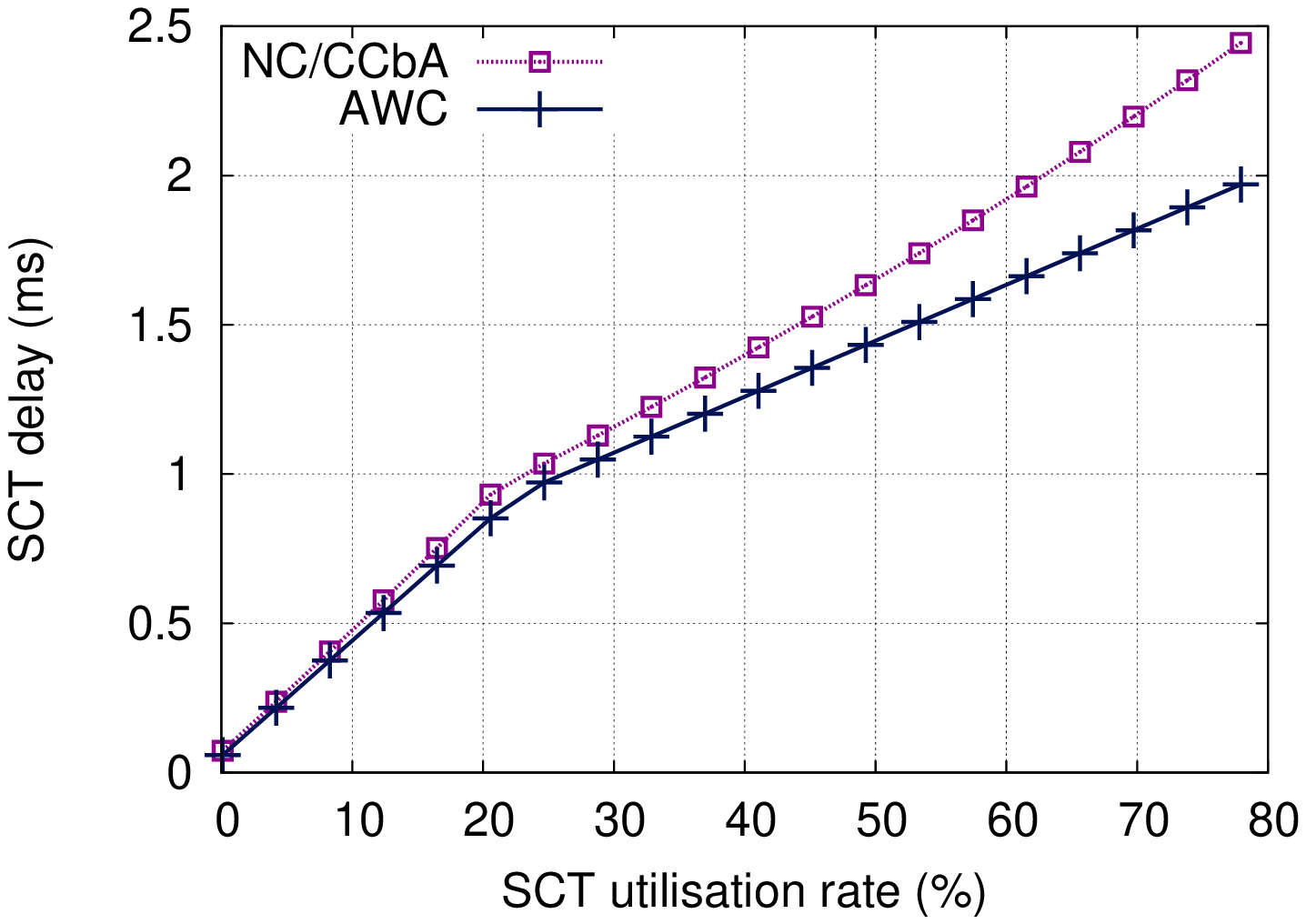}}
		\centering
		\subfigure[]{\includegraphics[width=0.345\columnwidth, angle=270]{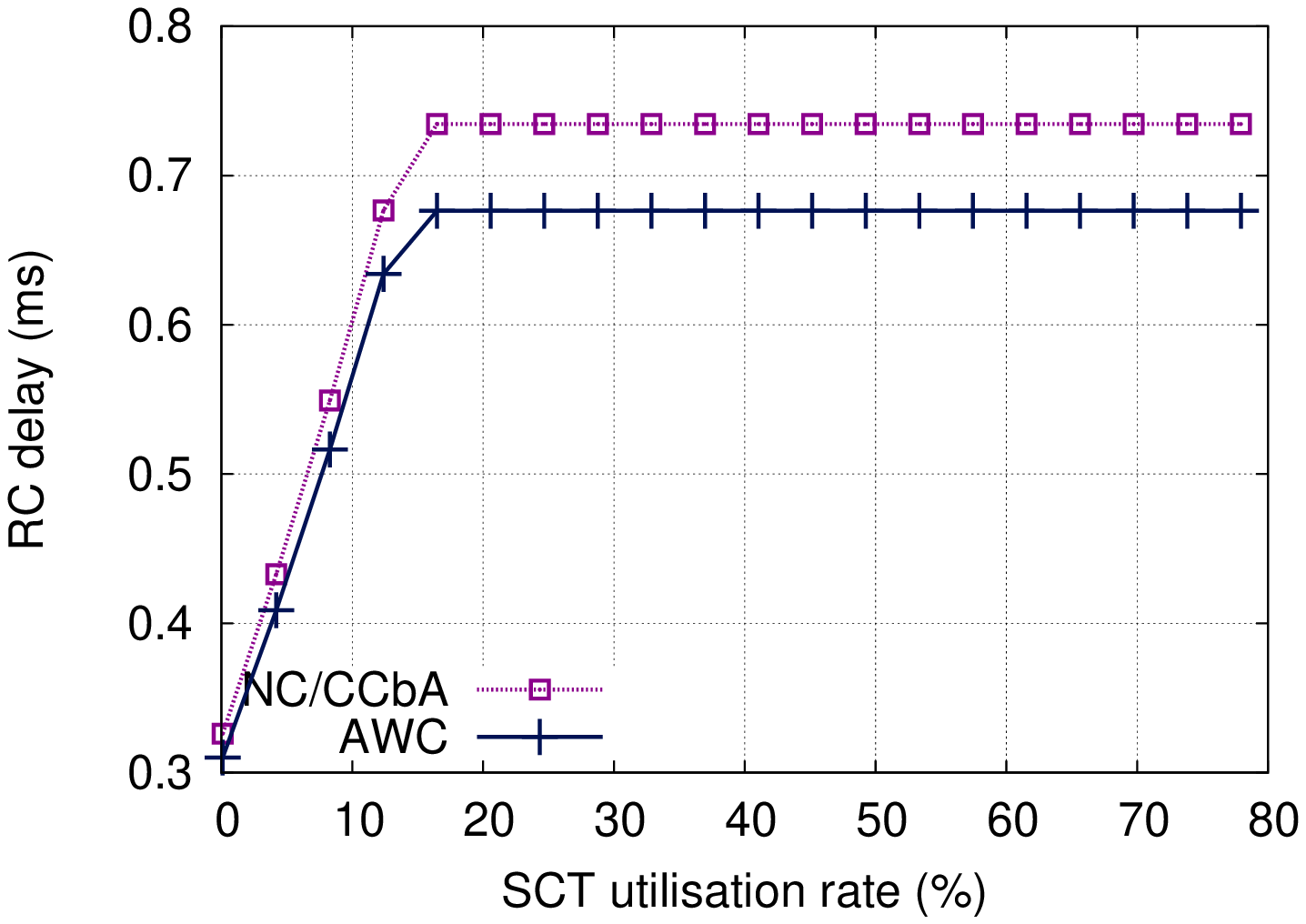}}
		\footnotesize \caption[NC vs AWC: impact of SCT maximum utilisation rate on delay bounds]{NC vs AWC - impact of SCT maximum utilization rate on: (a) SCT delay bounds; (b) RC delay bounds, with $Scenario_{SCT}=\left( UR_{SCT}\in\left[0.1:78\right], UR_{RC}=20,L_M=22~118,L_R=0,BW=0.46\right)$}
		\label{fig:SCTImpactCM}
	\end{figure}	
	
	Hence, this analysis shows that the SCT utilization rate has an inherent impact on SCT and RC delay bounds, where:	
	\begin{itemize}
		\item the SCT delay bound is ruled below $UR_{SCT}=20\%$ by the strict minimum service curve due to the BLS part, $(\beta_{SCT,0}^{sp}\otimes\beta_{SCT}^{bls})(t)$; whereas after $UR_{SCT}=20\%$, it is ruled by the strict minimum service curve due to the SP part, $\beta_{SCT,2}^{sp}(t)$;
		\item the RC delay bound is ruled below $UR_{SCT}=18\%$ by the strict minimum service curve due to the SP part, $\beta_{RC}^{sp}(t)$; whereas after $UR_{SCT}=18\%$, it is ruled by the strict minimum service curve due to the BLS part, $\beta_{RC}^{bls}(t)$.
	\end{itemize}
	These results infer that the variation of $UR_{SCT}$ has a large impact on both the SCT (resp. RC) delay bounds with a maximum variation of 2.5ms (resp. 0.4ms), i.e., the delay bound is multiplied by 24 (resp. 2.3).
	
\hfill\\
\textit{Impact of $UR_{RC}$}

	Similar analysis conducted for $Scenario_{RC}=(UR_{SCT}=20, UR_{RC}\in\left[0.5:72\right],$ \\$L_M=22~118,L_R=0,BW=0.46)$ in Fig. \ref{fig:RCImpactCM}  shows that the RC utilization rate has an inherent impact on SCT and RC delay bounds. We observe a behavior symmetrical to the one noticed in $Scenario_{SCT}$:	
	\begin{itemize}
		\item the SCT delay bound is ruled below $UR_{RC}=20\%$ by the strict minimum service curve due to the SP part; whereas after $UR_{RC}=20\%$, it is ruled by the strict minimum service curve due to the BLS part;
		\item the RC delay bound is ruled below $UR_{RC}=30\%$ by the strict minimum service curve the BLS part; whereas after $UR_{RC}=30\%$, it is ruled by the strict minimum service curve  the SP part.
	\end{itemize} 
	
	These results show that the variation of $UR_{RC}$ has a large impact on both the SCT (resp. RC) delay bounds with a maximum variation of 0.58ms (resp. 2.2ms), i.e., the delay bound is multiplied by 2.8 (resp. 23).

	\begin{figure}[htbp]
		\centering
		\subfigure[]{\includegraphics[width=0.345\columnwidth, angle=270]{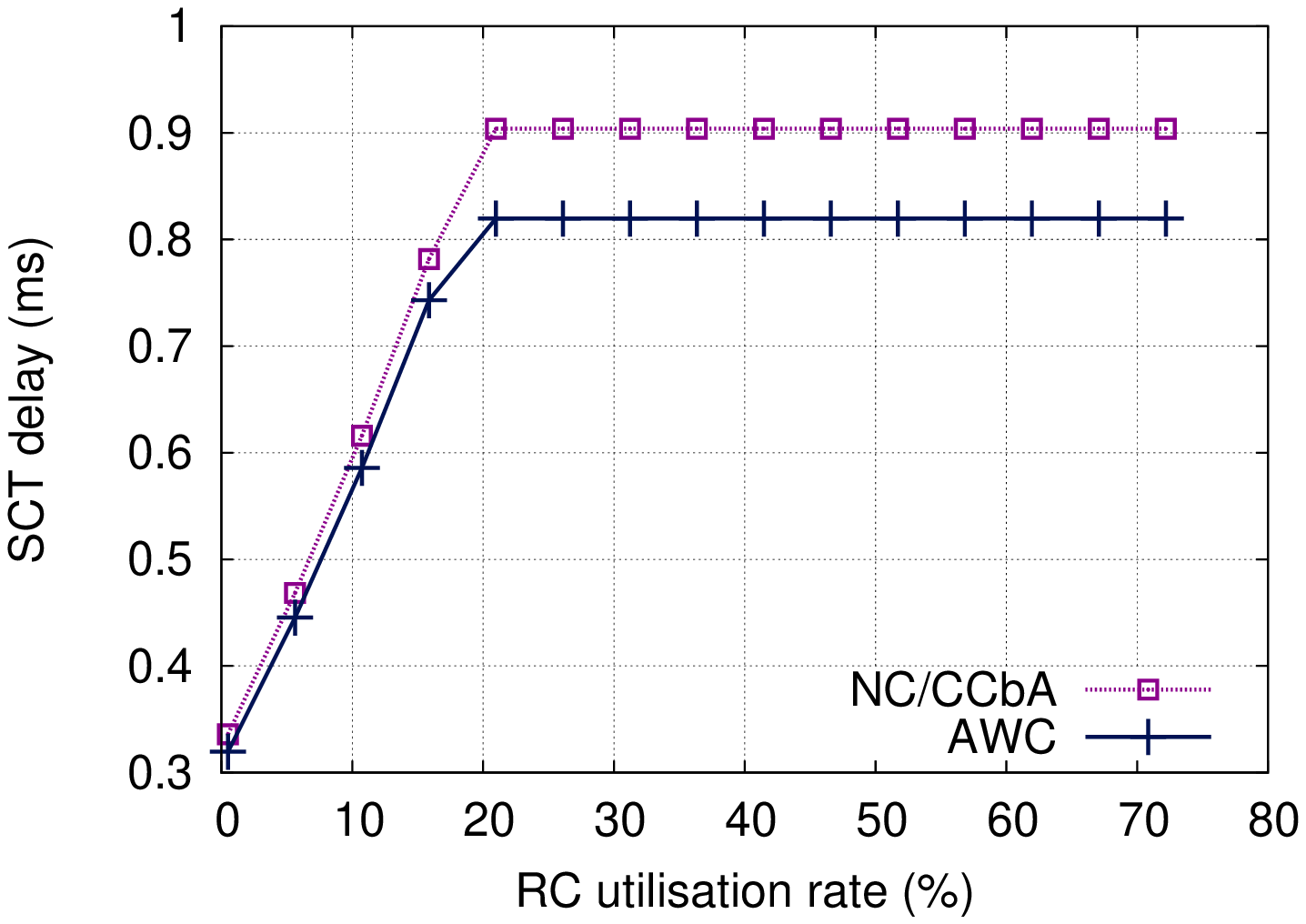}}
		\centering
		\subfigure[]{\includegraphics[width=0.345\columnwidth, angle=270]{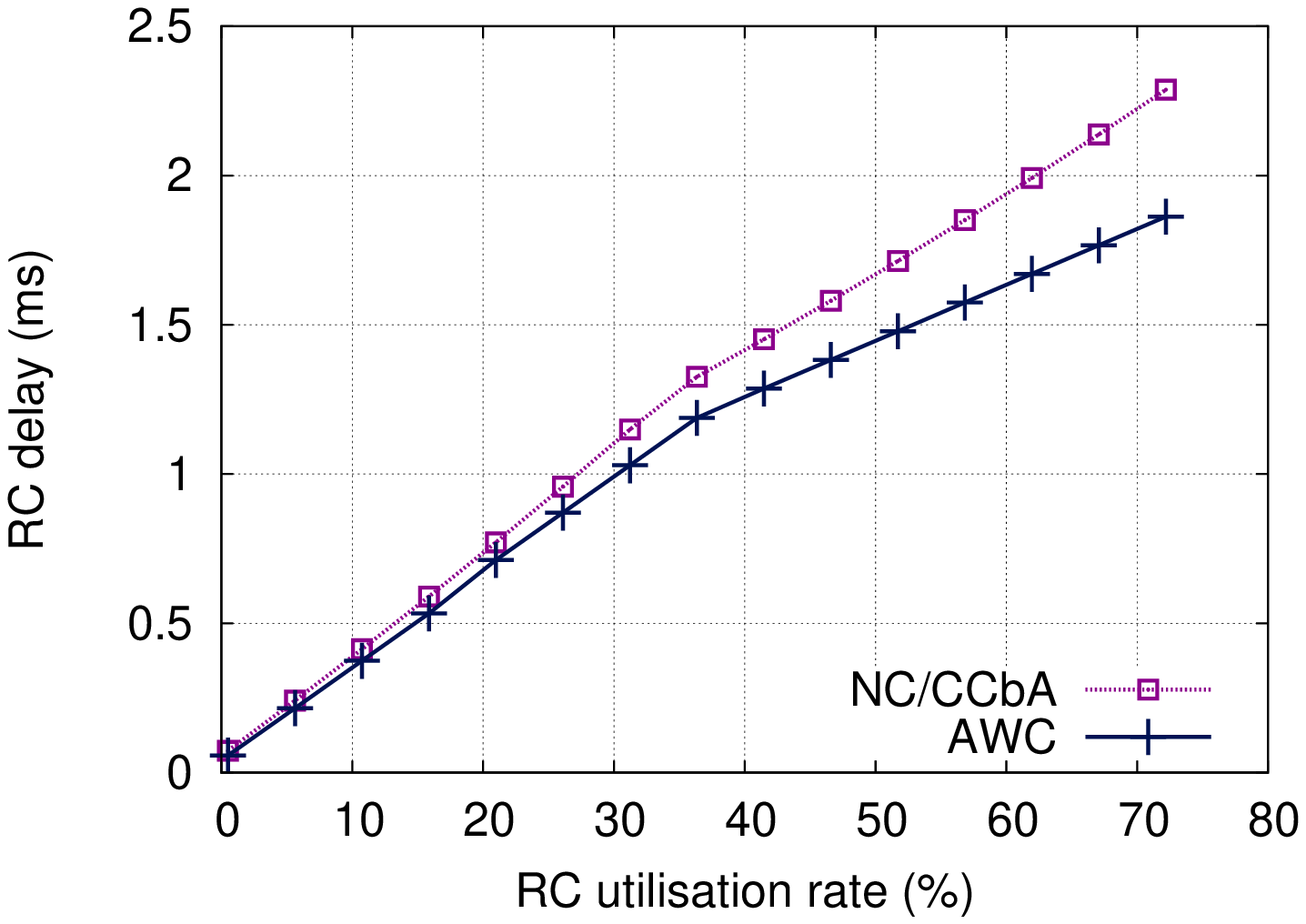}}
		\footnotesize \caption[NC vs AWC: impact of RC maximum utilisation rate on delay bounds]{NC vs AWC - impact of RC maximum utilisation rate on: (a) SCT delay bounds; (b) RC delay bounds, with $Scenario_{RC}=\left(UR_{SCT}=20, UR_{RC}\in\left[0.5:72\right],L_M=22~118,L_R=0,BW=0.46\right)$}
		\label{fig:RCImpactCM}
	\end{figure}

	
	\hfill\\
	\textit{Impact of $L_M$}
	
	Concerning SCT delay bounds, in Fig. \ref{fig:LMImpactCM}(a), before $L_M=50~000$~bits (which represents sending windows allowing the transmission of 200 consecutive SCT frames)
	they are ruled by $\beta^{bls}_{SCT}$. When $L_M$ decreases, the minimum service rate of $\beta_{SCT}^{bls}$ increases and its initial latency decreases. Consequently, the SCT delay bounds decrease with CCbA when $L_M$ decreases toward $L_R$. After $L_M=50~000$~bits, the SCT delay bound is constant, because it is ruled by a constant $\beta^{sp}_{SCT,0}$. 
	
	Concerning RC delay bounds, in Fig. \ref{fig:LMImpactCM}(b), they are ruled by $\beta_{RC}^{bls}(t)$ (see Corollary~\ref{cor:RCblssp}). When increasing $L_M$, both  the rate and initial latency increase.
	Before $L_M=7~000$~bits (which represents sending windows allowing the transmission of 24 consecutive SCT frames), the impact of the increasing rate is stronger, resulting in the delay bound decrease; whereas after $L_M=7~000$~bits, the impact of the initial latency takes over, resulting in a delay bound increase.
	
	These results show that the variation of $L_M$ has a limited impact on the SCT delay bounds with a maximum variation of 5\%, and a larger effect on RC delay bounds with a variation of  35\%.
	
	\begin{figure}[htbp]
		\centering
		\subfigure[]{\includegraphics[width=0.345\columnwidth, angle=270]{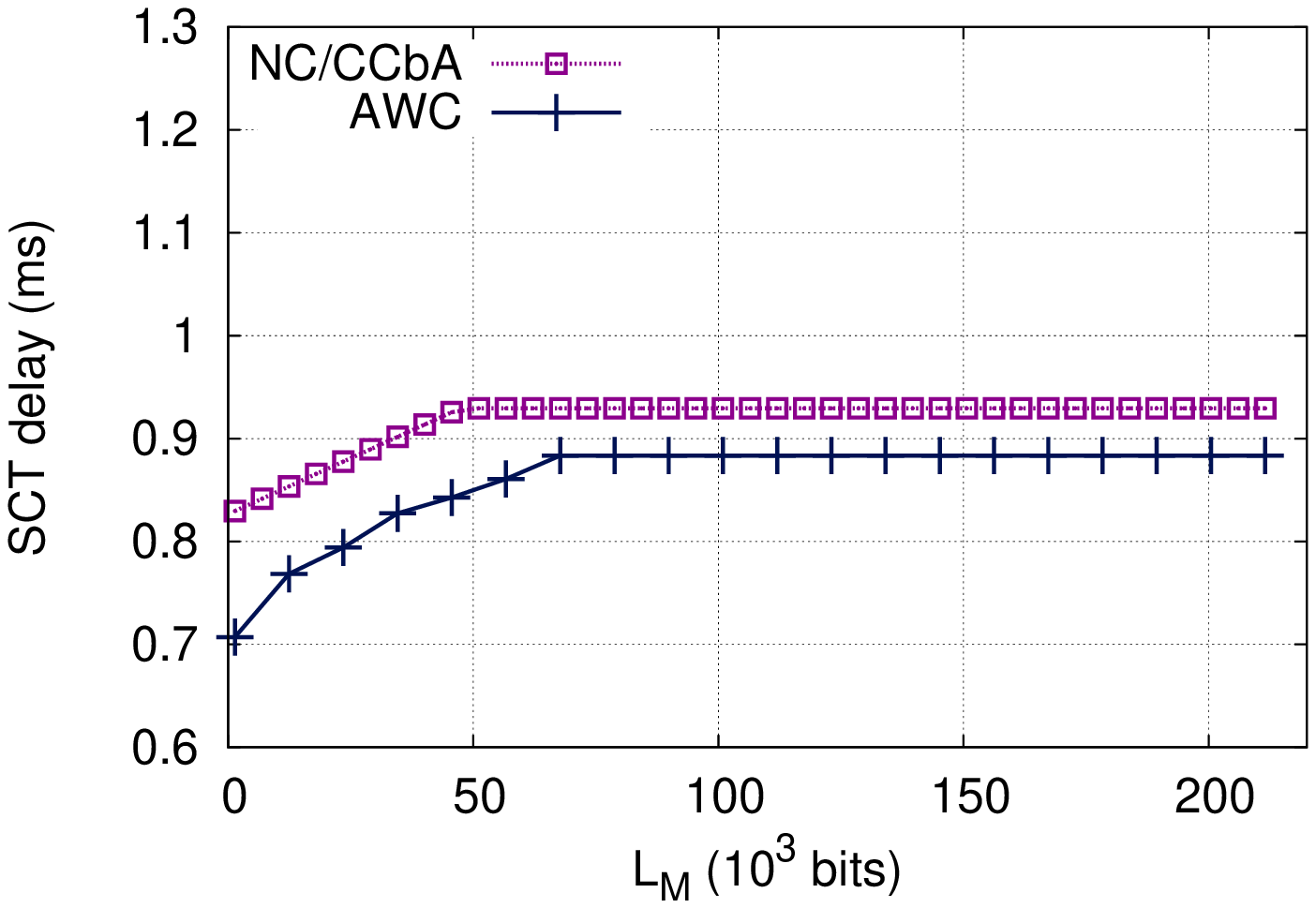}}
		\centering
		\subfigure[]{\includegraphics[width=0.345\columnwidth, angle=270]{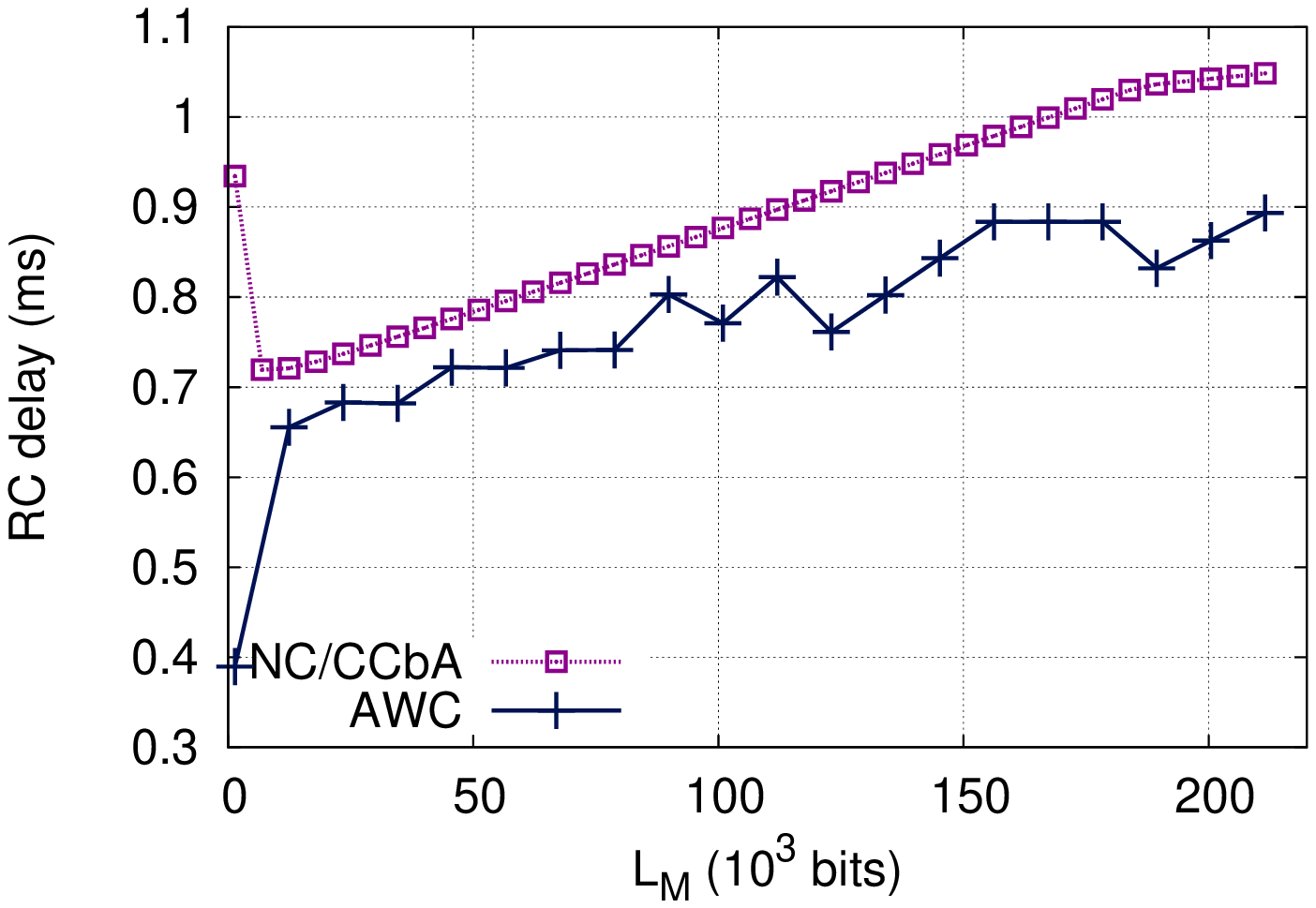}}
		\footnotesize \caption[NC vs AWC: impact of $L_M$ on delay bounds]{NC vs AWC - impact of $L_M$ on: (a) SCT delay bounds; (b) RC delay bounds, with $Scenario_{LM}=\left(UR_{SCT}=20, UR_{RC}=20,L_M\in\left[1382.4..216~830\right],L_R=1177.6 ,BW=0.46\right)$}
		\label{fig:LMImpactCM}
	\end{figure}
	\newpage				
	\noindent\textit{Impact of $BW$}
	
	As shown in Fig. \ref{fig:BWImpactCM}(a), when $BW$ is below 40\%, the SCT delay bound is constant since it is ruled by $\beta_{SCT,2}^{sp}$ (see Corollary~\ref{cor:SCTblssp}): the bandwidth allocated by the BLS is not sufficient to send the SCT traffic. As a consequence, the SCT traffic also uses the bandwidth left by the RC traffic. However, for BW higher than 40\%, SCT delay bound decreases. This is due to the fact that $I_{idle}=BW\cdot C$, thus the guaranteed rate of the SCT minimum service curve $\beta_{SCT}^{bls}$ increases while its initial latency  decreases.

	\begin{figure}[htbp]
		\centering
		\subfigure[]{\includegraphics[width=0.345\columnwidth, angle=270]{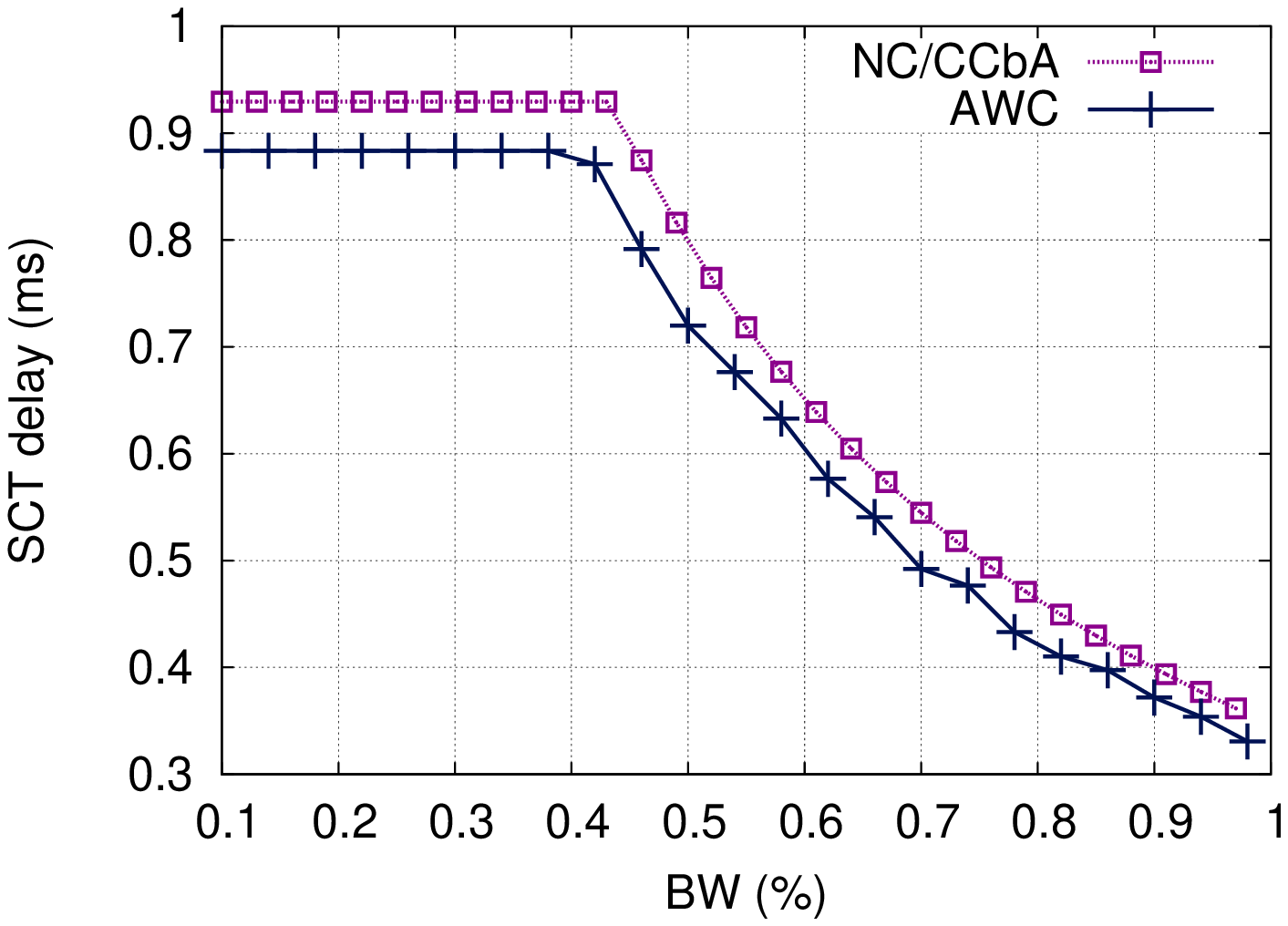}}
		\centering
		\subfigure[]{\includegraphics[width=0.345\columnwidth, angle=270]{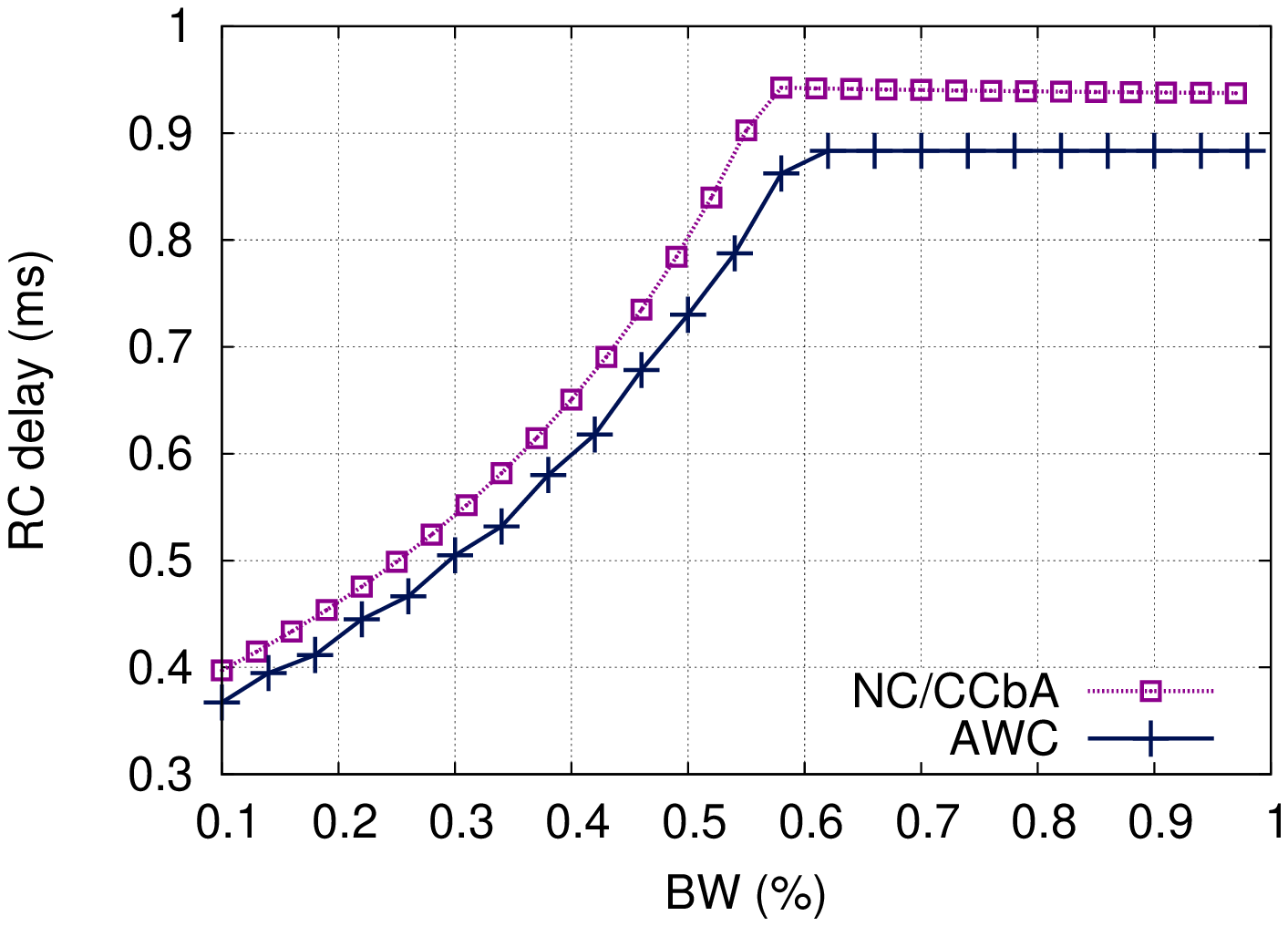}}
		\footnotesize \caption[NC vs AWC: impact of $BW$ on delay bounds]{NC vs AWC - impact of $BW$ on: (a) SCT delay bounds; (b) RC delay bounds, with $Scenario_{BW}=\left(UR_{SCT}=20, UR_{RC}=20,L_M=22~118,L_R=1177.6 ,BW\in\left[0..0.99\right]\right)$}
		\label{fig:BWImpactCM}
	\end{figure}
		
	Concerning the RC delay bounds, we observe in Fig. \ref{fig:BWImpactCM}  the opposite behaviors: for BW higher than 55\%, RC delay bound is constant and ruled by $\beta_{RC}^{sp}$  (see Corollary~\ref{cor:RCblssp}); whereas for BW lower than 55\%, RC delay bound is ruled by $\beta_{RC}^{bls}$. Thus, when increasing  $BW$, $I_{idle}$ increases. This leads to decreasing the guaranteed rate of $\beta_{RC}^{bls}$, and consequently to the delay bounds increase.

	These results  show that the variation of $BW$ has a high impact on both SCT and RC delay bounds with an increase of 0.60ms (resp. 0.55ms) for SCT (resp. RC), representing an increase of  170\% (resp. 137\%).

	\hfill\\
	\textit{Impact of $L_R$}

SCT and RC delay bounds when varying  $L_R$ are shown in Fig. \ref{fig:LRImpactCM}. We notice that the SCT delay bound with CCbA remains firmly below the limit set by 
$\beta_{SCT,2}^{sp}$, and it is always ruled by $\beta_{SCT}^{bls}\otimes\beta_{SCT,0}^{sp}$. When $L_R$ increases, $L_M-L_R$ decreases, leading to the slow decrease of the initial latency of $\beta_{SCT}^{bls}$. Additionally,	
$MFS_{RC}^{sat}$ decreases until it hits 0 at $L_R=MFS_{RC}\cdot \frac{I_{idle}}{C}$ (according to Theorem~\ref{cor:k-maxCCbA}). This happens  in Fig. \ref{fig:LRImpactCM} at $L_R=0.053\cdot L_M$.

	\begin{figure}[htbp]
		\centering
		\subfigure[]{\includegraphics[width=0.345\columnwidth, angle=270]{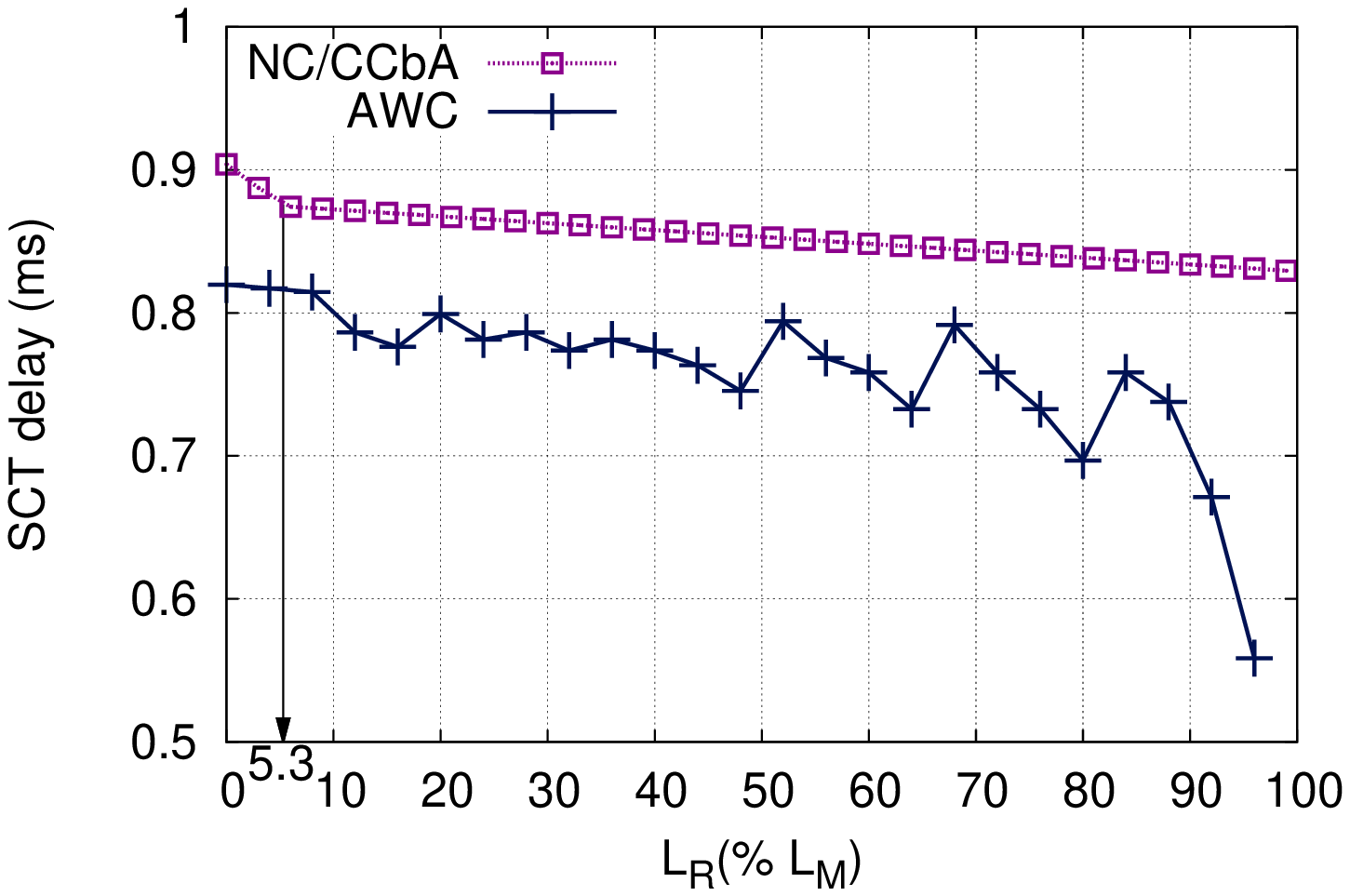}}
		\centering
		\subfigure[]{\includegraphics[width=0.345\columnwidth, angle=270]{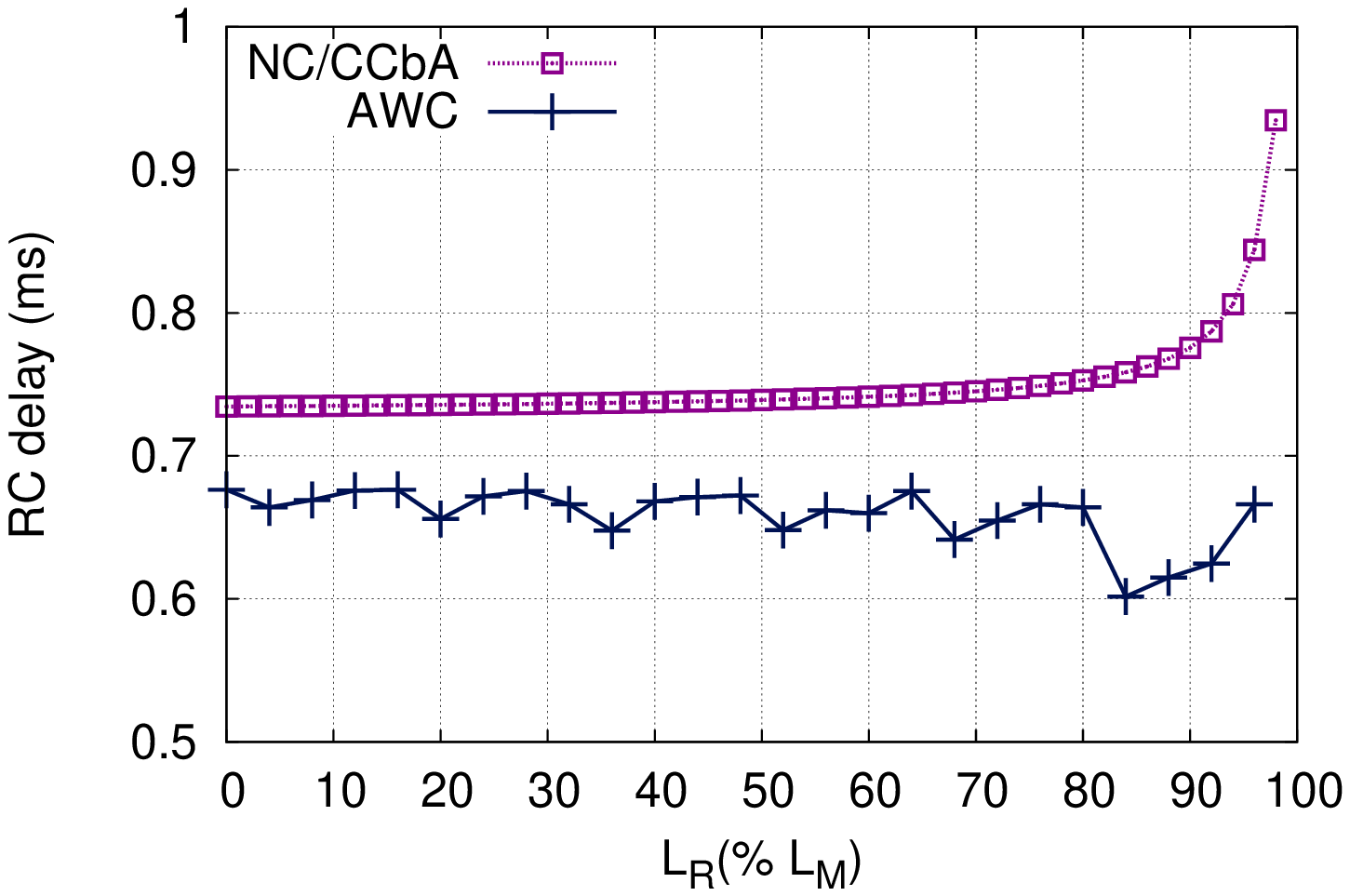}}
		\footnotesize \caption[NC vs AWC: impact of $L_R$ on delay bounds]{NC vs AWC - impact of $L_R$ on: (a) SCT delay bounds; (b) RC delay bounds, with $Scenario_{LR}=\left(UR_{SCT}=20, UR_{RC}=20,L_M=22~118,L_R\in\left[0..0.99\right]\cdot L_M,BW=0.46\right)$}
		\label{fig:LRImpactCM}
	\end{figure}
	
In Fig. \ref{fig:LRImpactCM}(b), the RC delay bound is ruled by $\beta_{RC}^{bls}$. Hence, the RC delay bounds increase when $L_R$ increases.

	These results show that $L_R$  has a limited impact on the SCT delay with a maximum variation of 9\%; whereas its impact is higher on RC delay bounds, with a variation of 28\%.
	\hfill\\
	
	\textbf{\textit{To conclude the sensitivity analysis of the CCbA model, both $L_R$ and $L_M$ have limited impact on the SCT delay bound (under 10\%), but a large one on RC delay bounds (around 30\%). Moreover, the largest impact is due to the the SCT and RC utilization rates and $BW$ parameter, with delay bound increases up to 137\% and 170\% for RC and SCT respectively.}}			
	
	\hfill\\
	\textbf{Tightness Analysis}	
	
	Thanks to the modelisation of both the BLS and SP parts of the output port multiplexer, the SCT and RC delay bounds are very tight in reference to AWC, under the different considered scenarios. For instance, as shown in Fig. \ref{fig:SCTImpactCM}(a), when varying the SCT utilization rate, the maximum gap between the AWC and the NC delay bounds of SCT is $0.5ms$, which represents an increase of $33\%$. 
	Moreover, when varying the RC utilization rate, we have similar results: 
	the largest percentage increase of the RC delay bounds happens for a gap of $0.5ms$ and represents $27\%$ in Fig. \ref{fig:RCImpactCM}(b). When varying $BW$, $L_M$ and $L_R$, we also have similar results in Fig. \ref{fig:LMImpactCM}, Fig. \ref{fig:LRImpactCM} and Fig. \ref{fig:BWImpactCM}, for both SCT and RC delay bounds.

	Finally, with a gap between AWC and the NC model usually below 15\% under the various scenarios, the proposed model can be considered  an accurate one.	
	\hfill\\
	
	\textbf{\textit{Thus, the tightness of the model is very high: the gap between the NC model and AWC is usually less than 15\%, with a peak at 28\%. }}

\hfill\\
\textbf{Comparing NC/CCbA to CPA and NC/WbA}
	
	In this section, we compare our proposed model to the CPA and NC/WbA models. We start by comparing the computation times, before studying the SCT and RC delay bounds.	These comparison are done on the single-hop case study.
	
	Concerning the NC models, i.e., WbA and CCbA, they lead to identical delay bounds at $L_R=0$ (due to identical service curves) and are presented as a single curve for $Scenario_{SCT}$ and $Scenario_{RC}$.	
	
	\hfill\\
	\textit{Computation times}
	
	For each scenario, we consider the computation time necessary to obtain all the delays in the corresponding scenario.

	With both NC models, we compute each delay bound through simple linear computations. Their computation delays are the same since only the rate and initial latency of the BLS node are different, which does not impact the computation time.
	With CPA however, the computation is much more complex:
	
	\begin{itemize}
		\item SCT delay bounds necessitate finding a maximum using a while-loop;
	
		\item RC delay bounds necessitate solving:
		\begin{itemize}
			\item a maximization problem;
			\item fixed point problems;
			\item ILP problems.
			
		\end{itemize}		
	\end{itemize}
	
	We can see clearly in Table \ref{table:computationTime} that the NC models necessitate much less computation power than the CPA model.  In fact, we can notice that the NC models are between 20~000 and 100~000 times faster than the CPA model.

	\begin{table}[h!]
		\footnotesize
		\centering
		\begin{tabular}{|c|c|c|c|}
			\hline
			scenario & CPA  & NC &  CPA/NC\\
			& (s)& (s)&\\
			\hline		
			varying SCT & 97.2 & 0.0051 & 19~058\\
			\hline		
			varying RC &  71.4 & 0.0032 & 22~187 \\	
			\hline		
			varying $L_M$ & 384.4&  0.0072 & 53~388 \\
			\hline		
			varying $L_R$ &   1059 & 0.0100 & 105~900\\					
			\hline		
			varying $BW$ &  390 & 0.0095 & 41~052\\					
			
			\hline
		\end{tabular}
		\footnotesize \caption{CPA and NC models computation times for the different scenarios}
		\bigskip
		\label{table:computationTime}
	\end{table}

	\hfill\\
	\textit{SCT delay bounds}

	We can see in Fig. \ref{fig:SCTImpactCPA}(a) that the SCT delay bounds of the three models are overlapping for low values of $UR_{SCT}$. Then, they diverge at $UR_{SCT}=20\%$. The NC curve starts to follow a linear curve with a lower increase rate. The CPA model however, keeps the same rate.
		\begin{figure}[htbp]
			\centering
			\subfigure[]{\includegraphics[width=0.345\columnwidth, angle=270]{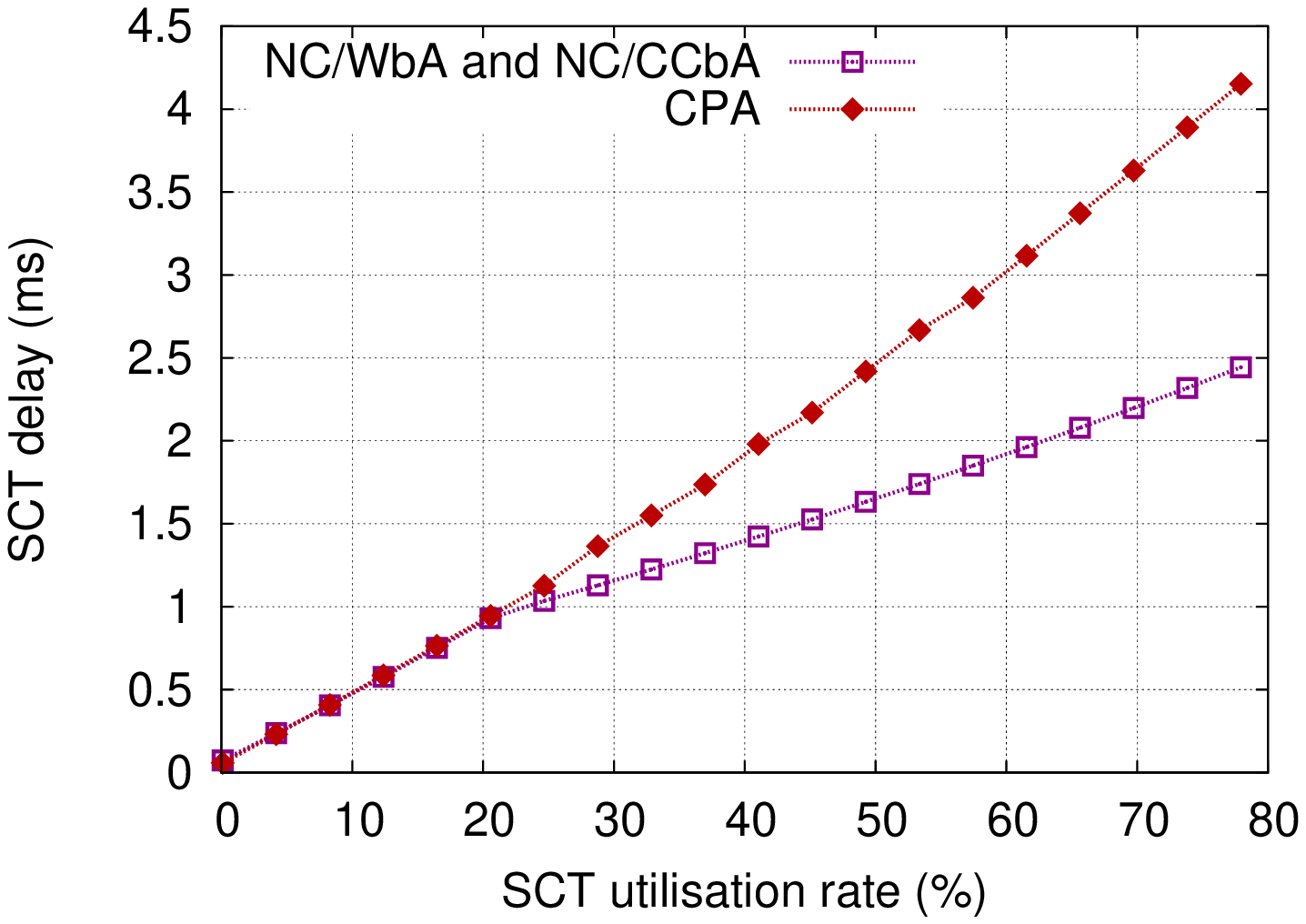}}
			\centering
			\subfigure[]{\includegraphics[width=0.345\columnwidth, angle=270]{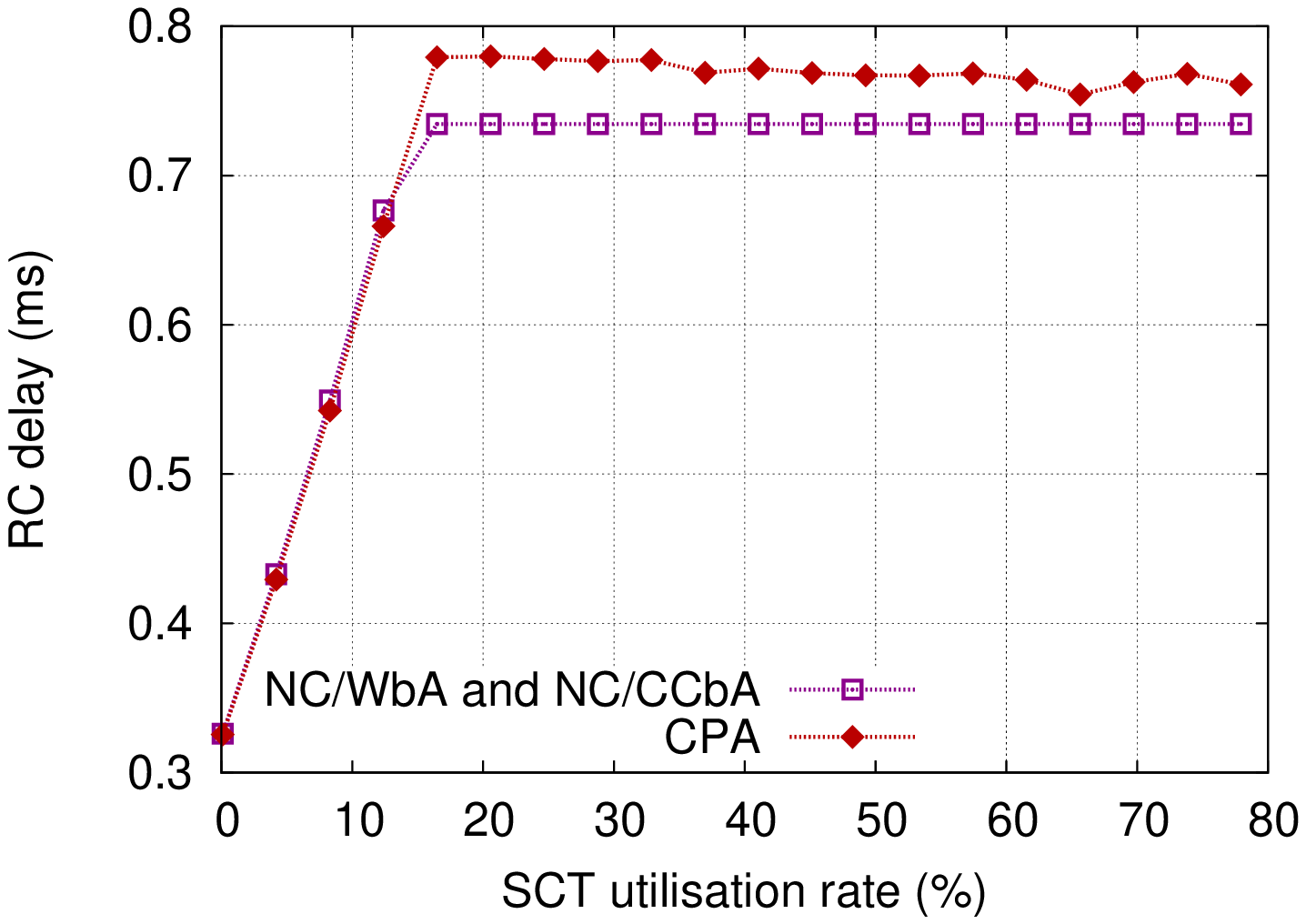}}
			\footnotesize \caption[NC vs CPA: impact of SCT maximum utilisation rate on delay bounds]{NC vs CPA - impact of SCT maximum utilisation rate on: (a) SCT delay bounds; (b) RC delay bounds, with $Scenario_{SCT}=\left( UR_{SCT}\in\left[0.1:78\right], UR_{RC}=20,L_M=22~118,L_R=0,BW=0.46\right)$}
			\label{fig:SCTImpactCPA}
		\end{figure}
		 		
	As a consequence, the gap between CPA and NC curves increases (up to 70\%), which shows the increasing pessimism of CPA under high SCT utilisation rates.


	\begin{figure}[htbp]
		\centering
		\subfigure[]{\includegraphics[width=0.345\columnwidth, angle=270]{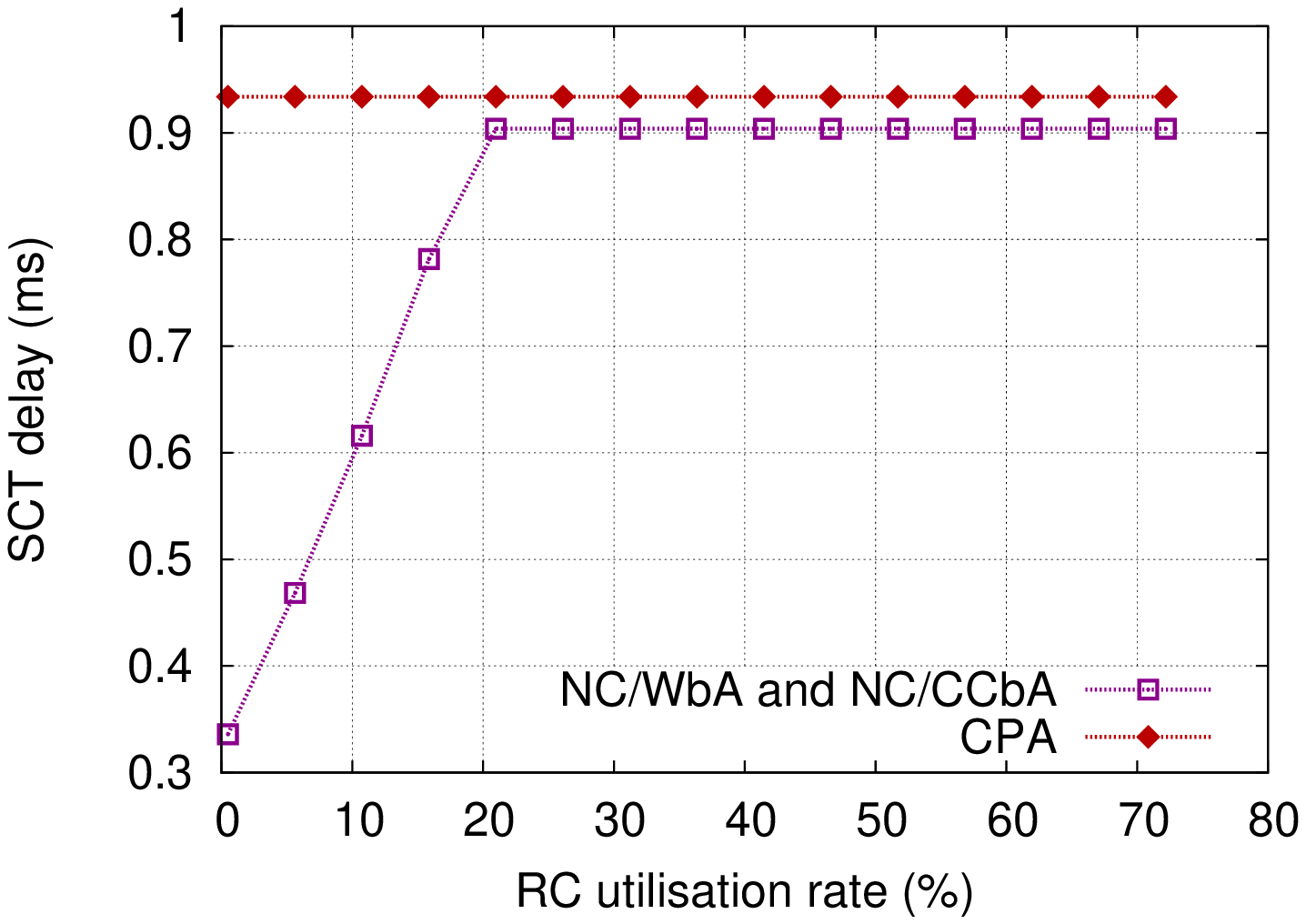}}
		\centering
		\subfigure[]{\includegraphics[width=0.345\columnwidth, angle=270]{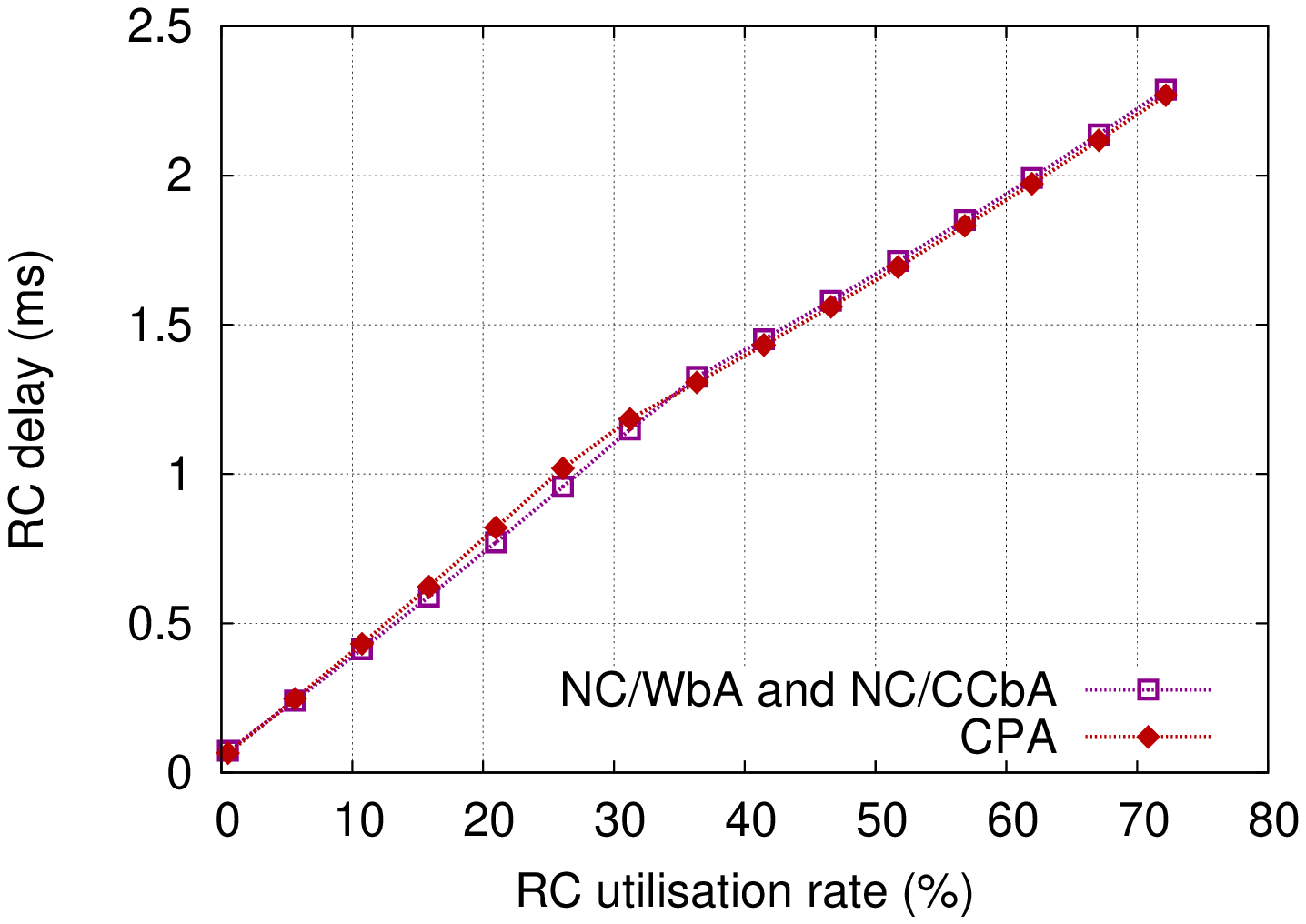}}
		\footnotesize \caption[NC vs CPA: impact of RC maximum utilisation rate on delay bounds]{NC vs CPA - impact of RC maximum utilisation rate on: (a) SCT delay bounds; (b) RC delay bounds, with $Scenario_{RC}=\left(UR_{SCT}=20, UR_{RC}\in\left[0.5:72\right],L_M=22~118,L_R=0,BW=0.46\right)$}
		\label{fig:RCImpactCPA}
	\end{figure}
	The main cause of this pessimism is due to the fact that the so-called CPA shaper blocking impact (see Section~\ref{WCTA}) does not take into account the RC rate. As $BW$ is close to 50\%, the idle and send slopes are very close: the replenishment and service intervals are very similar. So, when the SCT utilization rate becomes visibly larger than the RC one (over $UR_{SCT}=20\%$), the replenishment intervals are not completely filled: SCT traffic is sent even-though the SCT priority is low. This causes the decreasing SCT delay bounds under the NC models. Similar results are visible when varying the different parameters in Fig. \ref{fig:RCImpactCPA}(a), Fig. \ref{fig:LMImpactCPA}(a),  Fig. \ref{fig:LRImpactCPA}(a), and Fig. \ref{fig:BWImpactCPA}(a), where the delay bounds are generally more pessimistic with CPA model than the ones with NC models.
	
	\begin{figure}[htbp]
		\centering
		\subfigure[]{\includegraphics[width=0.345\columnwidth, angle=270]{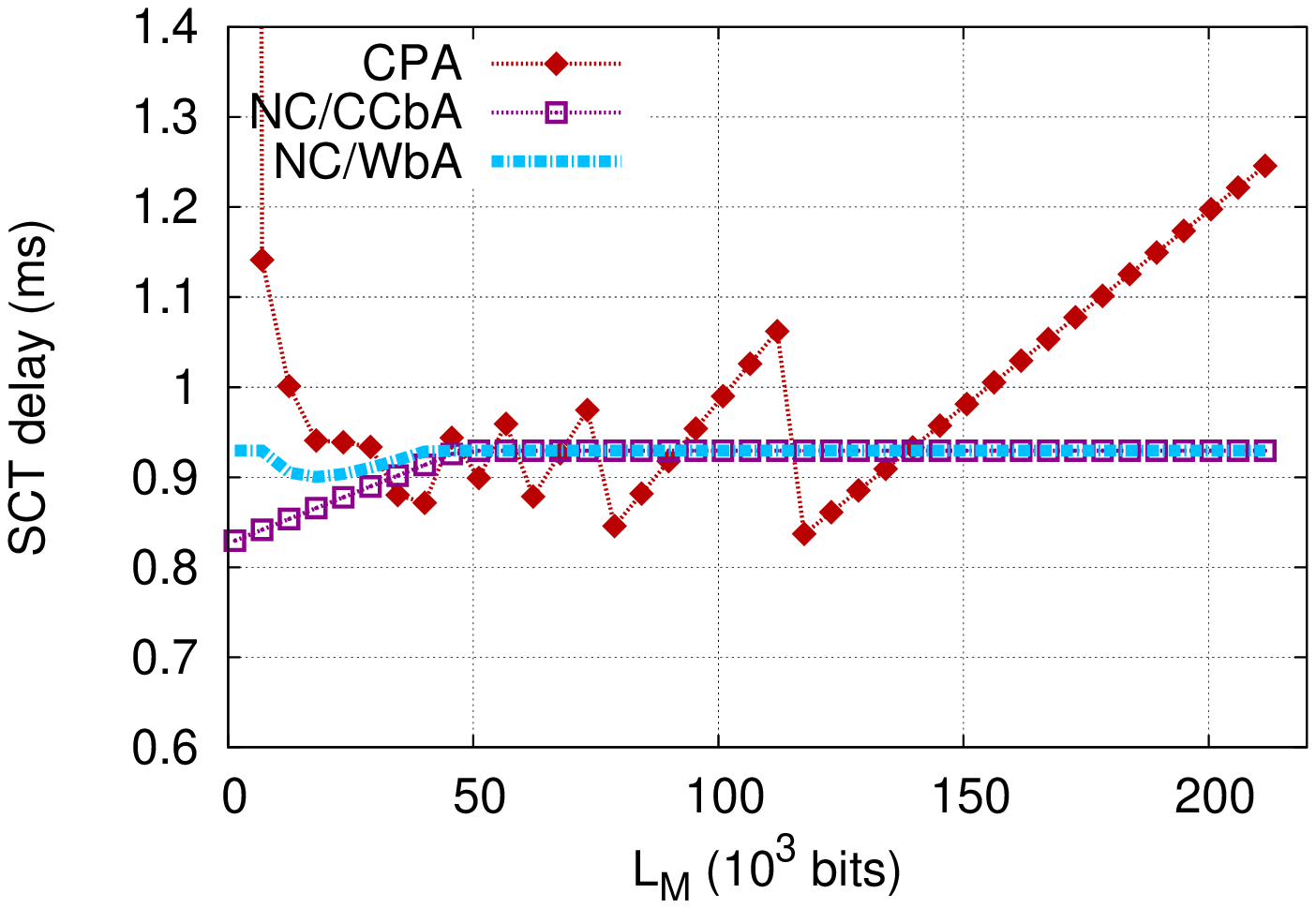}}
		\centering
		\subfigure[]{\includegraphics[width=0.345\columnwidth, angle=270]{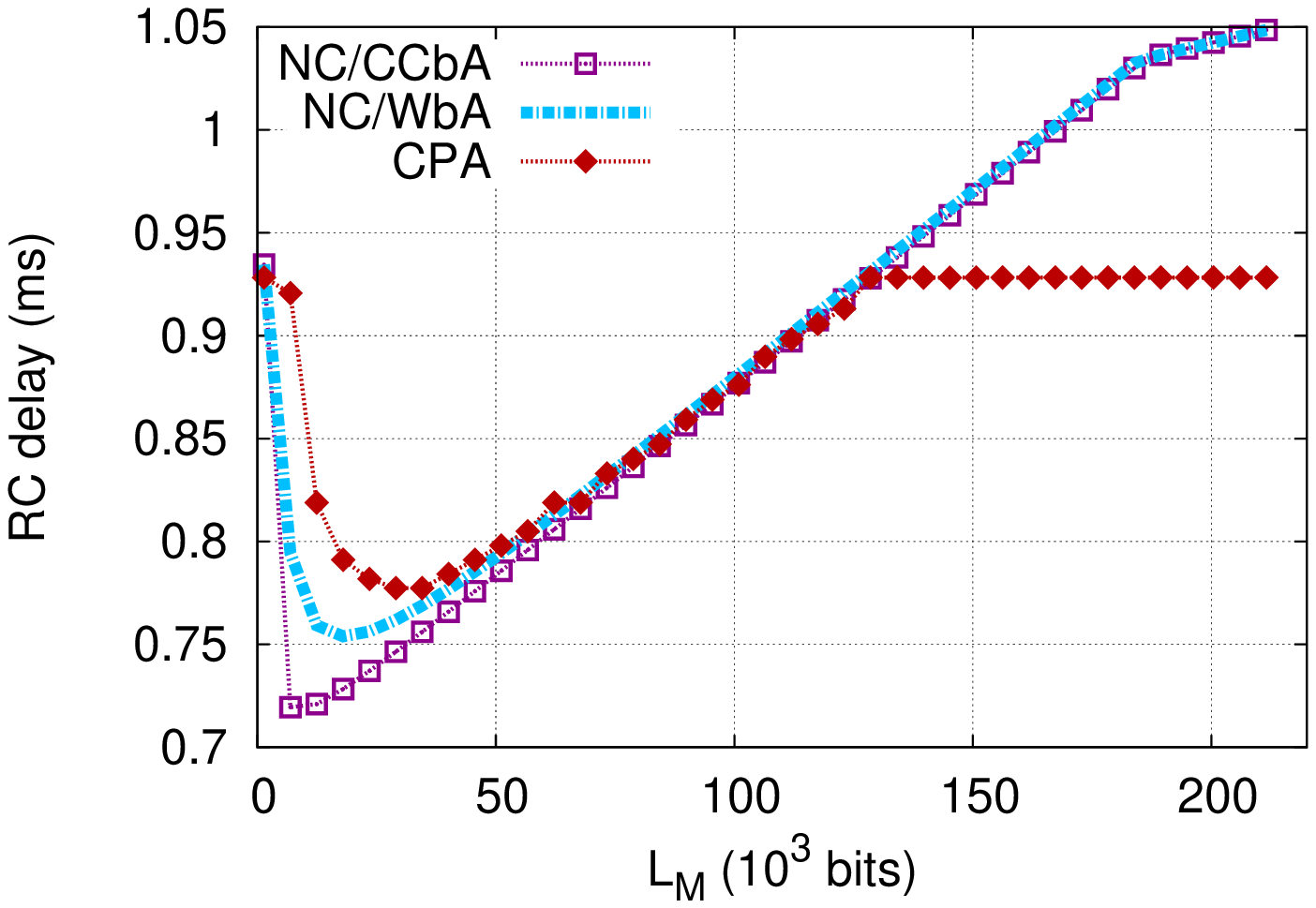}}
		\footnotesize \caption[NC vs CPA: impact of $L_M$ on delay bounds]{NC vs CPA - impact of $L_M$ on: (a) SCT delay bounds; (b) RC delay bounds, with $Scenario_{LM}=\left(UR_{SCT}=20, UR_{RC}=20,L_M\in\left[1382.4..216~830\right],L_R=1177.6 ,BW=0.46\right)$}
		\label{fig:LMImpactCPA}
	\end{figure}	
	
	However,  in  Fig. \ref{fig:LMImpactCPA}(a), we can see that SCT delay bounds with CPA are sometimes lower than the ones with NC models.  This fact confirms our conclusions in Section~\ref{WCTA} about the CPA model optimism. 

	Concerning the comparison of NC/WbA and NC/CCbA, we can see in all figures that the SCT delay bounds with CCbA is consistently equal of lower than with WbA. In particular in Fig. \ref{fig:LRImpactCPA}(a), we can see that, as expected in Section~\ref{WCTA}, the impact of $L_R$ is better taken into account with CCbA, with the SCT  delay bounds under CCbA below the delay bounds with WbA for $L_R>0$.

	\hfill\\
	\textit{RC delay bounds}
	
	
	Concerning the RC traffic, in Fig. \ref{fig:SCTImpactCPA}(b)) and Fig. \ref{fig:RCImpactCPA}(b) all three analytical models have the same shape. In NC, this is again thanks to the association of the BLS and SP parts, $\beta_{RC}^{bls}$ and $\beta_{RC}^{sp}$. In the CPA model, this is thanks to solving an ILP problem, which takes into account the maximum available SCT traffic.

	\begin{figure}[htbp]
		\centering
		\subfigure[]{\includegraphics[width=0.345\columnwidth, angle=270]{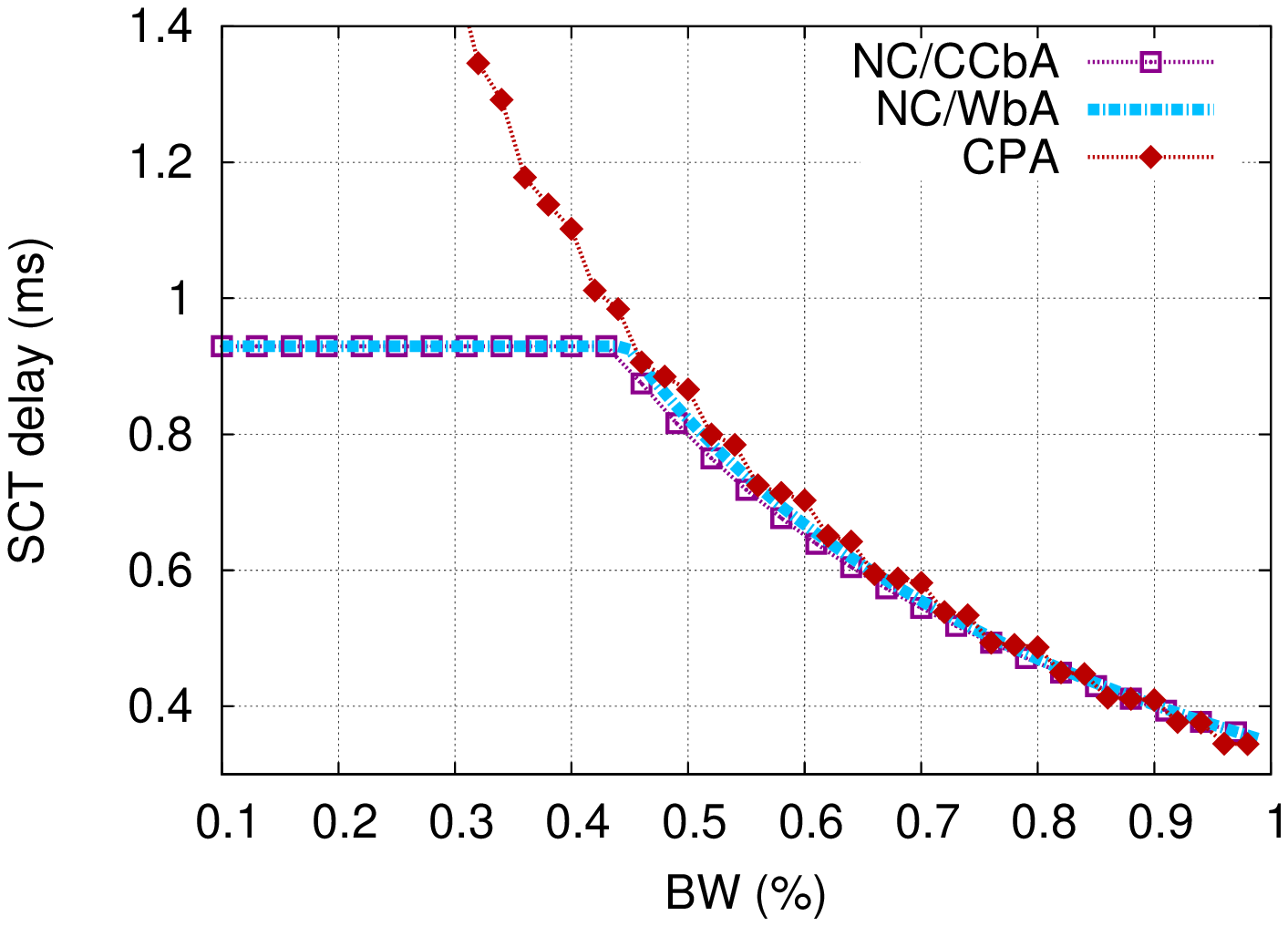}}
		\centering
		\subfigure[]{\includegraphics[width=0.345\columnwidth, angle=270]{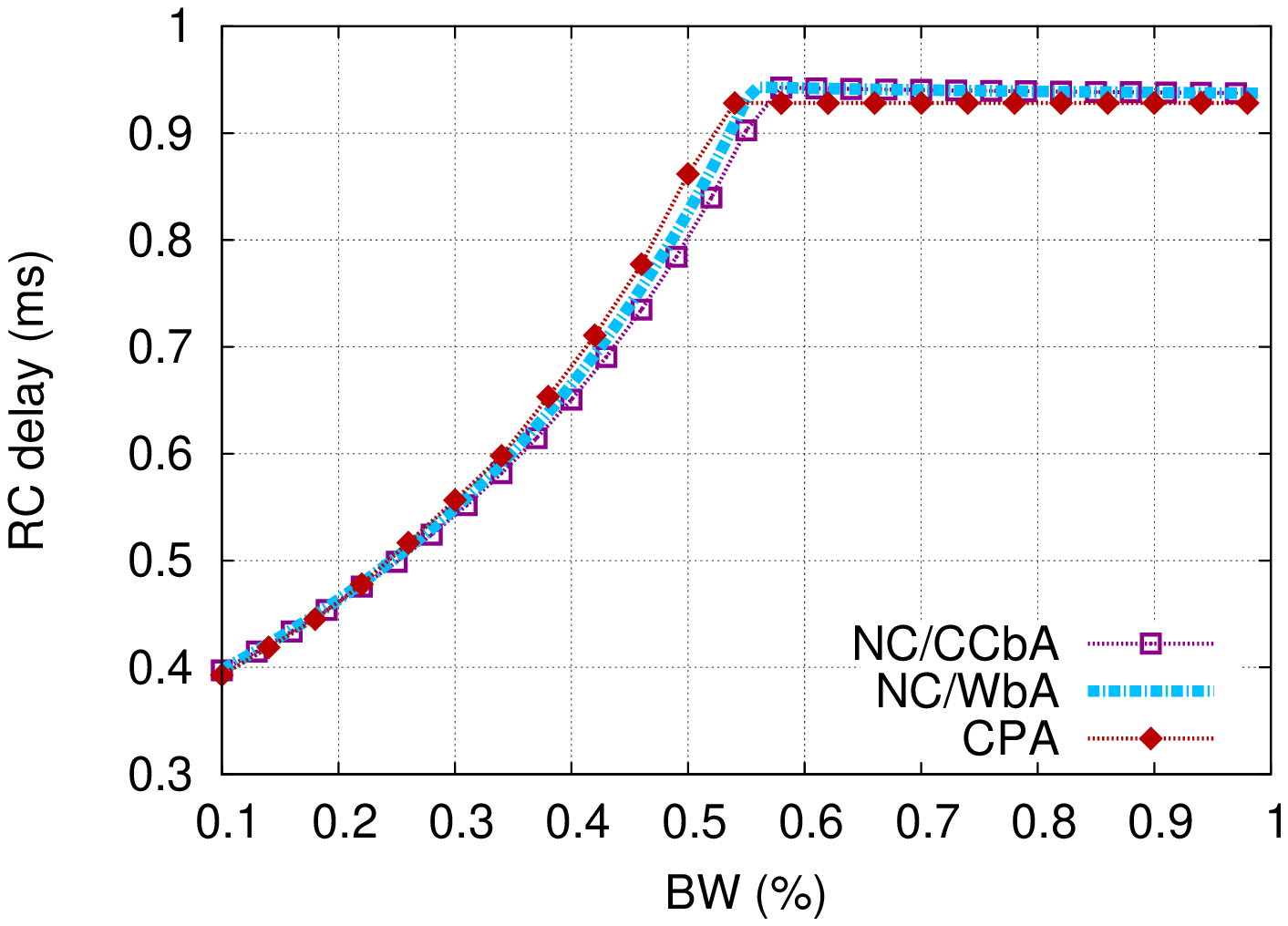}}
		\footnotesize \caption[NC vs CPA: impact of $BW$ on delay bounds]{NC vs CPA - impact of $BW$ on: (a) SCT delay bounds; (b) RC delay bounds, with $Scenario_{LR}=\left(UR_{SCT}=20, UR_{RC}=20,L_M=22~118,L_R\in\left[0..0.99\right]\cdot L_M,BW=0.46\right)$}
		\label{fig:BWImpactCPA}
	\end{figure}

	It is worth  noting that in the part of the curve ruled by the $sp$ node, NC and CPA have very similar delay bounds; whereas in the part ruled by the $bls$ node, CPA leads to slightly larger delay bounds. This can be explained by the pessimism of maximum replenishment intervals $t_{SCT}^{R+}$ in the ILP problem of CPA: the credit replenishment is over-evaluated as explained in Section~\ref{WCTA}. Hence, the resulting sending interval is also over-evaluated, which adds pessimism to the RC delay bounds.
	
	An exception of this general behavior is visible in Fig. \ref{fig:LMImpactCPA}(b) where for large $L_M$, the CPA model is less pessimistic than the NC models (up to 10\%).  

	\begin{figure}[htbp]
		\centering
		\subfigure[]{\includegraphics[width=0.345\columnwidth, angle=270]{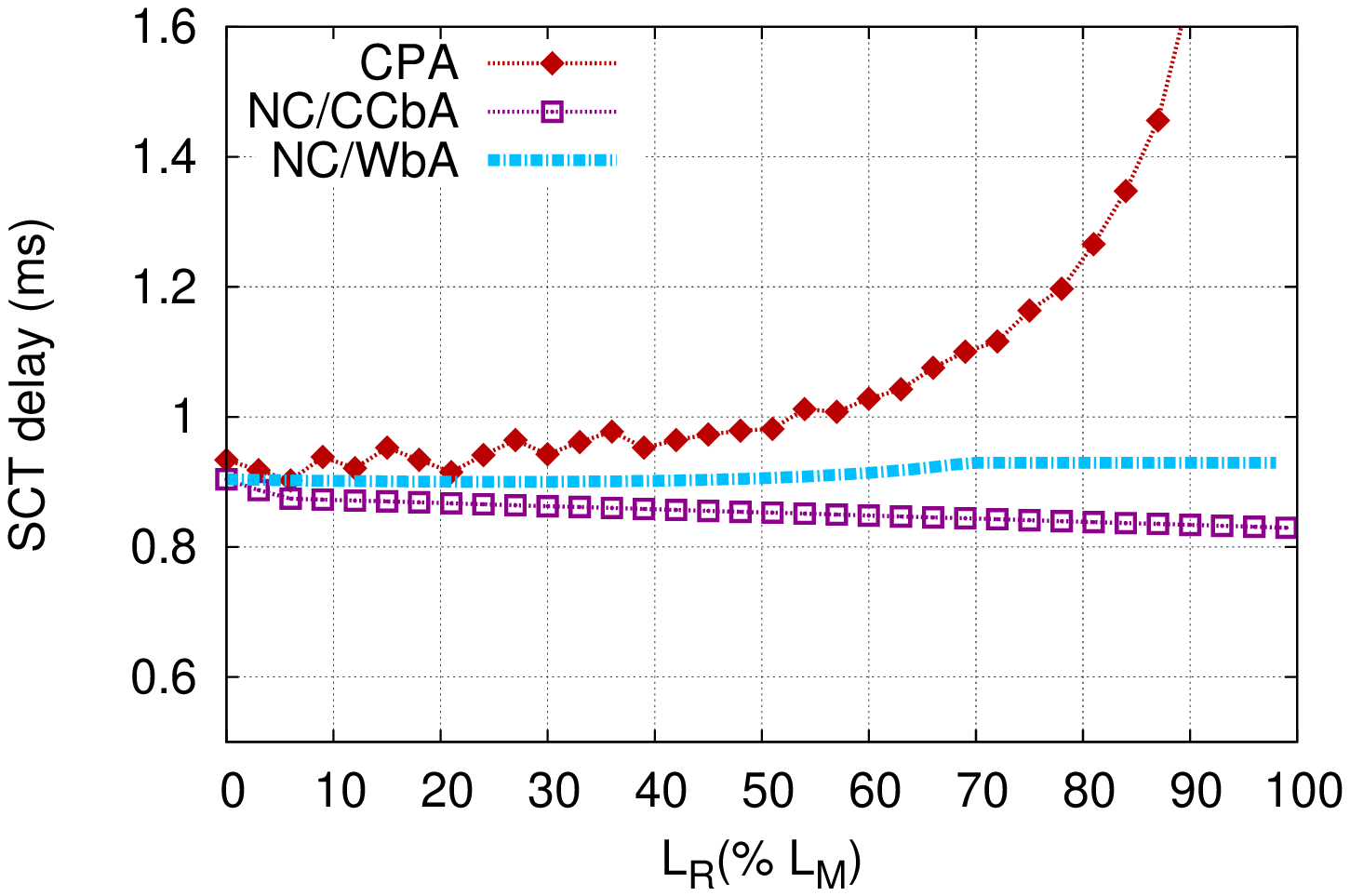}}
		\centering
		\subfigure[]{\includegraphics[width=0.345\columnwidth, angle=270]{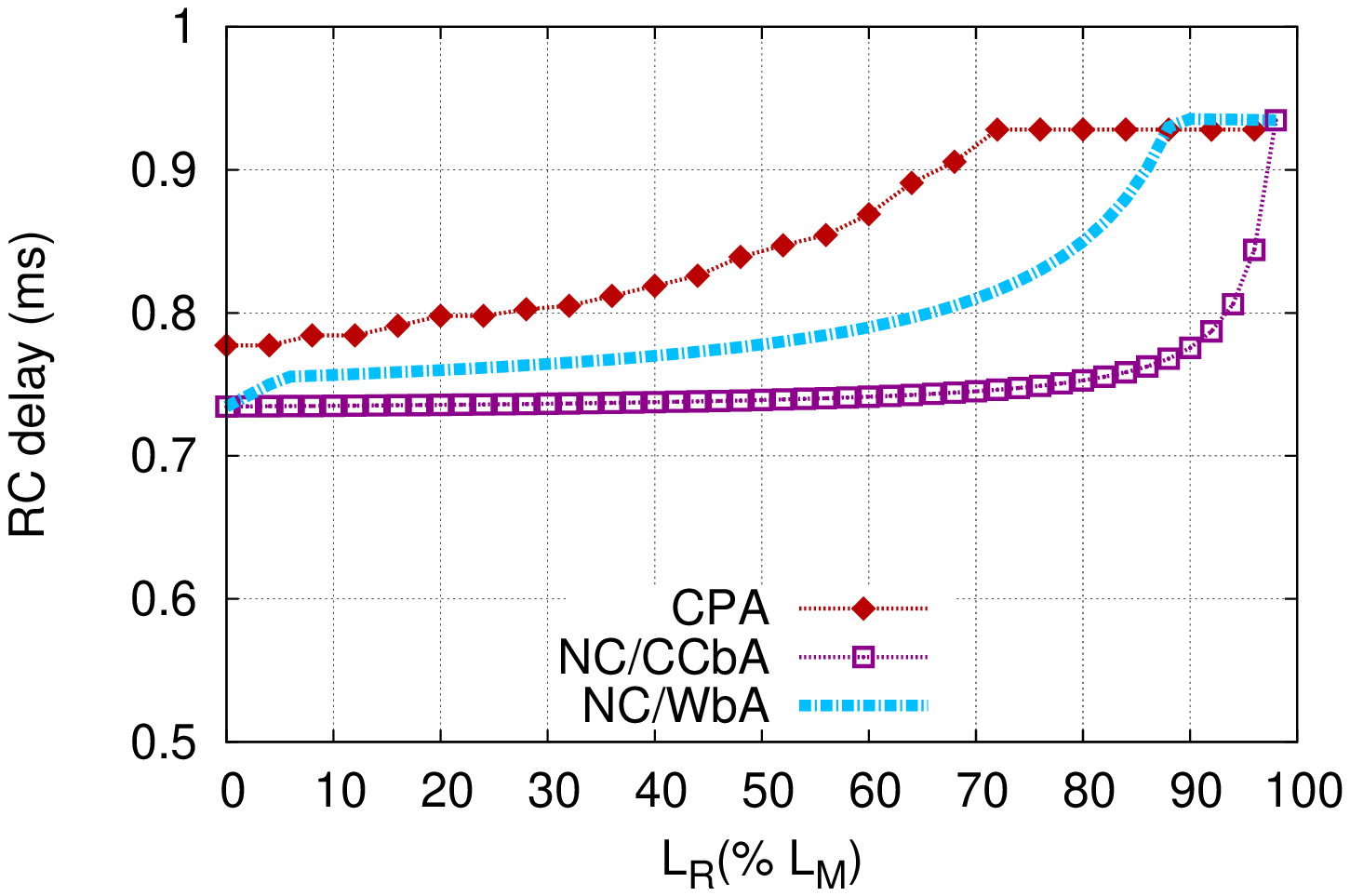}}
		\footnotesize \caption[NC vs CPA: impact of $L_R$ on delay bounds]{NC vs CPA - impact of $L_R$ on: (a) SCT delay bounds; (b) RC delay bounds, with $Scenario_{BW}=\left(UR_{SCT}=20, UR_{RC}=20,L_M=22~118,L_R=1177.6 ,BW\in\left[0..0.99\right]\right)$}
		\label{fig:LRImpactCPA}
	\end{figure}


	\hfill\\

	\textit{\textbf{From this analysis, we can point out the low complexity of our model compared to the CPA model. Moreover, concerning the SCT traffic, we have solved the optimism and pessimism issues of the CPA model. Finally, between the two NC models, i.e., WbA and CCbA, we have shown that CCbA is the clear choice as it better takes into account the impact of $L_R$, which results in a better tightness.}}
	
\subsection{Use-case 2: multiple BLS for 6 classes}
\label{Use-case2}
	In this part, we implement a 6 classes output port architecture to compute the delay in the network presented in Fig. \ref{fig:rlfullnetwork}. We use gCCbA to compute the delay bounds of SCT and RC traffics.
	
	\begin{figure}[htbp]
		\centering
		\subfigure[]{\includegraphics[width=0.40\linewidth]{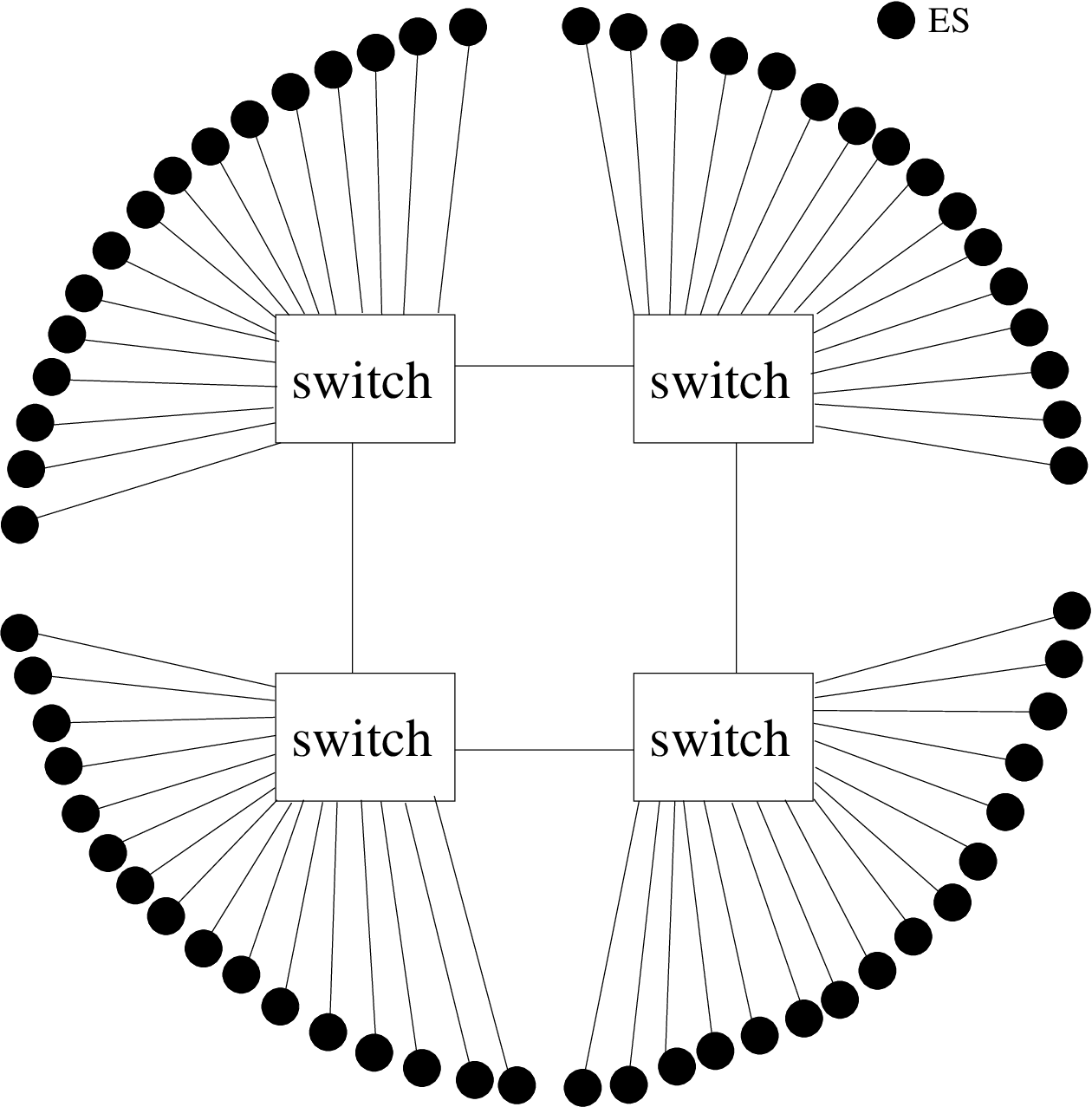}}
		\centering
		\subfigure[]{\includegraphics[width=0.50\linewidth]{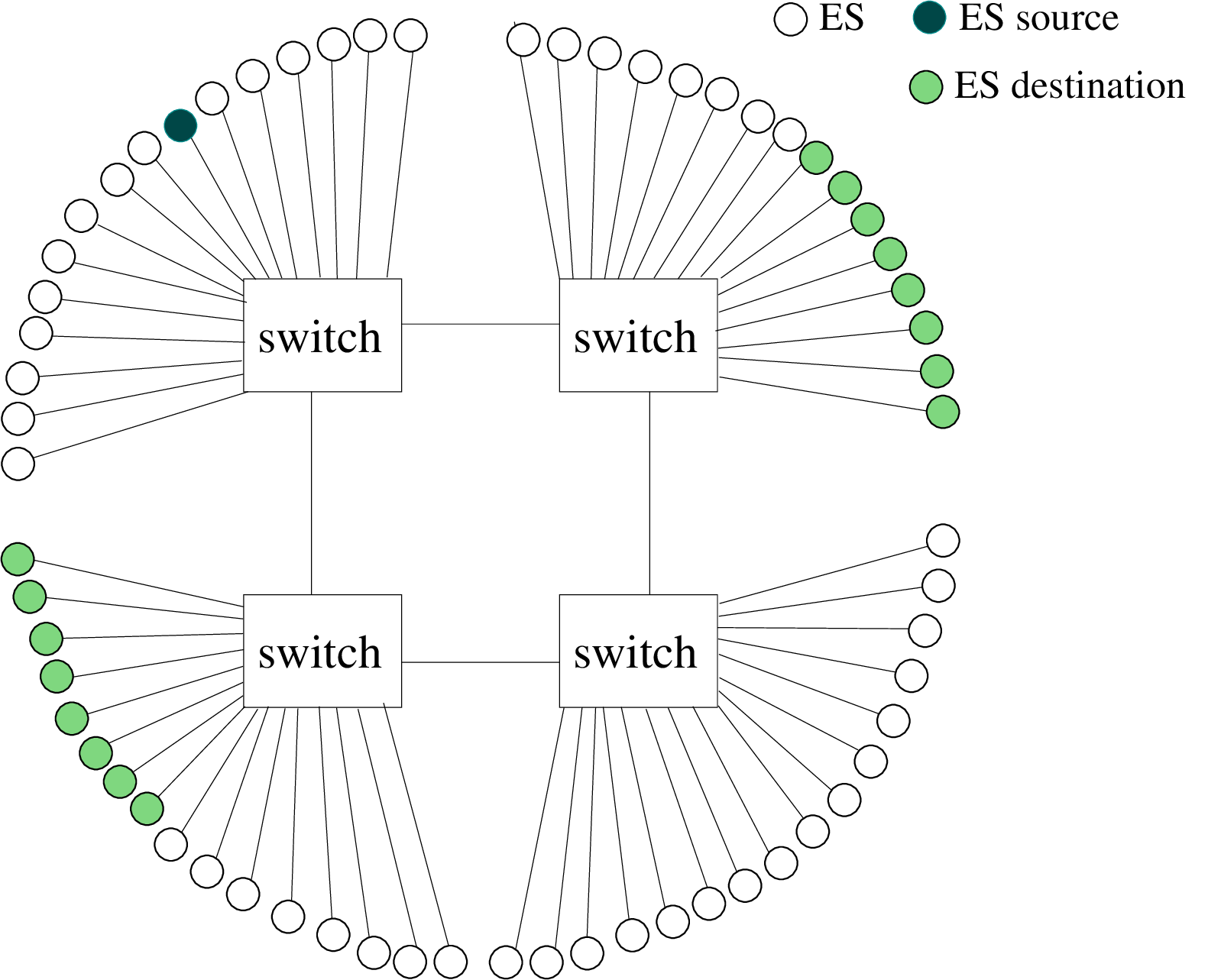}}
		\footnotesize \caption[Representative AFDX network]{Representative AFDX network: (a) Architecture; (b) Traffic communication patterns}
		\label{fig:rlfullnetwork}
	\end{figure}

	Then, we do a timing analysis of this network with different traffic and different BLS parameters. In particular, we show the gains in terms of delay bounds and schedulability when using the extended AFDX instead of current AFDX. As BE traffic does not have timing constraint and only impacts the other flow by a single maximum sized frame transmission, we only present here results for SCT and RC traffic.
	
	\hfill\\
	\textbf{Output port architecture and traffic}
	
	We study an output port architecture with 6 classes: 2 for each type of traffic. As shown in Fig. \ref{fig:arch6}, two BLS are activated: one for the class $SCT_2$ and one for the class $RC_2$. The aim is to give a class for very low SCT deadlines: $SCT_1$. A second class $SCT_2$ regroups the rest of the SCT traffic with larger deadline requirements. Concerning RC, a first class without shaper, denoted $RC_1$ is used for low RC deadlines, and a second class shaped by a BLS is set for larger deadlines. Finally, the BE traffic can also be separated in two classes if necessary.  As the current AFDX standard uses a SP scheduler and defines 2 classes, we compare the  architecture proposed in Fig. \ref{fig:arch6} with a 6-classes architecture scheduled with a simple Static Priority. This way, we can compare the proposed architecture with a natural extension of the current AFDX standard.

	\begin{figure}[htbp]
		\centering
		\includegraphics[width=0.7\textwidth]{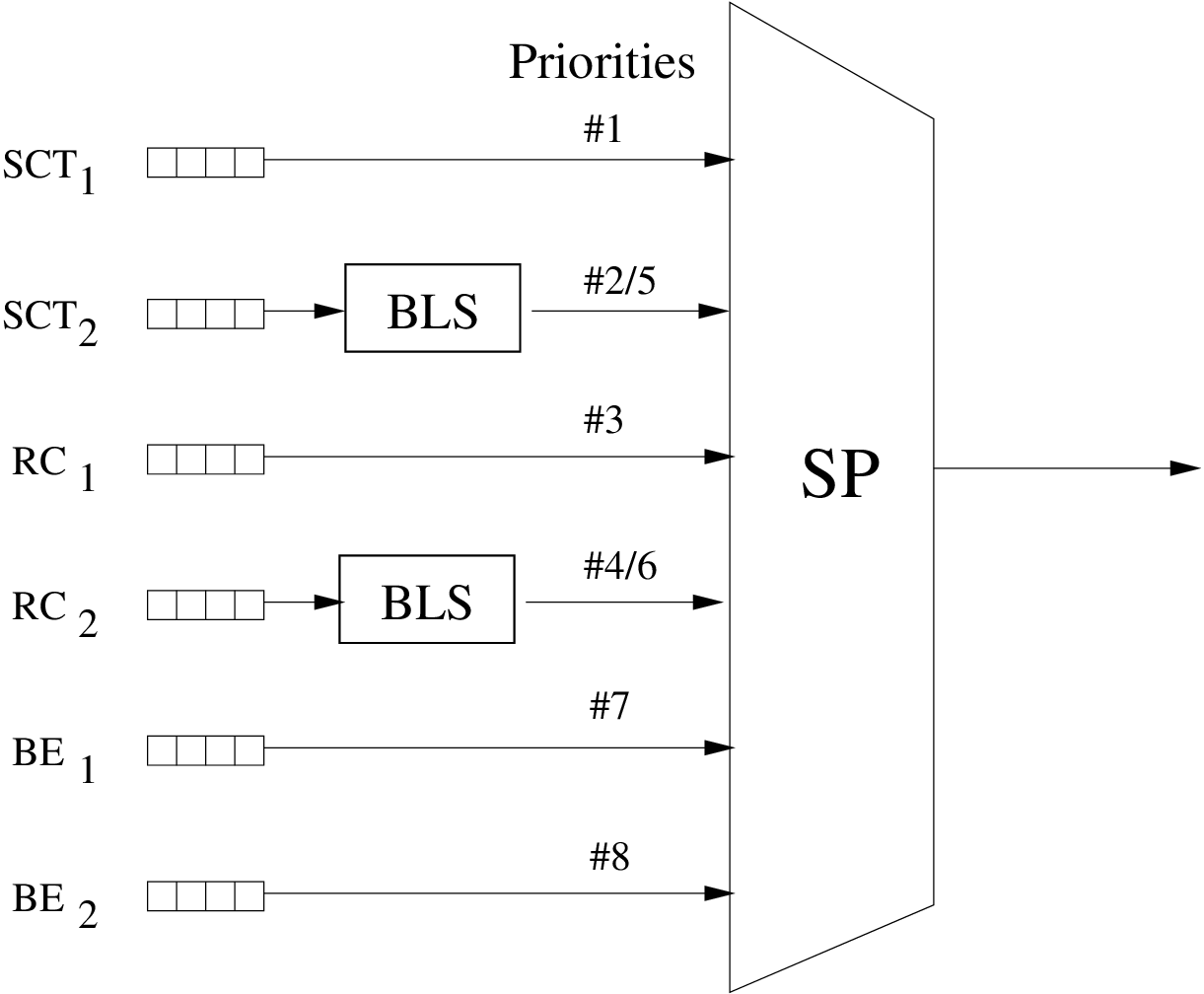}
		\caption{Architecture with 6 priorities and 2 BLS}
		\label{fig:arch6}
	\end{figure}
	The flows we consider for each class of traffic are detailed in Table \ref{table:p6multitrafficprofiles}. We set for each BLS class $k$ $p_L(k)$ higher than the BE priority. Additionally, we kept the order between the high and low priority: $SCT_2$ low priority (i.e. 4) is higher than $RC_2$ low priority (i.e. 5). 	Concerning the classification of the different classes in $HC(k)$, $MC(k)$, and $LC(k)$, here is a practical example for class $SCT_2$: $HC(SCT_2)=\{SCT_1\}$; $MC(SCT_2)=\{RC_1, RC_2\}$; $LC(SCT_2)=\{RC_2, BE_1, BE_2\}$.
	
	A consequence of this BLS configuration is that the impact $BW_{RC_2}$ on $SCT_2$ is only due to $\alpha_{RC_2}^{sp}=\alpha_{RC_2}\oslash\beta_{RC_2}^{bls}$ in $\beta_{SCT_2,5}^{sp}$ in Th.\ref{Th:blssp}. Hence, the impact of $BW_{RC_2}$ is null for the non-BLS classes, i.e., $SCT_1$ and $RC_1$, and not noticeable for $SCT_2$.
	
	\begin{table}[h!]
		\footnotesize
		\centering
		\begin{tabular}{|c|c|c|c|c|c|c|}
			\hline
			Name &	Priorities & Traffic  & MFS & BAG & deadline& jitter  \\
			& & type  & (Bytes)          & (ms)&(ms) &(ms)\\
			\hline
			$SCT_1$& 1 & SCT & 64 & 2 & 1 & 0 \\  
			\hline
			$SCT_2$ & 2/5 &  SCT & 128 & 4 & 4 & 0\\
			\hline
			$RC_1$ & 3 &	RC & 256 & 4 & 4 & 0 \\ 
			\hline
			$RC_2$& 4/6 & RC & 512 & 8 & 8 & 0 \\  
			\hline
			$BE_1$& 7 &  BE & 1024 & 2 & none & 0.5\\
			\hline
			$BE_2$ & 8 & BE & 1500 & 8 & none & 0.5 \\ 
			\hline
		\end{tabular}
		\footnotesize \caption{6 priorities: Avionics flow Characteristics}
		\label{table:p6multitrafficprofiles}
	\end{table}	
	
	Each ES generates $n_{k}^{es}$ flows ($k$ the class of the flow). So the bottleneck utilization rate for a class $k$ in the network presented in Fig. \ref{fig:rlfullnetwork} is $UR^{bn}_{k}=16\cdot n_{k}^{es}\cdot \frac{MFS_k}{BAG_k}$.
	
	In Table \ref{table34}, we present the different scenarios considered: in $Scenario_{SCT_2}$, we vary the utilization rate of the $SCT_2$ class;in $Scenario_{RC_1}$, we vary the utilization rate of the $RC_1$ class;in  $Scenario_{RC_2}$, we vary the utilization rate of the $RC_2$ class.
	We call the $default configuration$ the constant rates used in the different scenarios: 5\% for $SCT_1$, $RC_1$ and $RC_2$, 25\% for $SCT_2$, 15\% for $BE_1$ and 20\% for $BE_2$. We consider that the SCT classes represent 30\% of the traffic. Concerning RC classes, which are used by the traffic currently on the AFDX, we considered a future where the current traffic has tripled, going from a maximum rate of 30~Mbps to 100~Mbps for the RC type of traffic. Finally, the BE classes represent 35\% of the traffic.

	\begin{table}[h!]
		\footnotesize
		\begin{center}
			\begin{tabular}{|c|c|c|c|}
				\hline
				Scenarios &  $Scenario_{SCT_2}$ & $Scenario_{RC_1}$ & $Scenario_{RC_2}$  \\
				\hline
				$UR^{bn}_{SCT_1} (\%)$ & $5$ & $5$& $5$\\
				\hline
				$UR^{bn}_{SCT_2} (\%)$ & $[0.4..50]$ & $25$&$25$\\
				\hline
				$UR^{bn}_{RC_1} (\%)$ & $5$ & $[0.8..30 ]$ & $5$\\
				\hline
				$UR^{bn}_{RC_2} (\%)$ & $5$ & $5$&$[0.8..30 ]$\\
				\hline
				$UR^{bn}_{BE_1} (\%)$ & $15$ & $15$&$15$\\
				\hline
				$UR^{bn}_{BE_2} (\%)$ & $20$ & $20$ & $20$\\
				\hline
				\hline
				$n_{SCT_1}^{es}$&  $12$ & $12$ & $12$ \\
				\hline
				$n_{SCT_2}^{es}$&  $[1:5:122]$ & $36$ &  $36$  \\
				\hline
				$n_{RC_1}^{es}$&  $6$ &$[1:2:36]$ & $6$ \\
				\hline
				$n_{RC_2}^{es}$&  $6$ & $6$ &  $[1:2:36]$  \\
				\hline
				$n_{BE_1}^{es}$& $18$ & $18$ &  $18$  \\
				\hline
				$n_{BE_2}^{es}$&  $33$ &  $33$&   $33$  \\
				\hline
				
			\end{tabular}
		\end{center}
		\footnotesize \caption{Considered Test Scenarios with 6 classes}
		\label{table34}
	\end{table}
	
	The BLS parameters are identical is every output ports. For a class $k$, We set the $L_R^k$ as described in Section \ref{Use-case1},$L_R^k=\max_{j\in MC(k)}MFS_{j}\cdot BW^k$. Also following the conclusion of Section \ref{Use-case1},  we set a low value for $L_M^k$: $L_M^{k}=5\cdot MFS_k$. Finally, we saw in Section \ref{Use-case1} that $BW^k$ is the BLS parameters that influences most the delay bounds, so we will vary this parameter from $0.1$ to $0.9$ in the 3 scenarios.

	We chose not to present results of $SCT_1$  variations as their conclusions only confirm previous scenarios' conclusions. As $SCT_1$ traffic is not impacted by the variations of other class traffic, the $SCT_1$ delay bound is constant for all scenarios and equal to $2.45\cdot 10^{-4}$s, well below its 1ms deadline.
	
	\hfill\\
	\textbf{Impact of the variation of $SCT_2$ traffic}
	
	\label{ImpactSCT2}
	In this 1st scenario, we vary $SCT_2$ traffic and the reserved bandwidth of $SCT_2$, and $RC_2$, respectively denoted $BW^{SCT_2}$ and  $BW^{RC_2}$. The results are presented in Fig. \ref{fig:P2Impact}.	As the variation of $BW^{RC_2}$ does not impact the delay bounds of $SCT_2$ and $RC_1$ traffic, we present in Fig. \ref{fig:P2Impact}(a)(b) the impact of the variation of $BW^{SCT_2}$ on $SCT_2$ and $RC_1$. In Fig. \ref{fig:P2Impact}(c)(d) we present the impact of the variations of $BW^{SCT_2}$ and  $BW^{RC_2}$ on $RC_2$.

	\hfill\\
	\textit{$SCT_2$ delay bounds}
	
	First in Fig. \ref{fig:P2Impact}(a), as expected the BLS increases the delay bounds of the $SCT_2$ traffic compared to SP. This increase is minimal pour high values of $BW^{SCT_2}$, i.e., when the behavior is close to the SP behavior; and maximal  for low values of  $BW^{SCT_2}$, i.e., when the BLS is most effective. For $BW^{SCT_2}=\{0.10,0.50\}$, we are able to distinguish the two parts of the service offered, e.g., the BLS parts $\beta_{SCT_2}^{bls}$ for low values of $UR_{SCT_2}$, and the SP part $\beta_{SCT_2,5}^{sp}$ (identical for any values of $BW^{SCT_2}$)for higher values of $UR_{SCT_2}$. With $BW^{SCT_2}=\{0.74,0.90\}$, the SP part is not reached before $UR_{SCT_2}=50\%$.
	
	\hfill\\
	\textit{$RC_1$ delay bounds}
	
	Secondly in Fig. \ref{fig:P2Impact}(b), we find again a BLS behavior similar to Section~\ref{Use-case1}: the BLS limits the $RC_1$ delay bound. For low values of $UR_{SCT_2}$, the $RC_1$ under the extended AFDX (BLS) is ruled by $\beta_{RC_1}^{sp}$ and increases steadily. Then, the $RC_1$ delay bound is limited by $\beta_{RC_1}^{bls}$. The $RC_1$ delay bound limitation occurs at different $UR_{SCT}$  depending on the value of $BW^{SCT_2}$, i.e., it increases with $BW^{SCT_2}$. As a result, the gain in terms of $RC_1$ delay bound, decreases when $BW^{SCT_2}$ increases, compared to SP. For example, at $UR_{SCT_2}=25\%$ the $RC_1$ delay bounds is divided by 2.4 for $BW^{SCT_2}=0.5$, and by 5.25 for $BW^{SCT_2}=0.1$.

	\hfill\\
	\textit{$RC_2$ delay bounds}
	
	Thirdly, when varying $BW^{SCT_2}$ in Fig. \ref{fig:P2Impact}(c), the behavior of the $RC_2$ delay bound is similar to $RC_1$ delay bound, i.e., after a certain point, the $RC_2$ delay bound is limited by the BLS part $\beta_{RC_2}^{bls}$. Hence, the impact of the increase of $UR_{SCT_2}$ is mitigated by the BLS compare to SP. For example, at $UR_{SCT_2}=25\%$ the $RC_2$ delay bounds is decreases by 8.3\% for $BW^{SCT_2}=0.5$, and by 13.3\% for $BW^{SCT_2}=0.1$. The $RC_2$ delay bounds  can be decreased by up to 33\%.
	
	Fourthly, when varying $BW^{RC_2}$ in Fig. \ref{fig:P2Impact}(d),the impact of higher priorities on $RC_2$ is again mitigated by the extended AFDX (BLS). For low values of $BW^{RC_2}$, the behavior is ruled by $\beta_{RC_2,6}^{sp}$ and so the $RC_2$ delay bounds under BLS are close to those under the current AFDX (SP). For higher values of $BW^{RC_2}$ and $UR_{SCT_2}>18\%$, the behavior under BLS is ruled by $\beta_{RC_2}^{bls}$ and the $RC_2$ delay bound increases are mitigated with a $RC_2$ delay bound decrease up to 50\% compared to the delay bounds under the current AFDX (SP).

	\hfill\\
	\textit{$SCT_2$ schedulability}
	
	Finally, the extended AFDX (BLS) has a strong impact on the schedulability (when all the deadlines are fulfilled) of the $SCT_2$ class.  The maximum $UR_{SCT_2}$ are presented in Table~\ref{scheduSCT2}. The results show that the extended AFDX increases the schedulability compared to a standard SP. In particular, at  $BW^{SCT_2}=0.74$, i.e.,$RC_1$ delay bound close to its deadline, the $SCT_2$ schedulability is increased by 40\%.

	\begin{table}[h!]
		\footnotesize
		\begin{center}
			\begin{tabular}{|c|c|}
				\hline
				 $Scenario_{SCT_2}$ & maximum $UR_{SCT_2}$ (\% C) \\
				\hline
			     SP & 23 \\
			     \hline
			     BLS, $BW^{SCT_2}=0.90$,  $BW^{RC_2}=0.50$ & 23\\
				\hline
				 BLS, $BW^{SCT_2}=0.10$,  $BW^{RC_2}=0.50$ & 25\\
 				\hline
				BLS, $BW^{SCT_2}=0.50$,  $BW^{RC_2}=0.50$ & 25\\	
				\hline
				BLS, $BW^{SCT_2}=0.74$,  $BW^{RC_2}=0.50$ & 32\\
				\hline									
			\end{tabular}
		\end{center}
		\footnotesize \caption{$SCT_{2}$ schedulability limits}
		\label{scheduSCT2}
	\end{table}

	\textbf{These results show the impact of the extended AFDX (BLS) when $SCT_2$ varies: with good parameters, the delay bounds of lower priorities can be divided by up to 5.4 times for $RC_1$, and up to 2 for $RC_2$ compared to current AFDX (SP). The schedulability of $SCT_2$ is also enhanced, up to 40\%.}
	
	\begin{figure}[htbp]
		\centering
		\subfigure[Varying $BW^{SCT_2}$]{\includegraphics[width=0.4\textwidth, angle=270]{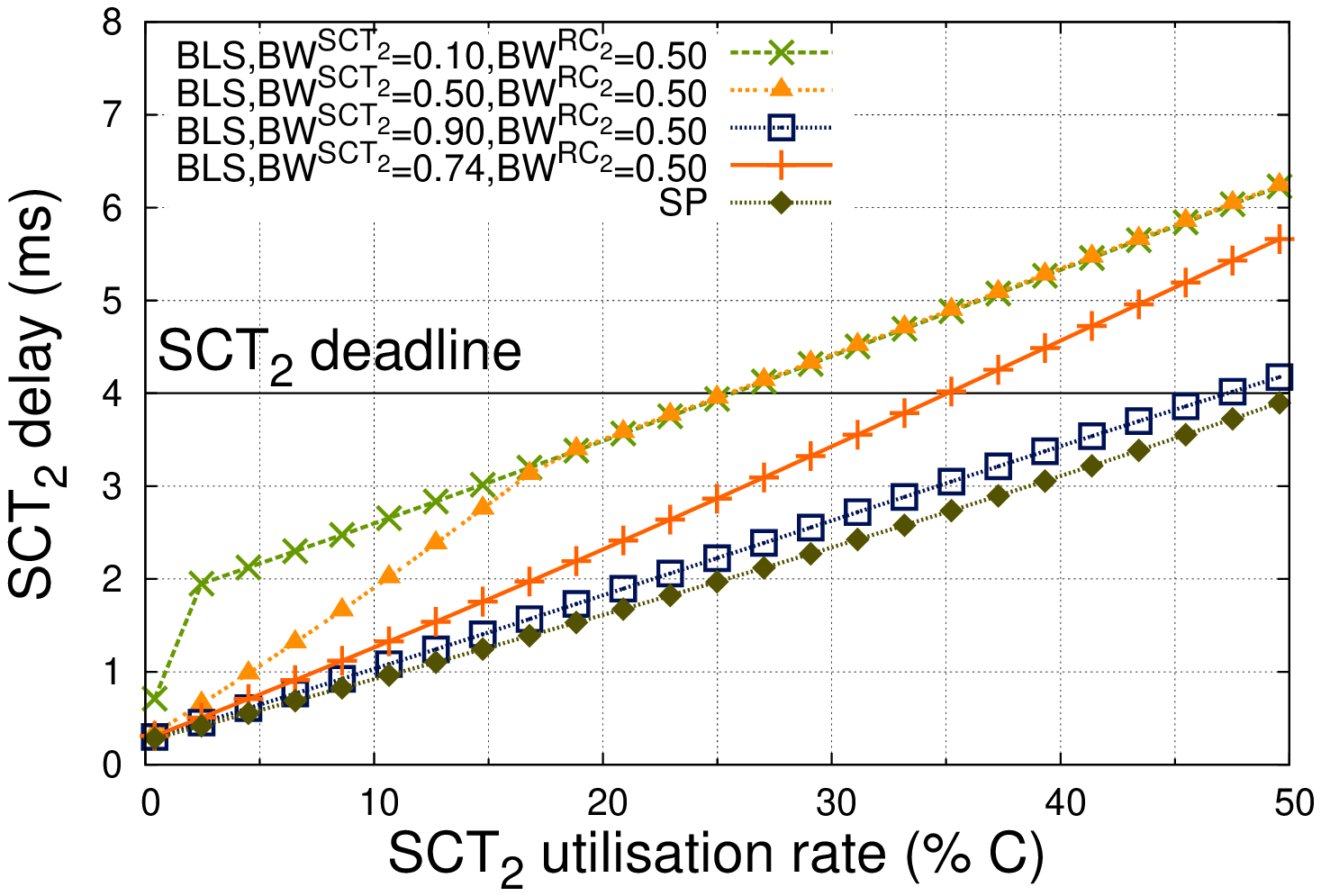}}
		\centering
		\subfigure[Varying $BW^{SCT_2}$]{\includegraphics[width=0.4\textwidth, angle=270]{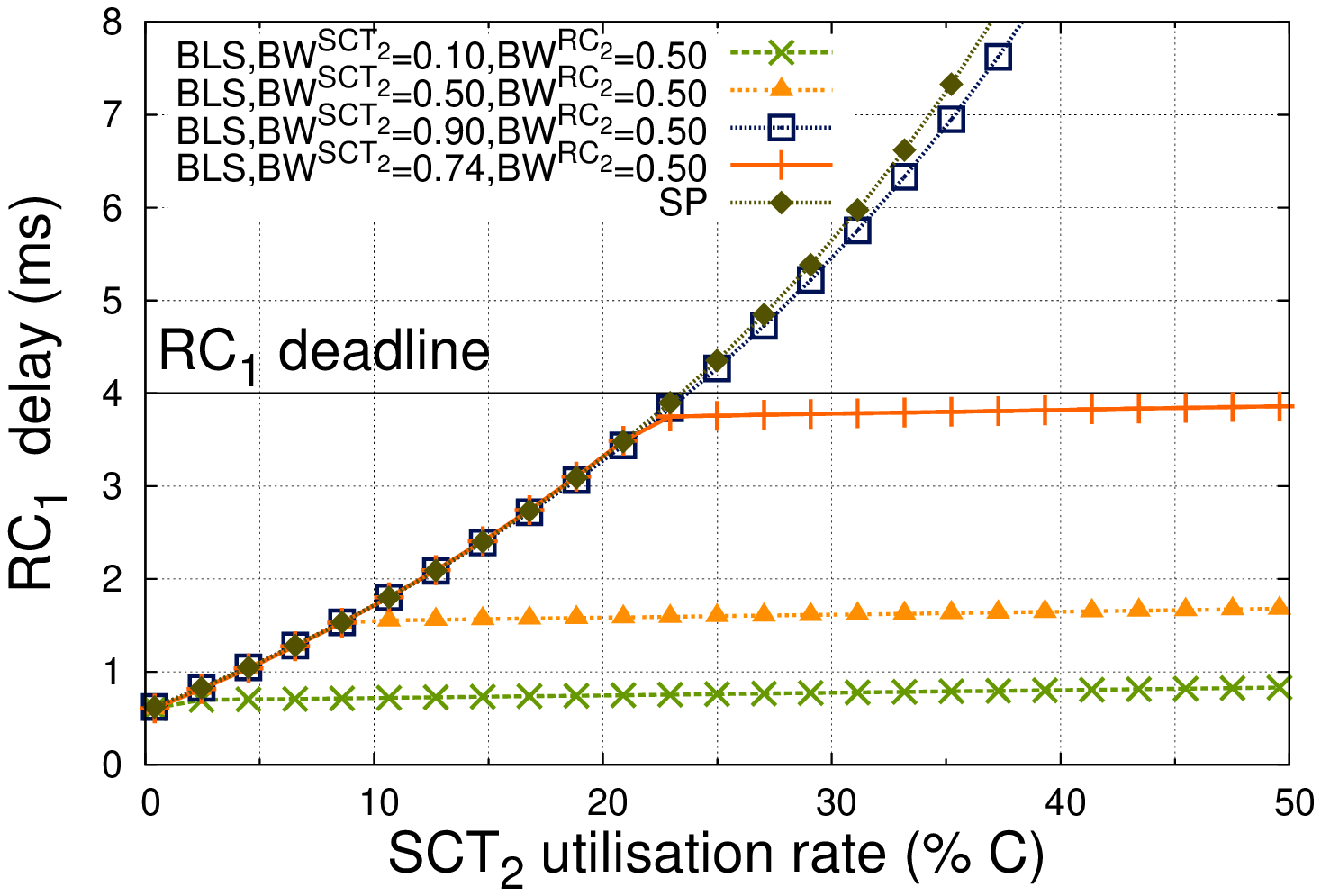}}
		\centering
		\subfigure[Varying $BW^{SCT_2}$]{\includegraphics[width=0.4\textwidth, angle=270]{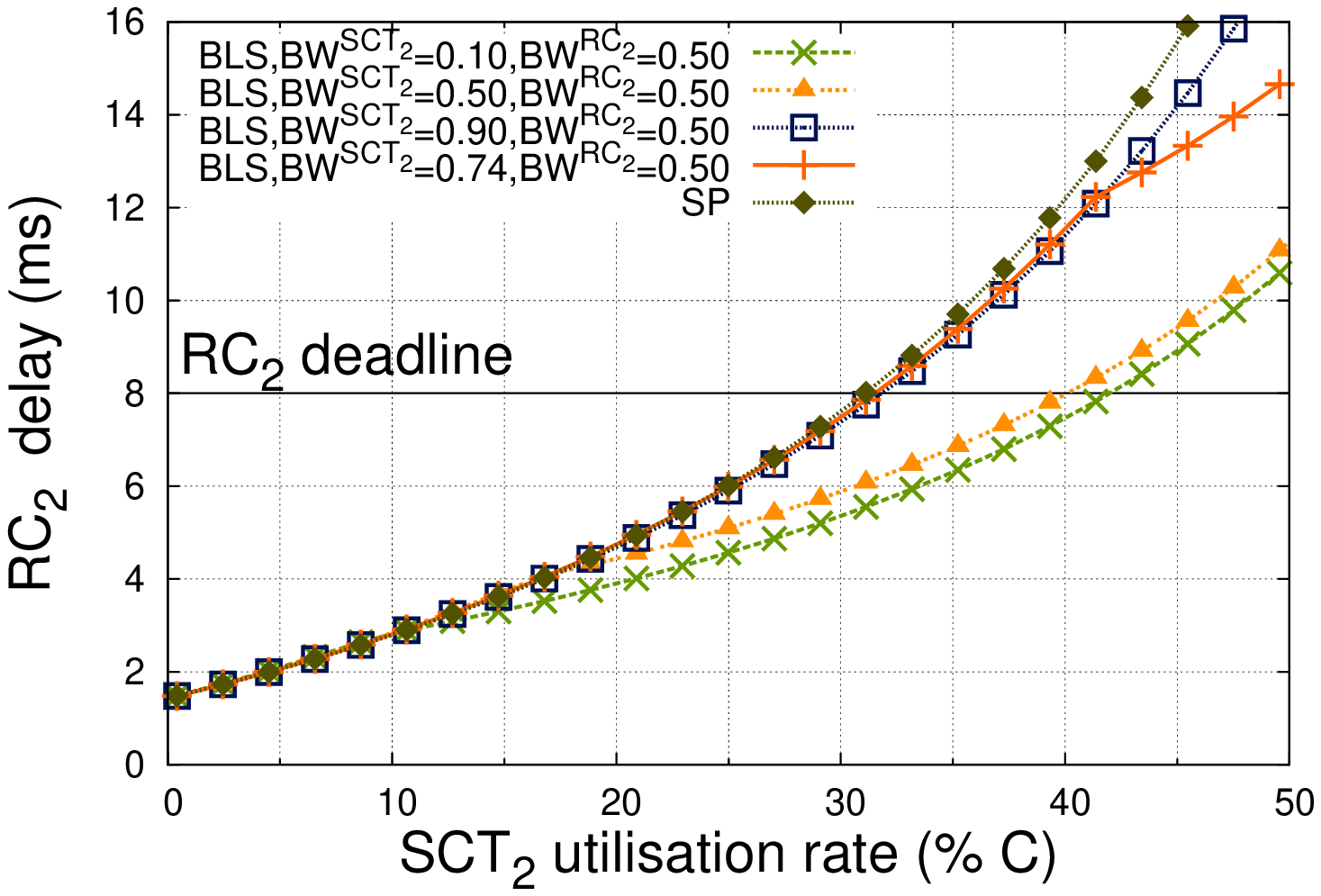}}
		\centering
		\subfigure[Varying $BW^{RC_2}$]{\includegraphics[width=0.4\textwidth, angle=270]{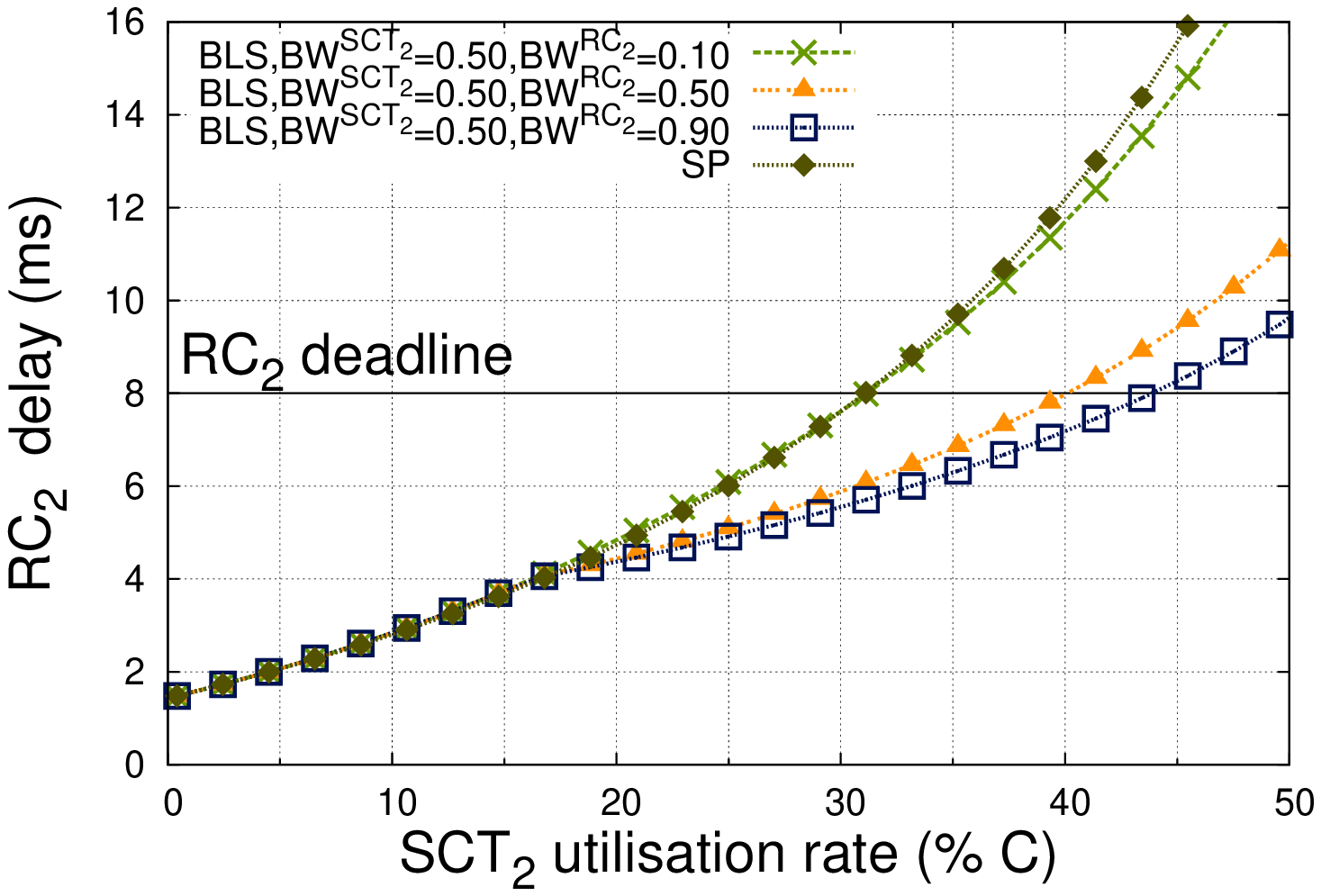}}
		
		\footnotesize \caption{Impact of $SCT_2$ utilisation rate on: (a) $SCT_2$ delays; (b) $RC_1$ delays; (c) and (d) $RC_2$ delays}
		\label{fig:P2Impact}
	\end{figure}

\newpage
\noindent
\textbf{Impact of the variation of $RC_1$ traffic}

	\label{ImpactRC1}
	In this 2nd scenario, we vary $RC_1$ traffic and the reserved bandwidth of $RC_2$,  denoted $BW^{RC_2}$. The results are presented in Fig. \ref{fig:P3Impact}.	As we have extensively studied the impact of variation of $BW^{SCT_2}$, we will not vary $BW^{SCT_2}$ here. Also, as the variation of $BW^{RC_2}$ does not impact the delay bounds of $SCT_2$ and $RC_1$ traffic, we present in Fig. \ref{fig:P3Impact}(a) and (b) the impact of the variation of $RC_1$ utilization rate  on $SCT_2$ and $RC_1$ delays for any $BW^{RC_2}$. In Fig. \ref{fig:P3Impact}(c) we present the impact of the variations of  $RC_1$ utilization rate and $BW^{RC_2}$ on $RC_2$ delays.	
	In this scenario, we set $BW^{SCT_2}=0.58$ so that $SCT_2$ delay bound is just below its 4ms deadline in Fig. \ref{fig:P3Impact}(a).

	\hfill\\
	\textit{$RC_1$ delay bounds}
	
	In Fig. \ref{fig:P3Impact}(b), the $RC_1$ delay bound is largely decreased by the extended AFDX (BLS). At the current utilization rate of the AFDX, i.e., $UR_{RC_1}=3\%$, the $RC_1$ delay bound under the extended AFDX (BLS) is divided by 2.7, compared to the bounds under the current AFDX (SP).
	
	\hfill\\
	\textit{$RC_2$ delay bounds}
	
	In Fig. \ref{fig:P3Impact}(c), the impact of the extended AFDX (BLS) on the $RC_2$ delay bounds is less visible. At the current utilization rate of the AFDX, i.e., $UR_{RC_1}=3\%$, the $RC_2$ delay bound under the extended AFDX (BLS) is decreased up to 17\%, compared to the bounds under SP. The maximum decrease of the $RC_2$ delay bound under extended AFDX (BLS) is 20\% at $UR_{RC_1}=1\%$, compared to the $RC_2$ delay bound under the current AFDX (SP).
		
	\hfill\\
	\textit{$RC_1$ schedulability}
	
	The schedulability of $RC_1$ is largely enhanced by the extended AFDX (BLS) compared to the current AFDX (SP). The maximum $RC_1$ utilization rate under SP is limited at $UR_{RC_1}=2.5\%$ in Fig. \ref{fig:P3Impact}(b) by the $RC_1$ deadline. With the extended AFDX  (BLS), the $RC_2$ deadline limits the maximum $RC_1$ utilization rate at $UR_{RC_1}=12\%$. Hence, the $RC_1$ schedulability is multiplied by 4.8 with the extended AFDX compared to the current AFDX (SP).

	\textbf{These results show the impact of the extended AFDX (BLS) when $RC_1$ varies: the delay bounds of lower priorities can be divided by 2.7 for $RC_1$ and decreased by up to 20\% for $RC_2$, compared to the current AFDX (SP).}
	
	\begin{figure}[htbp]
		\centering
		\subfigure[for any $BW^{RC_2}$]{\includegraphics[width=0.33\textwidth, angle=270]{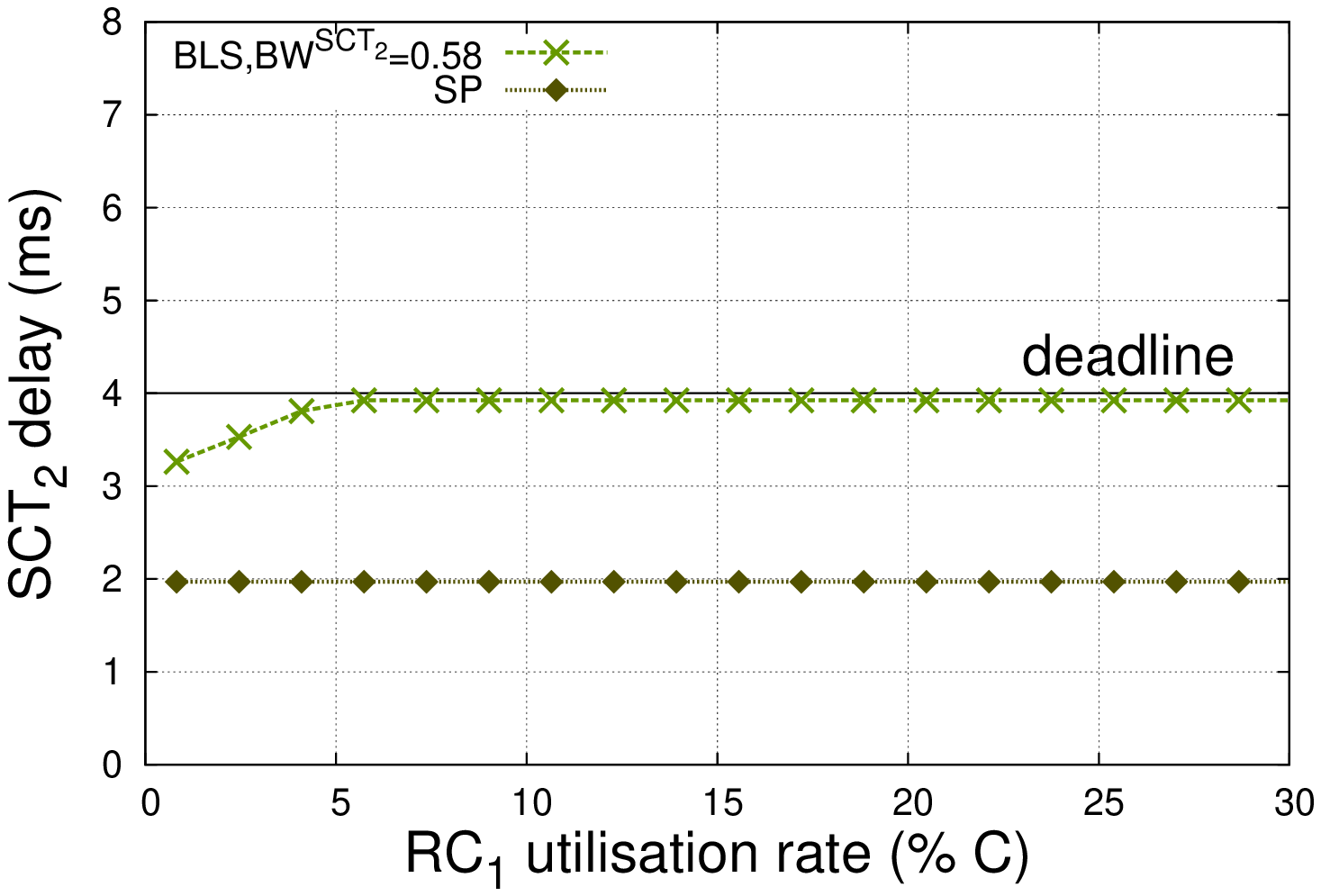}}
		\subfigure[for any $BW^{RC_2}$]{\includegraphics[width=0.33\textwidth, angle=270]{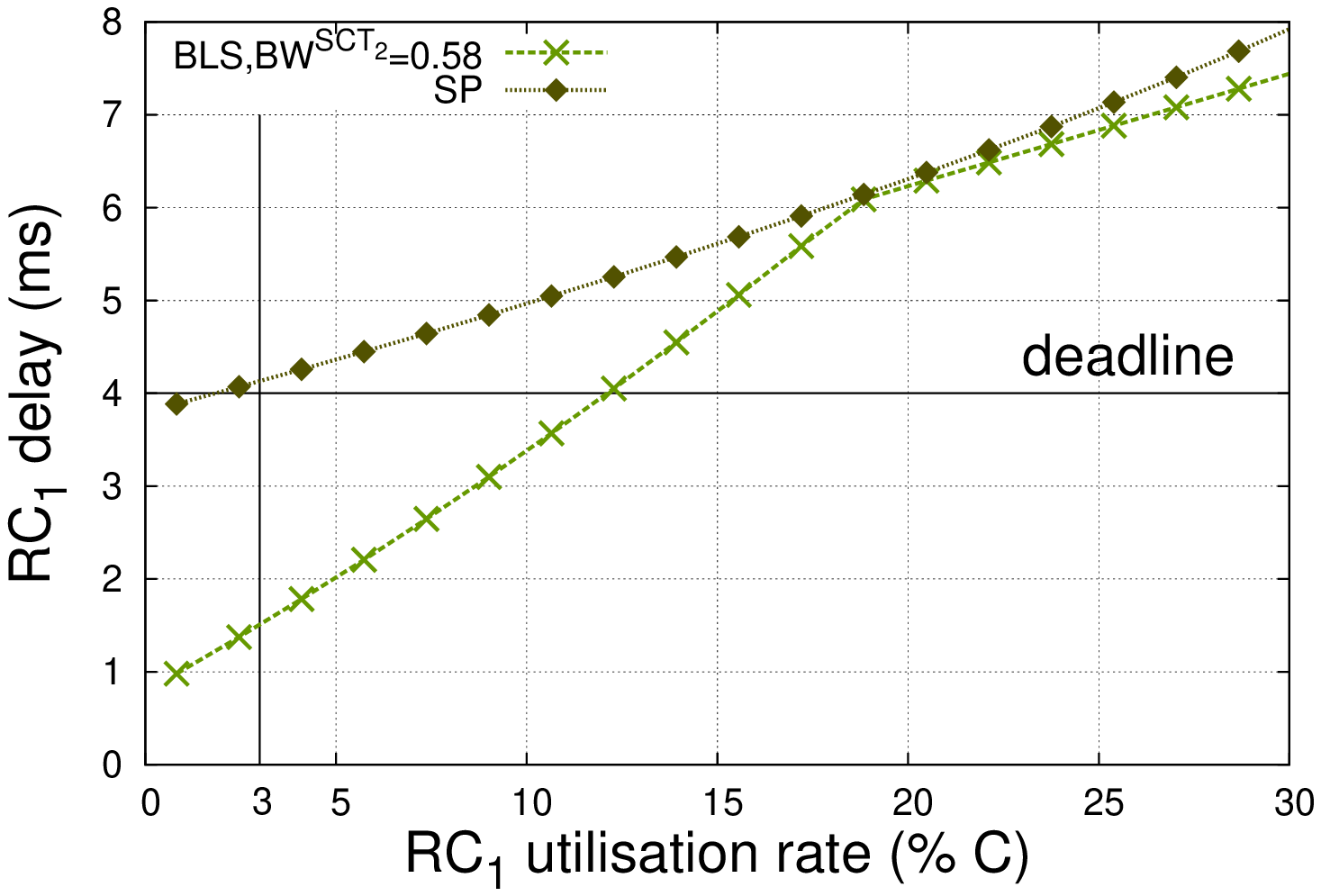}}
		\centering
		\subfigure[]{\includegraphics[width=0.4\textwidth, angle=270]{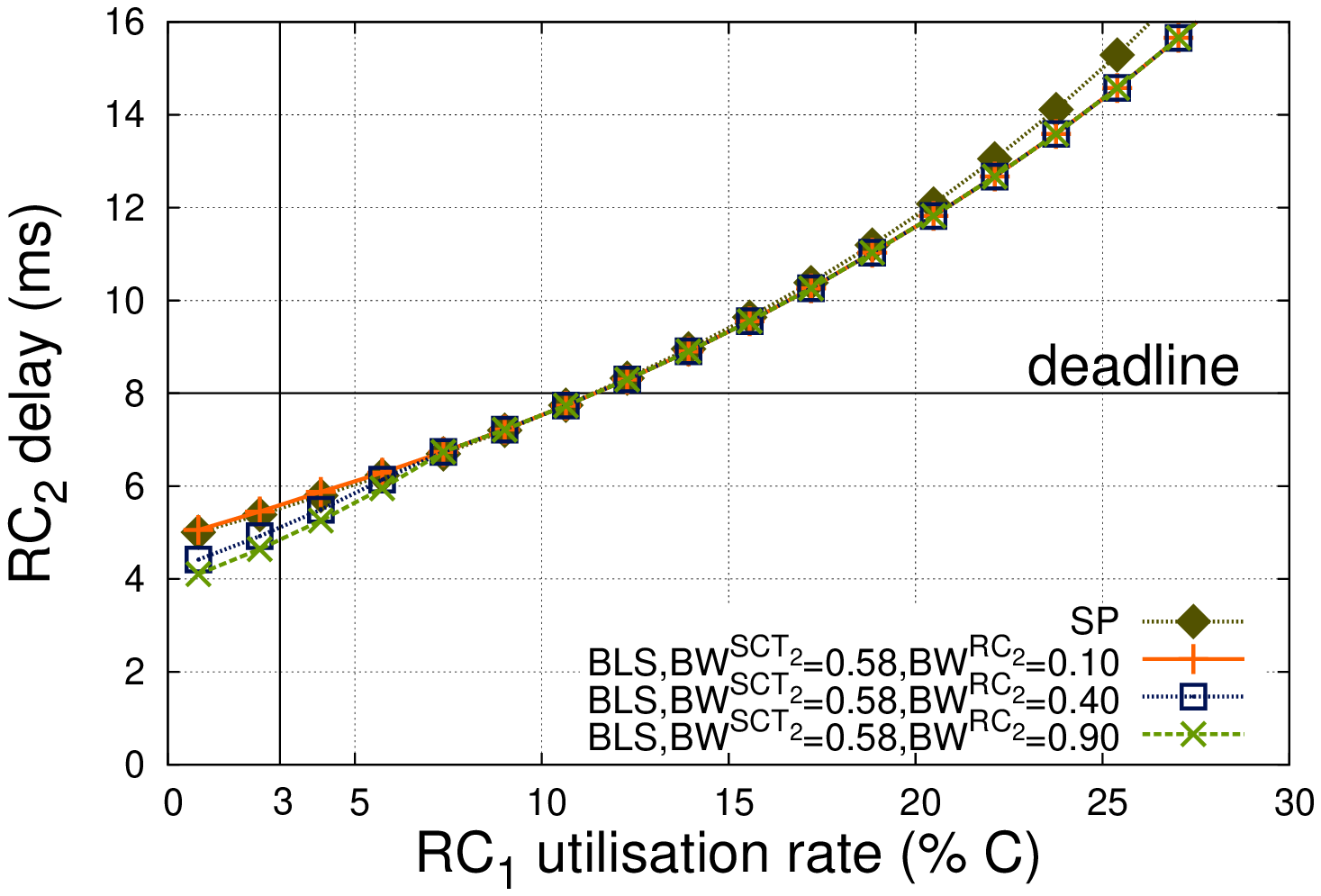}}
		
		\footnotesize \caption{Impact of $RC_1$ utilisation rate on: (a) $SCT_2$ delays; (b) $RC_1$ delays; (c)  $RC_2$ delays}
		\label{fig:P3Impact}
	\end{figure}
	
	\hfill\\
	\textbf{Impact of the variation of $RC_2$ traffic}	
	
	In this 3rd scenario, we vary $RC_2$ traffic and the reserved bandwidth of $RC_2$,  denoted $BW^{RC_2}$. The results are presented in Fig. \ref{fig:P4Impact}. As the variation of $BW^{RC_2}$ does not impact the delay bounds of $SCT_2$ and $RC_1$ traffic, we present in Fig. \ref{fig:P4Impact}(a) and (b) the impact of the variation of $RC_1$ on $SCT_2$ and $RC_1$ for any $BW^{RC_2}$. In Fig. \ref{fig:P2Impact}(b) we present the impact of the variations of $BW^{RC_2}$ on $RC_2$.	
	In this scenario, as before we set $BW^{SCT_2}=0.58$ so that $SCT_2$ delay bound is just below its 4ms deadline in Fig. \ref{fig:P4Impact}(a).

	\newpage
	\noindent
	\textit{$RC_1$ delay bounds}
	
	In Fig. \ref{fig:P4Impact}(b), the $RC_1$ delay bound is much decreased under the extended AFDX compared to under the current AFDX (SP), i.e., it is divided by 2.1.
	
	\hfill\\
	\textit{$RC_2$ delay bounds}
		
	In Fig. \ref{fig:P4Impact}(c), the $RC_2$ delay bound is positively impacted by the extended AFDX, with a maximum decreased of the $RC_2$ delay bounds of 22\%.  For example, at the current utilization rate of the AFDX, i.e., $UR_{RC_2}=3\%$, the $RC_2$ delay bound under the extended AFDX (BLS) is decreased by 16\% compared to the delay bounds under the current AFDX (SP).

	\hfill\\
	\textit{$RC_2$ schedulability}
	
	As before the extended AFDX (BLS) increases the schedulability compared to the current AFDX (SP).
	In fact in In Fig. \ref{fig:P4Impact}(b), with SP, the $RC_1$ delay bound is always greater than its deadline. On the contrary, with extended AFDX the $RC_1$ deadline is always fulfilled and the limitation to the schedulability is due to the $RC_2$ deadline, which is crossed at $UR_{RC_2}=12\%$.

	\textbf{These results confirm the positive impact of the extended AFDX (BLS) when $RC_2$ varies, i.e., the $RC_1$ delay bounds is divided by 2.1 and the $RC_2$ is decreased up to 22\% compared to the current AFDX (SP). Additionally, the schedulability of $RC_2$ is enhanced from $UR_{RC_2}=0$ under the current AFDX (SP), to $UR_{RC_2}=12\%$ under the extended AFDX (BLS)}.
	
	\begin{figure}[htbp]
		\centering
		\subfigure[for any $BW^{RC_2}$]{\includegraphics[width=0.33\textwidth, angle=270]{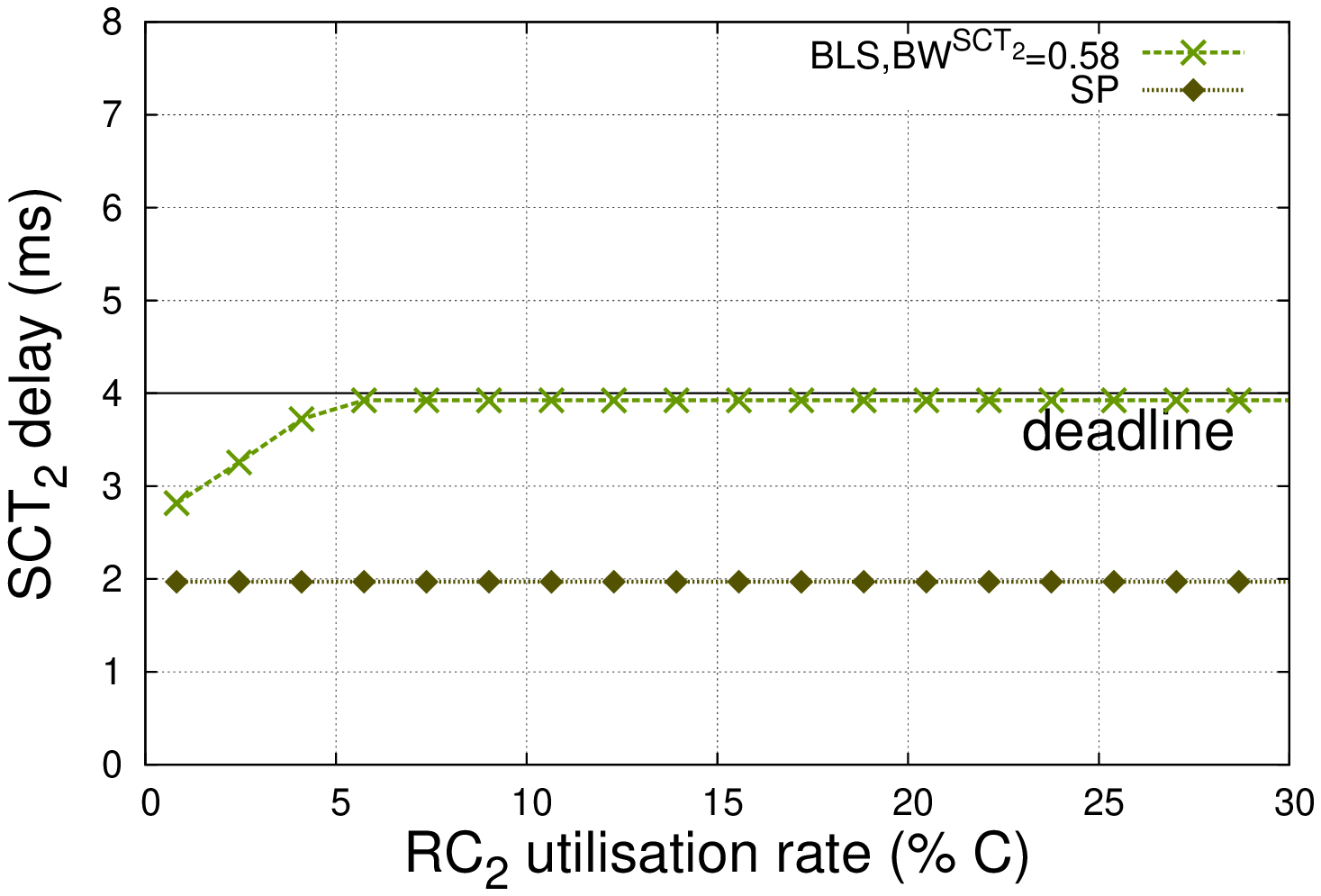}}
		\subfigure[for any $BW^{RC_2}$]{\includegraphics[width=0.33\textwidth, angle=270]{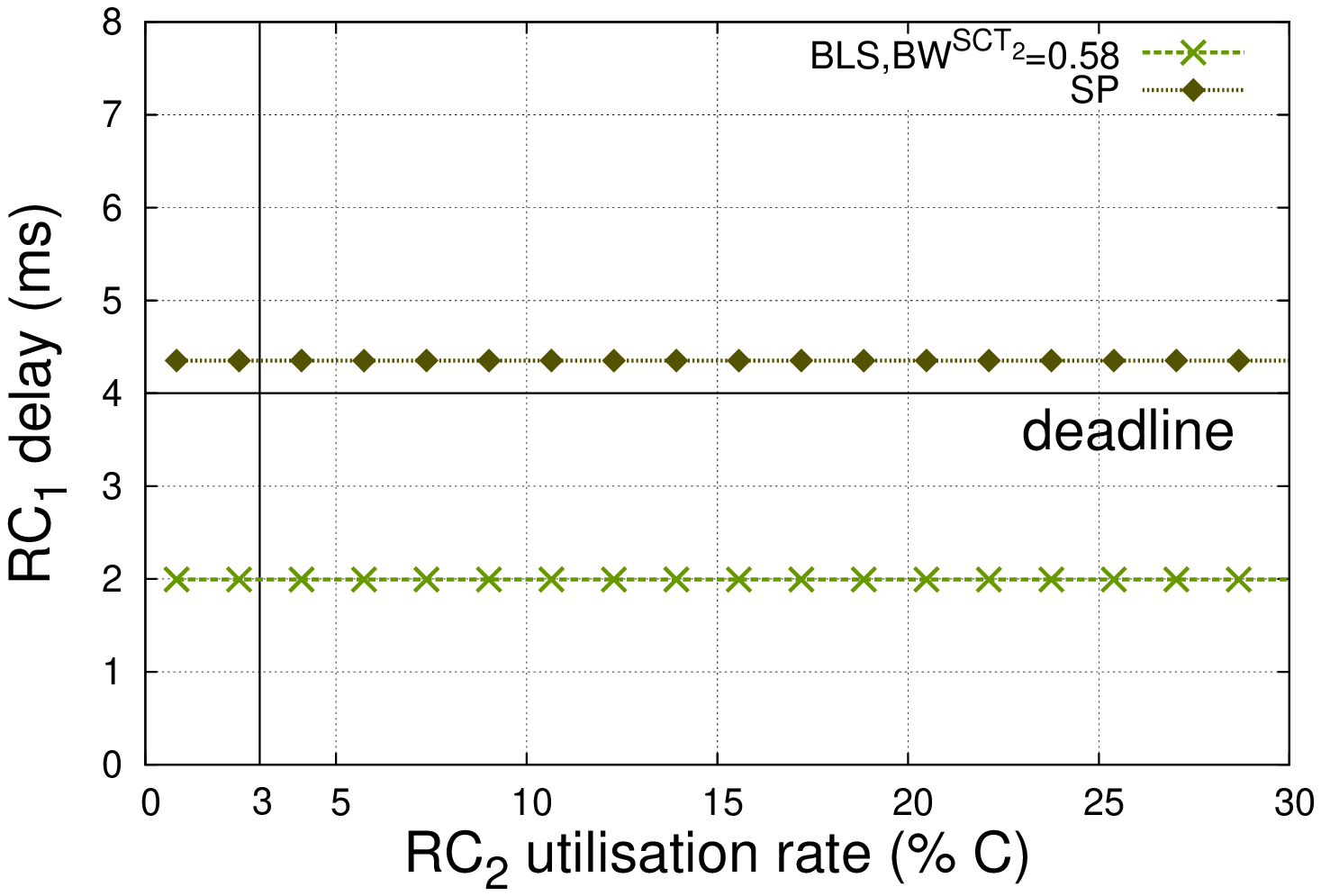}}
		\centering
		\subfigure[]{\includegraphics[width=0.4\textwidth, angle=270]{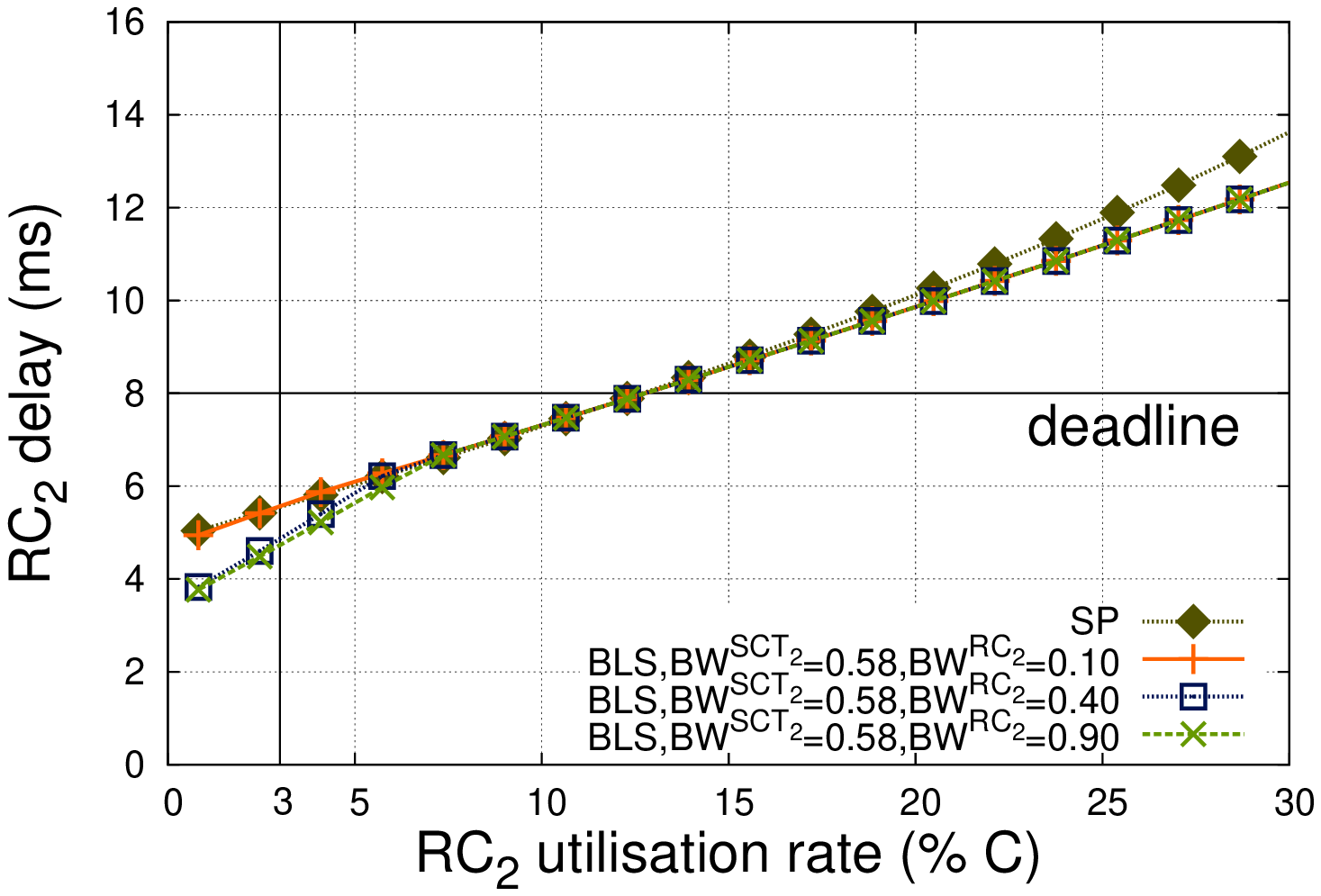}}
		
		\footnotesize \caption{Impact of $RC_2$ utilisation rate on: (a) $SCT_2$ delays; (b) $RC_1$ delays; (c) $RC_2$ delays}
		\label{fig:P4Impact}
	\end{figure}

\subsection{Use-case 3: adding A350 Flight Control to the AFDX}

The flight control traffic is the most important on an aircraft. It is currently on private MIL-STD-1553B networks to keep it isolated from other traffics. In this part, we study the possibility of adding the Flight Control traffic to the AFDX using our proposed 3-classes extended AFDX illustrated in \ref{fig:BLSshaper2}.

\hfill\\
\textbf{Defining a new network architecture}

The AFDX fulfills  the highest avionics requirements and can be used for Safety Critical traffic, such as flight Control traffic. However, the main concern is the way the Flight Control devices can be connected to the AFDX to guarantee the avionics requirements, particularly the safety rule stating that a single failure must not cause the loss of a function. After a reverse-engineering process of the current A350 flight control architecture illustrated in Fig. \ref{A350}, we have proposed a new architecture connecting the flight control calculators and actuators based on the extended AFDX technology, as illustrated in Fig. \ref{A350new}.

\begin{figure}[htbp]
	\centering
	\includegraphics[width=0.9\textwidth]{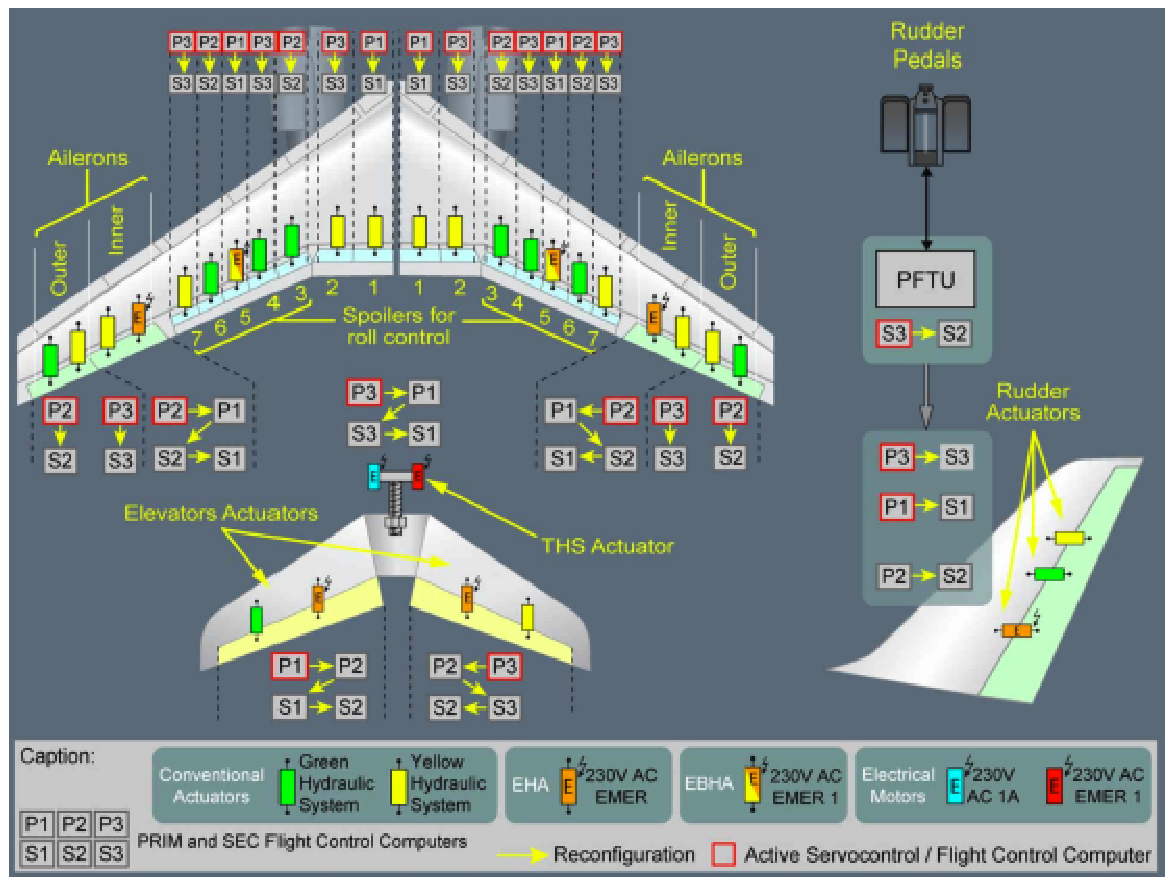}
	\footnotesize \caption[A350 flight control architecture]{A350 flight control architecture, image downloaded from https://www.quora.com/How-are-the-Airbus-A350-and-A330-different}
	\label{A350}
\end{figure}

First, to connect the calculators, we use 2 switches (SS1 and SS2), so that we do not lose a full network due to one switch loss. To obtain similar timing results for all actuators, we use the remaining 5 switches (L1, L2, C, R1, R2) to connect the actuators, taking into account the energy supplier network, the type of actuator to fulfill the safety rule, and trying to homogenize the number of connections on each switch.
The result is visible in Fig. \ref{A350new}, we obtain a diamond-like structure with central switches (L1, L2, C, R1, R2) connected to 6 or 7 actuators.

\begin{figure}[htbp]
	\centering
	\includegraphics[width=0.88\textwidth]{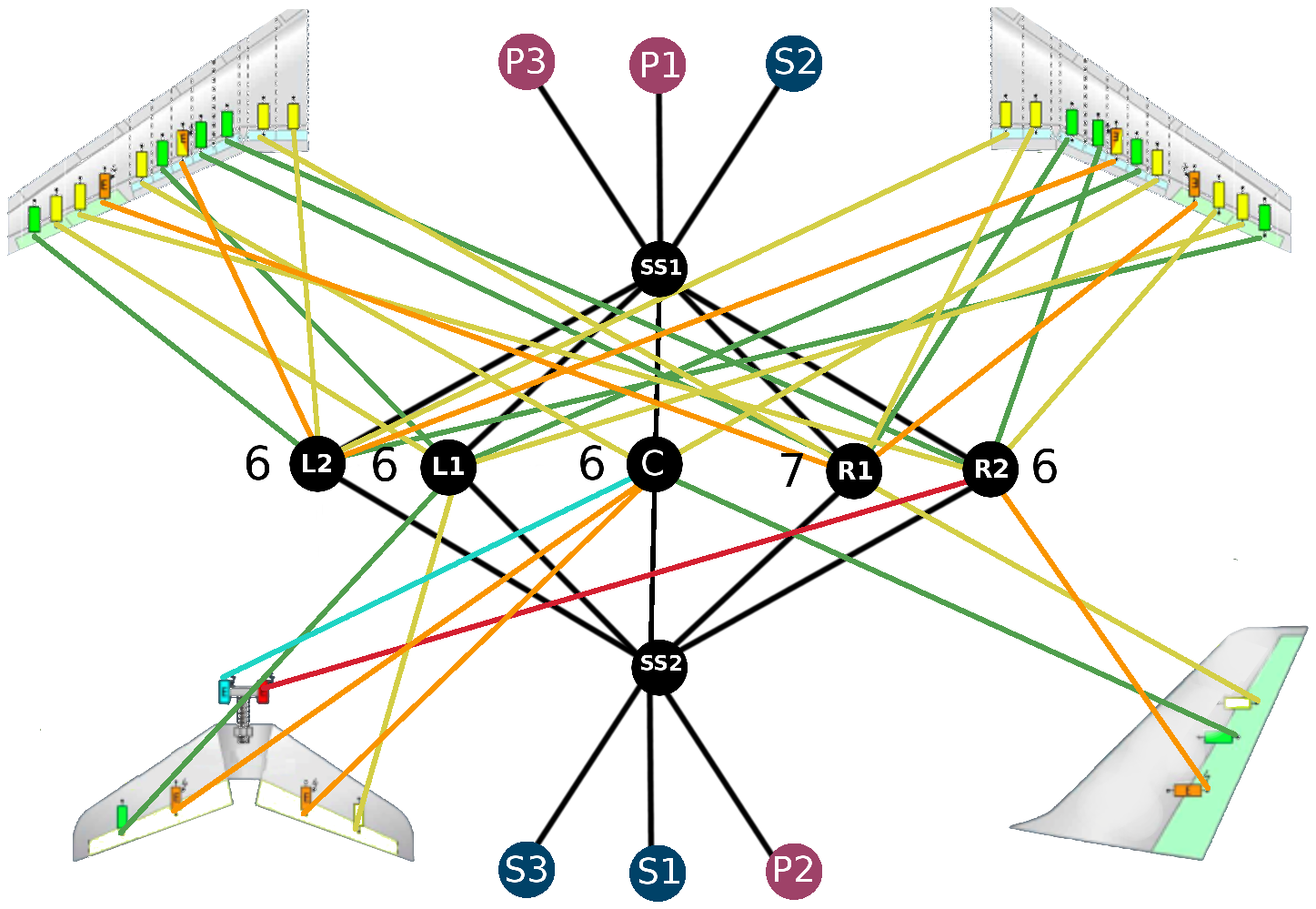}
	\footnotesize \caption{Primary network: new A350 flight control architecture}
	\label{A350new}
\end{figure}

In this architecture, the BLS is incorporated within the output ports from SS1/SS2 to the central switches, from the central switches to SS1/SS2, and from SS1/SS2 to the calculators. However, for the output port linking the central switches to the actuators, the BLS is not useful since the flight control traffic is the only type of traffic in these output ports. Hence, we obtain the network presented in Fig. \ref{outputportLayout}, where there is only one BLS along the path of each flow from a calculator to an actuator.

\begin{figure}[htbp]
	\centering
	\includegraphics[width=0.75\textwidth]{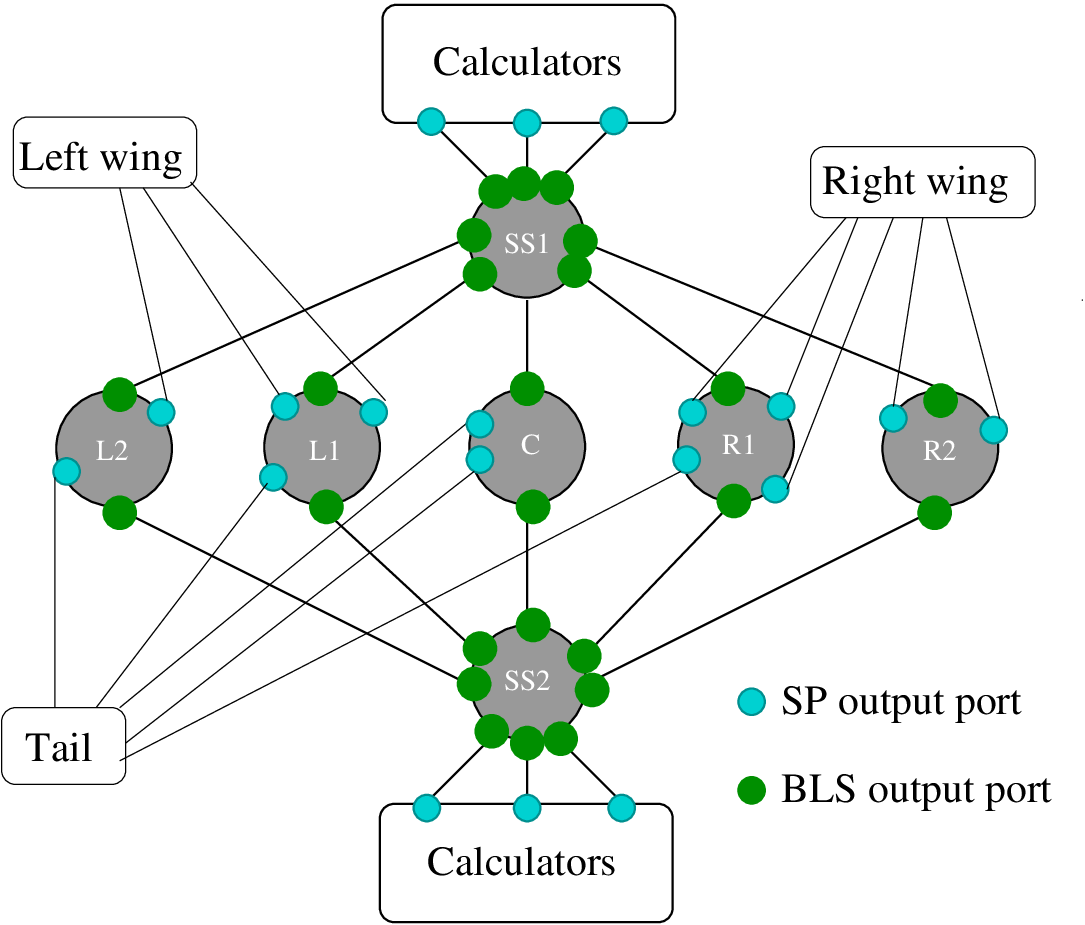}
	\footnotesize \caption{Output port type layout for the extended AFDX}
	\label{outputportLayout}
\end{figure}

\newpage
\noindent
\textbf{Timing analysis of the proposed solution}

In this section, we will consider many different scenarios to assess the performance of the extended AFDX solution in terms of delay bounds. The SCT traffic consists of  Flight Control frames and the considered scenarios are described in Table~\ref{FCscenarios}.

We use the architecture presented in Fig. \ref{A350new}. Hence, each calculator generates 17 flows with a periodicity $BAG_{SCT}$ and a frame size $MFS_{SCT}$. Then, the frames arrive in SS1 or SS2. The output port with the heaviest load receives 7 flows from each of the 3 calculators connected to the considered switch, thus a total of 21 flows. Finally, the flows arrive in the central switches. Each central switch output ports receives 6 flows, one from each calculator.

The current heaviest load on the current 100Mbps AFDX is 30\%. With the considered 1Gbps network, the heaviest load is 3\%. Hence, we consider that in each output port there are $n_{RC}$ flows defined by $MFS_{RC}$ and $BAG_{RC}$ to obtain an utilization rate of 3\%:

\noindent \hfil$n_{RC}=\ceil[\bigg]{\frac{3}{100}\cdot C\cdot \frac{BAG_{RC}}{MFS_{RC}}}$

Hence, we obtain:

\noindent \hfil $\alpha_ {RC}(t)=n_{RC}\cdot \left( \frac{MFS_{RC}}{BAG_{RC}}\cdot t +MFS_{RC}\right) $

We use this $\alpha_ {RC}$ in both the calculators and the first switch (either SS1 or SS2). The delay bounds are computed with the CCbA modelisation. In Section~\ref{Use-case2}, we highlighted the importance of good parameterization. Hence, in this section, the BLS parameters are computed with the tuning method proposed in\cite{finzituning}.

\begin{table}[h!]
	\footnotesize
	\begin{center}
		\begin{tabular} {|c|c|c|c|c|c|c|}
			\hline
			scenario &$MFS_{SCT}$ & $BAG_{SCT}$ & $Deadline_{SCT}$ & $MFS_{RC}$ & $BAG_{RC}$ & $Deadline_{RC}$\\
			&(bytes) & ($\mu$s) & ($\mu$s) & (bytes) & ($\mu$s) & (ms) \\
			\hline
			1 & 64 &  2000 & 1000 & 320 &  2000 & 2000\\
			\hline
			2 & 64 &  1000 & 1000 & 320 &  2000 & 2000\\
			\hline
			3 & 64 &  1000 & 500 & 320 &  2000 & 2000\\
			\hline
			4 & 64 &  1000 & 1000 & 320 &  4000 & 4000\\
			\hline
			5 & 64 &  1000 & 1000 & 640 &  2000 & 2000\\
			\hline
			6 & 64 &  1000 & 1000 & 1280 &  2000 & 2000\\
			\hline
			7 & 128 &  1000 & 1000 & 320 &  2000 & 2000\\
			\hline
			8 & 256 &  1000 & 1000 & 320 &  2000 & 2000\\
			\hline
		\end{tabular}
	\end{center}
	\footnotesize \caption{Flight Control application: scenarios}
	\label{FCscenarios}
\end{table}

The computed SCT and RC delay bounds under the different scenarios for the extended and current AFDX are detailed in Table \ref{FCsresults}. We consider two performance  measurements. First, the end-to-end SCT delay bounds between a calculator and an end-system, denoted $delay_{SCT}^{end2end,m}$ for either the extended AFDX ($m =BLS$)  or the current AFDX ($m =SP$). The goal is to verify that the end-to-end SCT deadline is fulfilled. Secondly, 
the RC delay bounds in the switches SS1 and SS2 under BLS.  We denote $delay_{RC}^{SW1,m}$,
the delay bound of the RC traffic in the first switch in the considered path,  with $SW1$ being SS1 or SS2 depending on the considered flow, with $m\in\{BLS,SP\}$. It is worth noting that due to the symmetry of the network the delay bound in SS1 is identical to the one in SS2.

\textbf{From the results in Table \ref{FCsresults}, we can see that for very diverse configurations, SCT delay bounds still fulfill the deadlines. Additionally, the RC delay bounds with the extended AFDX is always lower than the ones with the current AFDX, up to 49.9\% of improvement in SW1. This confirms the efficiency of our proposal to handle mixed-criticality traffic in a realistic situation.}

\begin{table}[h!]
	\footnotesize
	\begin{center}
		\begin{tabular} {|c|c|c|c|c|}
			\hline
			scenario & $delay_{SCT}^{end2end,BLS}$ & $delay_{SCT}^{end2end,SP}$ & $delay_{RC}^{SW1,BLS}$ & $delay_{RC}^{SW1,SP}$\\
			&($\mu$s) & ($\mu$s) &  ($\mu$s) & ($\mu$s)\\
			\hline
			1 & 120.00 & 130.27 & 44.89 & 55.16\\
			\hline
			2 & 121.96 & 132.7 & 44.89 & 55.63 \\
			\hline
			3 & 121.96 & 132.7 & 44.89 & 55.63 \\
			\hline
			4 & 184.54 & 195.73 & 84.73 & 95.92 \\
			\hline
			5 & 129.91 & 140.70 &  47.46 & 58.25  \\
			\hline
			6 & 145.80 & 156.68 & 52.60 & 63.48 \\
			\hline
			7 & 176.53 & 198.74 & 45.55 & 67.76 \\
			\hline
			8 & 298.62 & 345.21 & 46.79 & 93.38 \\
			\hline
		\end{tabular}
	\end{center}
	\footnotesize \caption{Flight Control application: results}
	\label{FCsresults}
\end{table}

\section{Conclusion}

To homogenize the avionics communication architecture for heterogeneous traffic, we have proposed a new modelisation of an extension of the AFDX incorporating several TSN/BLS.

First, we have studied the existing BLS models. In particular, we have showed that the existing CPA model~\cite{thiele2016formal} can lead to optimistic bounds. Then, after comparing the two existing 3-classes Network Calculus BLS models, i.e. WbA~\cite{Finzi-sies-18} and CCbA~\cite{Finzi-wfcs-18}, we have concluded that the Continuous Credit-based Approach (CCbA) is the most promising. Hence, we have generalized it to multiple classes and multiple BLS.

Secondly, we have presented the system model to specify the considered network and traffic models. In particular, we have detailed an extended AFDX switch output port able to support multiple BLS.

Thirdly, we have generalized the CCbA and proposed a formal timing analysis of the extended AFDX to formally prove that the hard real-time requirements are fulfilled. We have started by detailing the schedulability conditions. Then, we have detailed the modelisation of the BLS and of the extended AFDX output port multiplexer, before presenting the end-to-end delay bound computation. This has led to a discussion on the nature of the BLS showing the BLS is not a greedy shaper, and that it is better to consider the association of the BLS and NP-SP as a scheduler.

Finally, we have done a performance analysis of the extended AFDX. First, sensitivity and  tightness analyses have shown the good properties of our modelisation, in particular concerning the impact of $L_R$ with CCbA compared to WbA. Then, the comparison to the CPA model has shown that our model is less complex and has solved the optimistic and  pessimistic issues of the CPA modelisation of the BLS. A second case study with six classes and two BLS has highlighted the good properties of the BLS and its ability to reduce the RC classes delay bounds (delays divided up to 5.4 times) and enhanced the schedulability of any class (e.g. from 0 to 12\%), compared to a Static Priority Scheduler.
We have finished with a concrete application, i.e., adding the A350 flight control to the AFDX. Results show that the extended AFDX decreases the delay bounds of the existing AFDX traffic compared to a standard AFDX (up to 49.9\%), hence mitigating the impact of the added traffic on existing traffic.

\newpage
\appendix

\section{Computing Achievable Worst-Cases}
\label{AWCs}

We consider the 3-classes case study, where the SCT class is shaped by a BLS, presented in Fig. \ref{fig:BLSshaper2}  As there is only one shaped class: SCT, we use  $k=\emptyset$ to simplify, for the non-ambiguous notations, such as $L_M^k$ or $I_{send}^k$. Our aim is to compute  Achievable Worst-Cases for SCT and RC delays, i.e., realistic worst-cases.

In this section, we consider several hypothesis :
\begin{itemize}
	\item  traffics are packetized, i.e., we need to integrate the non-preemption impact;
	\item the same frame size for each class SCT, RC and BE, i.e., homogeneous traffic within each traffic class;
\end{itemize}

We use the four curves presented in Fig. \ref{fig:SCTbeta10} and Fig. \ref{fig:SCTgamma10} to compute two Achievable Worst-Cases for each traffic class, i.e., SCT and RC, for  the single-hop network defined in Section \ref{Use-case1}. It is worth noting there is not strict order between the different cases. Based on these scenarios, we will be able to calculate so called Achievable Worst-Case delays to have an idea on the tightness of both RC and SCT delay bounds.

As illustrated in Fig. \ref{fig:SCTbeta10} and Fig. \ref{fig:SCTgamma10}, there is an alternation of sending windows (when SCT traffic is sent) and idle windows (when RC traffic is sent). We call a \textit{cycle} a sending window followed by an idle window (or an idle window followed by a sending window). A so called maximum-sized cycle is made of so called realistic maximum sending and idle windows: $\Delta_{send}^{real}$, $\Delta_{idle}^{real}$.

\subsection{SCT achievable worst-cases}
\label{AWCSCT}
We start by presenting the methodology, before considering the two BLS behaviours described in Fig. \ref{fig:SCTbeta10}.

\subsubsection{Methodology}

\begin{figure}[h]
	\centering	
	\includegraphics[width=0.5\columnwidth]{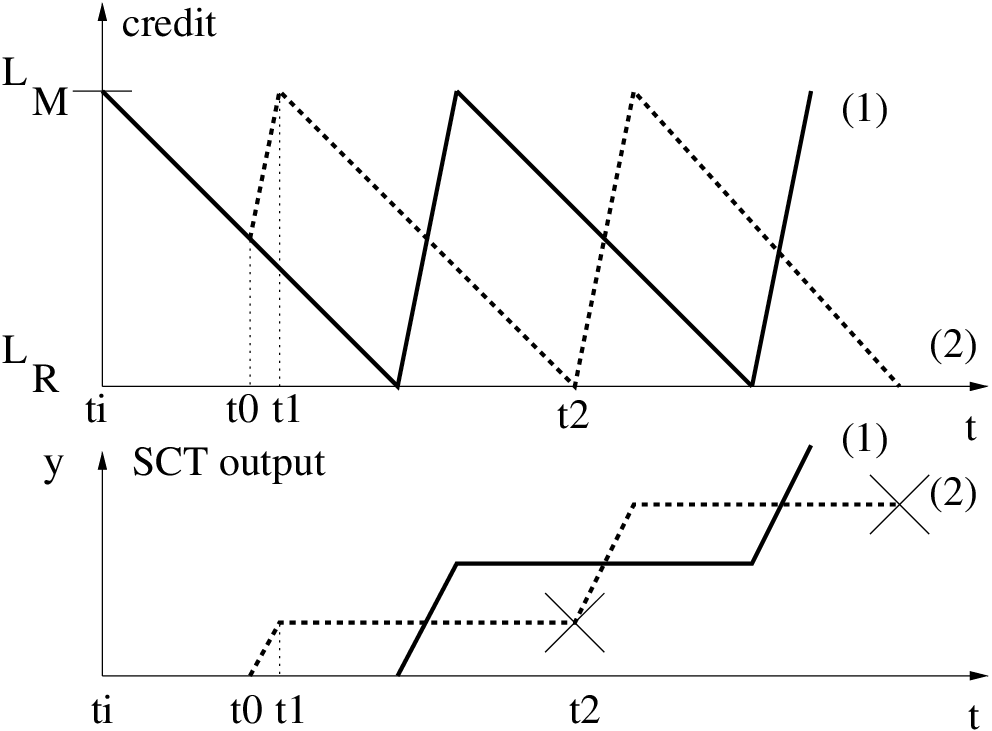}
	\caption{Two examples of worst-case BLS behaviours}
	\label{fig:SCTbeta10}
\end{figure}

To compute the worst-case delay of the SCT class, $delay_{SCT}^{max}$, we need to take into account of the following effects:

\begin{itemize}
	\item BE class impact due to the non-preemption feature. We need to consider the transmission of a maximum-sized BE frame that may be transmitted before a SCT frame;
	\item Transmission time of SCT burst: it is the time needed for the output port multiplexer to transmit the maximum SCT burst $b_{SCT}=n_{SCT}^{in}\cdot MFS_{SCT}$, with a transmission capacity $C$, when taking into account the shaping effect of the upstream links. 
	Each one of these link has a capacity $C$, resulting in the following transmission time: 
	$$\frac{n_{SCT}^{in}\cdot MFS_{SCT}}{C}-\frac{n_{SCT}^{in}\cdot MFS_{SCT}}{n_{SCT}^{links}\cdot C}$$
	\item RC class blocking effect $\Delta_{RC}^{blocking}$: it is the blocking effect of the shaper, which enforces the presence of $idle$ $windows$ (resp. $sending$ $windows$) to send the RC (resp. SCT) traffic.
	
\end{itemize}

Thus, we need to compute $\Delta_{RC}^{blocking}$ in each case to obtain the Achievable worst-cases.

The blocking effect depends the number of $realistic$ idle windows $\Delta_{idle}^{real}$ used by the RC traffic, denoted $Ncyle_{RC}^{used}$.
The computation of $Ncyle_{RC}^{used}$ is based on both:

i) $Ncycle_{RC}^{needed}$, the number of cycles $needed$ to send the RC traffic  during $realistic$ idle windows $\Delta_{idle}^{real}$;

ii) $Ncycle_{RC}^{available}$, the number of cycles $available$ to RC while the SCT traffic is being transmitted during $realistic$ sending windows $\Delta_{send}^{real}$. If a sending window is started, it means that a full idle window  $\Delta_{idle}^{real}$ is available to RC. We denote $Ncyle_{SCT}^{used}$ the number of $\Delta_{send}^{real}$ used by the SCT traffic. Hence, we have:
$$N{cycle}_{RC}^{available}= \ceil[\Bigg]{Ncyle_{SCT}^{used}}$$

Thus, we obtain: $$Ncyle_{RC}^{used}=\min(Ncycle_{RC}^{available},Ncycle_{RC}^{needed})$$

The computation of the number of cycles $Ncycle_{SCT}^{used}$ necessary to compute  $\Delta_{RC}^{blocking}$ is based on:

i) the SCT and RC traffics;

ii) realistic windows, which depend on the chosen BLS behaviour and will be computed in each specific case.

So first, we assess the SCT and RC traffics.
A strong hypothesis we make while computing the SCT traffic is that we do not consider the SCT traffic that may arrive while the SCT burst is being transmitted. Considering this additional traffic leads to the need of computing a fixed-point problem so we discard it here as it causes acceptable optimism rather than unacceptable pessimism. 
Hence, the maximum amount of considered SCT traffic is: $$B_{SCT}^{max}=b_{SCT}=n_{SCT}^{in}\cdot MFS_{SCT}$$

However, not considering the amount of RC traffic arriving while SCT traffic is waiting leads to a large optimism. Hence, to compute the impact of RC, we need to compute the the maximum amount of RC traffic that arrives while the SCT burst is being sent, i.e., during $delay_{SCT}^{max}$:
$$B_{RC}^{max}(delay_{SCT}^{max})=n_{RC}^{in}\cdot MFS_{RC}\cdot (1+\frac{delay_{SCT}^{max}}{BAG_{RC}})$$

Hence, the SCT delay is as follows:

\begin{equation}
delay_{SCT}^{max}=\frac{MFS_{BE}}{C}+\Delta_{RC}^{blocking}(delay_{SCT}^{max})+\frac{n_{SCT}^{in}\cdot MFS_{SCT}}{C}-\frac{n_{SCT}^{in}\cdot MFS_{SCT}}{n_{SCT}^{links}\cdot C}
\end{equation}

As $\Delta_{RC}^{blocking}(delay_{SCT}^{max})$ depends on $delay_{SCT}^{max}$, $delay_{SCT}^{max}$ can be computed by \textbf{solving this fixed point problem}. We consider $\frac{n_{SCT}^{in}\cdot MFS_{SCT}}{C}$ to be a good starting point.

\subsubsection{SCT Achievable Worst-Case 1}


We start by computing an Achievable Worst-Case  for the SCT class, denoted SCT AWC-1 using the plain line curve (1) in Fig. \ref{fig:SCTbeta10}.

\hfill\\
\textbf{1. Computing the RC blocking delay $\Delta_{RC}^{blocking}(delay_{SCT}^{max})$}

The RC blocking delay is defined by:
$$\Delta_{RC}^{blocking}(delay_{SCT}^{max})=Ncycle_{RC}^{used}\cdot \Delta_{idle}^{real}$$\\

So, to compute the RC blocking delay, we need to compute the number of $\Delta_{idle}^{real}$ windows used by RC, $Ncycle_{RC}^{used}$.

To compute the number of cycles, we first need the realistic windows.

\hfill\\
\textbf{2. Computing realistic sending and idle windows}

They are computed as the upper integer value of maximum  number of frame that can be sent during a minimum window multiplied by the transmission time of a frame. Concerning the sending window, we consider that the window starts at $L_R^{min}=\max(0,L_R-I_{idle}\cdot\frac{MFS_{RC}}{C})$. This last hypothesis may be slightly optimistic as the window can in fact starts between $L_R$ and $L_R^{min}$ depending on the frames transmissions and sizes.

$$\Delta_{send}^{real}=\ceil[\Bigg]{ \frac{\frac{{L_M-L_R^{min}}}{I_{send}}}{\frac{MFS_{SCT}}{C}}}\cdot \frac{MFS_{SCT}}{C}$$

$$\Delta_{idle}^{real}=\ceil[\Bigg]{ \frac{\frac{{L_M-L_R}}{I_{idle}}}{\frac{MFS_{RC}}{C}}}\cdot \frac{MFS_{RC}}{C}$$

\hfill\\
\textbf{3. Computing $Ncycle_{RC}^{needed}$}

We compute the number of cycles necessary to send the RC traffic $B_{RC}^{max}(delay_{SCT}^{max})$:
$$Ncycle_{RC}^{needed}= \frac{B_{RC}^{max}(delay_{SCT}^{max})}{C\cdot \Delta_{idle}^{real}}$$

\hfill\\
\textbf{4. Computing $N{cycle}_{RC}^{available}$}

Finally, the number of windows available to RC is:
$$N{cycle}_{RC}^{available}= \ceil[\Bigg]{\frac{n_{SCT}^{in}\cdot MFS_{SCT}}{C\cdot \Delta_{send}^{real}}}$$


\subsubsection{SCT Achievable Worst-Case 2}

For the second achievable worst-case, denoted SCT AWC-2, we use the dotted curve (2) in Fig. \ref{fig:SCTbeta10}.

\hfill\\
\textbf{1. Computing the RC blocking delay $\Delta_{RC}^{blocking}(delay_{SCT}^{max})$}

To compute the RC blocking delay, we need to take into account the RC and SCT traffics sent between ti and t1. So, we compute:

i) the SCT traffic sent between t0 and t1;

ii) the RC traffic sent between ti and t0, and the corresponding window $\Delta_{tit0-idle}^{real}$.

Then, we can compute $N{cycle}_{RC}^{used}$, the number of maximum-sized cycles, i.e., $\Delta_{send}^{real}+\Delta_{idle}^{real}$ used to send the RC burst remaining after $t1$, i.e., $B_{RC}^{max}(delay_{SCT}^{max})-\Delta_{tit0-idle}^{real}\cdot C$.

Thus, we obtain the following RC blocking delay:
$$\Delta_{RC}^{blocking}(delay_{SCT}^{max})=\Delta_{tit0-idle}^{real} +Ncycle_{RC}^{used}\cdot \Delta_{idle}^{real}$$

To compute the number of cycles, we first need the realistic windows between ti and t0, and between t0 and t1.

\hfill\\
\textbf{2. Computing realistic sending and idle windows}

We consider that during the first idle window, there is no RC traffic backlogged when the credit reaches $\frac{L_M}{2}$. So SCT traffic is sent until there is again RC traffic, i.e., when the credit reaches $L_R$ (see Fig. \ref{fig:SCTbeta10}).

First, we will use the same realistic maximum sending and idle windows as for AWC-1. We will also compute the realistic sending and idle windows from $L_R$ to $\frac{L_M}{2}$ and from $\frac{L_M}{2}$ to $L_R^{min}$. 

$$\Delta_{t0t1-send}^{real}=\ceil[\Bigg]{ \frac{{\frac{L_M}{2}-L_R^{min}}}{I_{send}}\cdot{\frac{C}{MFS_{SCT}}}}\cdot \frac{MFS_{SCT}}{C}$$

$$\Delta_{tit0-idle}^{real}=\ceil[\Bigg]{ {\frac{{\frac{L_M}{2}-L_R}}{I_{idle}}}\cdot{\frac{C}{MFS_{RC}}}}\cdot \frac{MFS_{RC}}{C}$$ 

\hfill\\
\textbf{3. Computing $N{cycle}_{RC}^{needed}$}

The number of maximum-sized cycles necessary to send the RC traffic after t1 is the numbers of cycles necessary to send the RC burst minus the traffic sent during $\Delta_{tit0-idle}^{real}$:
$$Ncycle_{RC}^{needed}= \frac{B_{RC}^{max}(delay_{SCT}^{max})-\Delta_{tit0-idle}^{real}\cdot C}{C\cdot \Delta_{idle}^{real}}$$

\hfill\\
\textbf{4. Computing $N{cycle}_{RC}^{available}$}

To compute the number of cycle available to RC after t1, we consider the remaining SCT burst after t1. Thus, we have:

$$N{cycle}_{RC}^{available}= \ceil[\Bigg]{\frac{n_{SCT}^{in}\cdot MFS_{SCT}-\Delta_{t0t1-send}^{real}\cdot C}{C\cdot \Delta_{send}^{real}}}$$

\subsection{RC achievable worst-cases}

We start by presenting the methodology, before considering the two BLS behaviours described in Fig. \ref{fig:SCTgamma10}.

\begin{figure}[h]
	\centering
	\includegraphics[width=0.6\columnwidth]{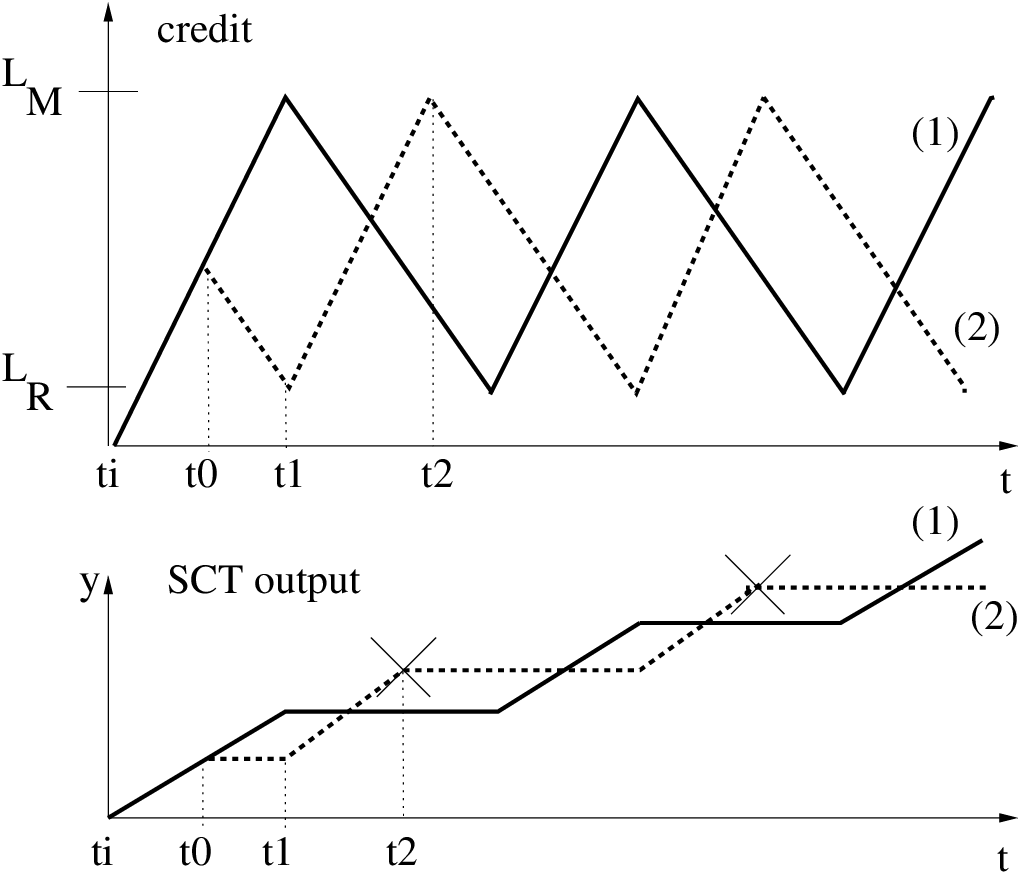}
	\caption{Two examples of best-case BLS behaviours}
	\label{fig:SCTgamma10}
\end{figure}

\subsubsection{Methodology}

To compute the worst-case delay of the RC class, $delay_{RC}^{max}$, we need to do an account of the following effects:

\begin{itemize}
	\item BE class impact due to the non-preemption feature. We need to consider the transmission of a maximum-sized BE frame that may be transmitted before a RC frame;
	\item Transmission time of RC burst: it is the time needed for the output port multiplexer to transmit the maximum RC burst $b_{RC}=n_{RC}^{in}\cdot MFS_{RC}$, with a transmission capacity $C$, when taking into account the shaping effect of the upstream links. 
	Each one of these link has a capacity $C$, resulting in the following transmission time: 
	$$\frac{n_{RC}^{in}\cdot MFS_{RC}}{C}-\frac{n_{RC}^{in}\cdot MFS_{RC}}{n_{RC}^{links}\cdot C}$$
	\item SCT class blocking effect $\Delta_{SCT}^{blocking}$: it is the blocking effect of the shaper, which enforces the presence of $idle$ $windows$ (resp. $sending$ $windows$) to send the RC (resp. SCT) traffic.
	
\end{itemize}

Thus, we need to compute $\Delta_{SCT}^{blocking}$ in each case to obtain the Achievable worst-cases.

The blocking effect depends the number of $realistic$ sending windows $\Delta_{send}^{real}$ used by the SCT traffic, denoted $Ncyle_{SCT}^{used}$.
The computation of $Ncyle_{SCT}^{used}$ is based on both:

i) $Ncycle_{SCT}^{needed}$, the number of cycles $needed$ to send the SCT traffic  during $realistic$ sending windows $\Delta_{send}^{real}$;

ii) $Ncycle_{SCT}^{available}$, the number of cycles $available$ to SCT while the RC traffic is being transmitted during $realistic$ idle windows $\Delta_{idle}^{real}$. If an idle window is started, it means that a full sending window  $\Delta_{send}^{real}$ is available to SCT. We denote $Ncyle_{RC}^{used}$ the number of $\Delta_{idle}^{real}$ used by the RC traffic. Hence, we have:
$$N{cycle}_{SCT}^{available}= \ceil[\Bigg]{Ncyle_{RC}^{used}}$$

Thus, we obtain: $$Ncyle_{SCT}^{used}=\min(Ncycle_{SCT}^{available},Ncycle_{SCT}^{needed})$$

The computation of the number of cycles  $Ncycle_{SCT}^{used}$ necessary to compute  $\Delta_{SCT}^{blocking}$ is based on:

i) the SCT and RC traffics;

ii) realistic windows, which depend on the chosen BLS behaviour and will be computed in each specific case.

So first, we assess the SCT and RC traffics.
A strong hypothesis we make while computing the RC traffic is that we do not consider the RC traffic that may arrive while the RC burst is being transmitted. Considering this additional traffic leads to the need of computing a fixed-point problem so we discard it here as it causes acceptable optimism rather than unacceptable pessimism. 
Hence, the maximum amount of considered RC traffic is: $$B_{RC}^{max}=b_{RC}=n_{RC}^{in}\cdot MFS_{RC}$$

However, not considering the amount of SCT traffic arriving while RC traffic is waiting leads to a large optimism. Hence, to compute the impact of SCT, we need to compute the the maximum amount of SCT traffic that arrives while the RC burst is being sent, i.e., during $delay_{RC}^{max}$:
$$B_{SCT}^{max}(delay_{RC}^{max})=n_{SCT}^{in}\cdot MFS_{SCT}\cdot (1+\frac{delay_{RC}^{max}}{BAG_{SCT}})$$

Hence, the RC delay is as follows:
\begin{equation}
delay_{RC}^{max}=\frac{MFS_{BE}}{C}+\Delta_{SCT}^{blocking}(delay_{RC}^{max})+\frac{n_{RC}^{in}\cdot MFS_{RC}}{C}-\frac{n_{RC}^{in}\cdot MFS_{RC}}{n_{RC}^{links}\cdot C}
\end{equation}

As $\Delta_{SCT}^{blocking}(delay_{RC}^{max})$ depends on $delay_{RC}^{max}$, $delay_{RC}^{max}$ can be computed by \textbf{solving this fixed point problem}. We consider $\frac{n_{RC}^{in}\cdot MFS_{RC}}{C}$ to be a good starting point.

\subsubsection{RC Achievable Worst-Case 1}


We compute an Achievable Worst-Case  for the RC class, denoted RC AWC-1 using the plain line curve (1) in Fig. \ref{fig:SCTgamma10}.

\hfill\\
\textbf{1. Computing the SCT blocking delay $\Delta_{SCT}^{blocking}(delay_{RC}^{max})$}

An important difference between SCT and RC is the presence of an initial maximum sending windows starting at 0, denoted $\Delta_{send,0}^{real}$. It differs from the usual maximum sending windows which start at $L_R$. Hence, the initial SCT burst sent during this  $\Delta_{send,0}^{real}$ must be taken into account throughout this computation of the SCT blocking delay.

To compute the SCT blocking delay, we need to compute:  

i) $N{cycle}_{SCT}^{available}$, the number of cycles available to SCT after t0, taking into account the impact of $\Delta_{send,0}^{real}$;

ii) $Ncycle_{SCT}^{needed}$, the number of cycles needed to send the SCT traffic remaining after t1, i.e., the SCT burst minus the traffic sent during $\Delta_{send,0}^{real}$.


Finally, we can compute $\Delta_{SCT}^{blocking}$, the interfering SCT traffic delay due to the shaper blocking effect:

$$\Delta_{SCT}^{blocking}(delay_{RC}^{max})=\min\left(Ncycle_{SCT}^{needed},Ncycle_{SCT}^{available}\right)\cdot \Delta_{send}^{real} + \Delta_{send,0}^{real}$$

%

To compute the number of cycles, we first need the realistic windows.

\hfill\\
\textbf{2. Computing realistic sending and idle windows}

They are the same as Section~\ref{AWCSCT}, except for the fact that we consider for $\Delta_{send}^{real}$  that the window starts at $L_R$ (instead of $L_R^{min}$) which may again be slightly optimistic as the window can in fact starts between $L_R$ and $L_R^{min}$ depending on the frames transmissions and sizes. Additionally, we consider the initial sending window $\Delta_{send,0}^{real}$.
$$\Delta_{send}^{real}=\ceil[\Bigg]{ \frac{\frac{{L_M-L_R}}{I_{send}}}{\frac{MFS_{SCT}}{C}}}\cdot \frac{MFS_{SCT}}{C}$$

$$\Delta_{idle}^{real}=\ceil[\Bigg]{ \frac{\frac{{L_M-L_R}}{I_{idle}}}{\frac{MFS_{RC}}{C}}}\cdot \frac{MFS_{RC}}{C}$$ 

$$\Delta_{send,0}^{real}=\ceil[\Bigg]{ \frac{\frac{{L_M-0}}{I_{send}}}{\frac{MFS_{SCT}}{C}}}\cdot \frac{MFS_{SCT}}{C}$$

\newpage
\noindent
\textbf{3. Computing $Ncycle_{SCT}^{needed}$}

Next, we compute the number of window cycles necessary to send the SCT traffic remaining after t0. We must consider the first sending window, $\Delta_{send,0}^{real}$ which starts at $0$:

The amount of traffic sent during this window is:
$$b_{SCT,0}^{max}=\Delta_{send,0}^{real}\cdot C$$

Finally, the number of cycles necessary to send the remaining SCT traffic after t0 is:
$$Ncycle_{SCT}^{needed}= \frac{B_{SCT}^{max}(delay_{RC}^{max})-b_{SCT,0}^{max}}{C\cdot \Delta_{send}^{real}}$$

\hfill\\
\textbf{3. Computing $N{cycle}_{SCT}^{available}$}

tThe first available window is $\Delta_{send,0}^{real}$, and is taken into account directly in $\Delta_{SCT}^{blocking}(delay_{RC}^{max})$. So, We must remove 1 from $N{cycle}_{RC}^{used}$. Finally, we have:
$$N{cycle}_{SCT}^{available}= \ceil[\Bigg]{\frac{n_{RC}^{in}\cdot MFS_{RC}}{C\cdot \Delta_{idle}^{real}}}-1$$

\subsubsection{RC Achievable Worst-Case 2}

For the second achievable worst-case, denoted RC AWC-2, we use the dotted curve (2) in Fig. \ref{fig:SCTgamma10}.

\hfill\\
\textbf{1. Computing the SCT blocking delay $\Delta_{SCT}^{blocking}(delay_{RC}^{max})$}

To compute the SCT blocking delay, we need to take into account the RC and SCT traffics sent between ti and t1. So, we  compute:

i) the RC traffic sent between t0 and t1;

ii) the SCT traffic sent between ti and t0, and the corresponding window $\Delta_{tit0-send}^{real}$.

Then, we can compute $N{cycle}_{SCT}^{used}$, the number of maximum-sized cycles, i.e., $\Delta_{send}^{real}+\Delta_{idle}^{real}$ used to send the SCT burst remaining after $t1$, i.e., $B_{SCT}^{max}(delay_{RC}^{max})-\Delta_{t0t1-send}^{real}\cdot C$.

Thus, we obtain the following SCT blocking delay:
$$\Delta_{SCT}^{blocking}(delay_{RC}^{max})=\Delta_{tit0-send}^{real} +Ncycle_{SCT}\cdot \Delta_{send}^{real}$$

To compute the number of cycles, we first need the realistic windows between ti and t0, and between t0 and t1.

\hfill\\
\textbf{2. Computing realistic sending and idle windows}

We consider that during the first sending window, there is no SCT traffic backlogged when the credit reaches $\frac{L_M}{2}$. So RC traffic is sent until there is again SCT traffic, i.e., when the credit reaches $L_R$ (see Fig. \ref{fig:SCTgamma10}).

First, we will use the same realistic maximum sending and idle windows as for AWC-1. We will also compute the realistic sending and idle windows from $0$ to $\frac{L_M}{2}$ and from $\frac{L_M}{2}$ to $L_R$. 

$$\Delta_{tit0-send}^{real}=\ceil[\Bigg]{ \frac{{\frac{L_M}{2}-0}}{I_{send}}\cdot{\frac{C}{MFS_{SCT}}}}\cdot \frac{MFS_{SCT}}{C}$$

$$\Delta_{t0t1-idle}^{real}=\ceil[\Bigg]{ \frac{{\frac{L_M}{2}-L_R}}{I_{idle}}\cdot{\frac{C}{MFS_{RC}}}}\cdot \frac{MFS_{RC}}{C}$$

\newpage
\noindent
\textbf{3. Computing $N{cycle}_{SCT}^{needed}$}

The number of maximum-sized cycles necessary to send the SCT traffic is the numbers of cycles necessary to send the SCT burst minus the traffic sent during $\Delta_{tit0-send}^{real}$:
$$N{cycle}_{SCT}^{needed}= \frac{B_{SCT}^{max}(delay_{RC}^{max})-\Delta_{tit0-send}^{real}\cdot C}{C\cdot \Delta_{send}^{real}}$$

\hfill\\
\textbf{4. Computing $N{cycle}_{SCT}^{available}$}

We consider the remaining RC burst after t1. Thus, we have:
$$N{cycle}_{SCT}^{available}= \ceil[\Bigg]{\frac{n_{RC}^{in}\cdot MFS_{RC}-\Delta_{t0t1-idle}^{real}\cdot C}{C\cdot \Delta_{idle}^{real}}}$$

\section{Generalized Continuous-Credit-based Approach (gCCbA) model proofs}

In this section, we detail the proofs of the strict minimum and maximum service curves. Both proofs are based on three lemmas  presented in the next section.
\subsection{Continuous-credit Lemmas}

We denote $R^*_{k} (t)$ the output cumulative function of the class $k$ traffic, and $\Delta R^*_{k}(\delta)$ its variation during an interval $\delta$.

The BLS credit tries to keep an accurate accounting of the traffic sent. There are two situations when it loses track due to non-preempted transmissions: 
\begin{enumerate}
	\item  when the credit reaches $L_M^k$ and the current class $k$ frame has not finished its transmission; 
	\item when the credit reaches $0$ and the current frame is still being transmitted. 
\end{enumerate}

We call this the saturation of the credit, either at $L_M^k$ by class $k$ traffic, or at $0$ by other traffics.  The saturation at $L_M^k$ can only occur when a class $k$ frame is being transmitted, while the saturation at 0 can not occur when a class $k$ frame is being transmitted.

Hence, we call $\Delta R^{*}_{L_M^k, sat}(\delta)$ (resp.$\Delta R^{*}_{0, sat}(\delta)$) the part of $\Delta R^{*}_{k}(\delta)$ (resp. $\delta\cdot C-\Delta R^{*}_{k}(\delta)$), that can be sent during any interval $\delta$ while the credit is saturated at $L_M^k$ (resp. at 0).

We present here three lemmas  linked to the credit saturation and necessary to the service curve proofs. First in Lemma \ref{minmaxsum}, we show how to bound the sum of the credit consumed and the credit gained, depending on the credit saturations. Then, we detail the bounds of the credit saturations at $L_M^k$ in Lemma~\ref{lemmasat0}, and at $0$ in Lemma ~\ref{lemmasatlm}.

\begin{lemma}[Continuous credit bounds]\label{minmaxsum}
	We consider a shaped class $k$, with a maximum credit level $L_M^k$. $\forall \delta$, computing the sum of the credit consumed and gained give the following inequations:	
	
	\begin{eqnarray}
	L_M^k  \geqslant & 	\left( \begin{aligned}
	\Delta R^*_{k}(\delta)-\frac{\Delta R^{*}_{L_M^k,sat}(\delta)}{C}\cdot I_{send}^k\\-(\delta - \frac{\Delta R^{*}_{0,sat}(\delta)}{C})\cdot I_{idle}^k
	\end{aligned}\right)  & \geqslant  -L_M^k\nonumber 
	\end{eqnarray}
\end{lemma}
\begin{proof}
	
	In an interval $\delta$, the accurate consumed credit is the time it takes to send the non-saturating traffic $\frac{\Delta R^*_{k}(\delta)-\Delta R^{*}_{L_M^k,sat}(\delta)}{C}$ multiplied by the sending slope:
	$$credit_{consumed}^k=\left( \frac{\Delta R^*_{k}(\delta)-\Delta R^{*}_{L_M^k,sat}(\delta)}{C}\right) \cdot I_{send}^k$$ And conversely, the accurate gained credit is the remaining time $\delta - \frac{\Delta R^*_{k}(\delta)}{C}$ minus the saturation time $\frac{\Delta R^{*}_{0,sat}(\delta)}{C}$,  multiplied by the signed idle slope:	 
	$$credit_{gained}^k=\left( \delta - \frac{\Delta R^*_{k}(\delta)+\Delta R^{*}_{0,sat}(\delta)}{C}\right) \cdot (-I_{idle}^k)$$

	Thus $\forall \delta\in \mathds R^+$, using the fact that $I_{send}^k+I_{idle}^k=C$,  the sum of the gained credit  and the consumed credit is:
	
	\begin{eqnarray}
	credit_{consumed}^k+credit_{gained}^k\nonumber&=&
	(\frac{\Delta R^*_{k}(\delta)-\Delta R^{*}_{L_M^k,sat}(\delta)}{C})\cdot (I_{send}^k)\nonumber\\&&
	+(\delta - \frac{\Delta R^*_{k}(\delta)+\Delta R^{*}_{0,sat}(\delta)}{C})\cdot (-I_{idle}^k)\nonumber	\\
	&=&\Delta R^*_{k}(\delta)-\frac{\Delta R^{*}_{L_M^k,sat}(\delta)}{C}\cdot I_{send}^k\nonumber\\&&-(\delta - \frac{\Delta R^{*}_{0,sat}(\delta)}{C})\cdot I_{idle}^k\nonumber
	\end{eqnarray}
	
	We know that the credit is a continuous function with a lower bound: 0 and an upper bound $L_M^k$. So the sum of credit consumed and gained is always bounded by $-L_M^k$ and $+L_M^k$.
	
	\begin{eqnarray}
	L_M^k  \geqslant & credit_{consumed}^k+credit_{gained}^k& \geqslant  -L_M^k \nonumber \\
	L_M^k  \geqslant & 	\left( \begin{aligned}
	\Delta R^*_{k}(\delta)-\frac{\Delta R^{*}_{L_M^k,sat}(\delta)}{C}\cdot I_{send}^k\\-(\delta - \frac{\Delta R^{*}_{0,sat}(\delta)}{C})\cdot I_{idle}^k
	\end{aligned}\right)  & \geqslant  -L_M^k\nonumber 
	\end{eqnarray}
\end{proof}
\begin{lemma}[credit saturation at $0$]\label{lemmasat0}
	We consider a shaped class $k$, with the aggregate traffic of priority strictly higher than $p_H(k)$ $\alpha_h$-constrained with $\alpha_h=r_h\cdot t +b_h $. 
	
	$\forall \delta$, the amount of traffic sent while the traffic is saturated at $0$ is such as:
	\begin{eqnarray}
	0	\leqslant\Delta R^{*}_{0,sat}(\delta)
	\leqslant 	\left( \begin{aligned} &\sum\limits_{h\in HC(k)} r_h\cdot \delta+b_h\\&+
	MFS_{MC(k)}^{sat}\cdot\left(  \frac{\delta}{\Delta^{k,\beta}_{inter}}+1\right)\end{aligned}\right)  \nonumber
	\end{eqnarray}
	with: 				
	\begin{eqnarray}
	&& MFS_{MC(k)}^{sat}=\max(\max_{j\in MC(k)}MFS_{j}-\frac{C}{I_{idle}^k}\cdot L_R^k,0)\nonumber\\
	&&\Delta^{k,\beta}_{inter}= \frac{\max_{j\in MC(k)}MFS_{j}}{C}+\frac{L_M^k-L_R^{k,min}}{I_{send}^k}+\frac{L_M^k-L_R^k}{I_{idle}^k}\nonumber\\
	&& L_R^{k,min}=\max(L_R^k-\frac{\max_{j\in MC(k)}MFS_{j}}{C}\cdot I_{idle}^k,0)\nonumber
	\end{eqnarray}  
\end{lemma}
\begin{proof}

	First, we know that $\Delta R^{*}_{0,sat}(\delta)\geqslant0$.
	Secondly, we consider the impact of MC(k) and HC(k), the impact of LC(k) being taken into account in $\beta_{k}^{sp}(t)$ to compute an upper bound.

	\textit{Impact of MC(k) on $\Delta R^{*}_{0,sat}(\delta)$}
	
	In the presence of class $k$ frames, the saturation of the credit at 0 can occur if an additional frame is sent while the credit is decreasing and about to reach $L_R^k$. Due to non-preemption, the frame finishes its transmission even though the class $k$ priority is now higher.
	
	To be able to compute the largest impact of the non-preemption of MC(k) frames on class $k$ traffic, we must find the highest number of non-preempted frames that can be sent during a time interval $\delta$. Then, we must compute the part of the non-preempted frame sent while the credit is saturated.
	
	So first, we must compute the smallest duration between two occurrences of the phenomenon. Fig. \ref{fig:minComputesketch1} illustrates the following explanation. 
	
	\begin{figure}[h]
		\centering
		\includegraphics[width=0.6\linewidth]{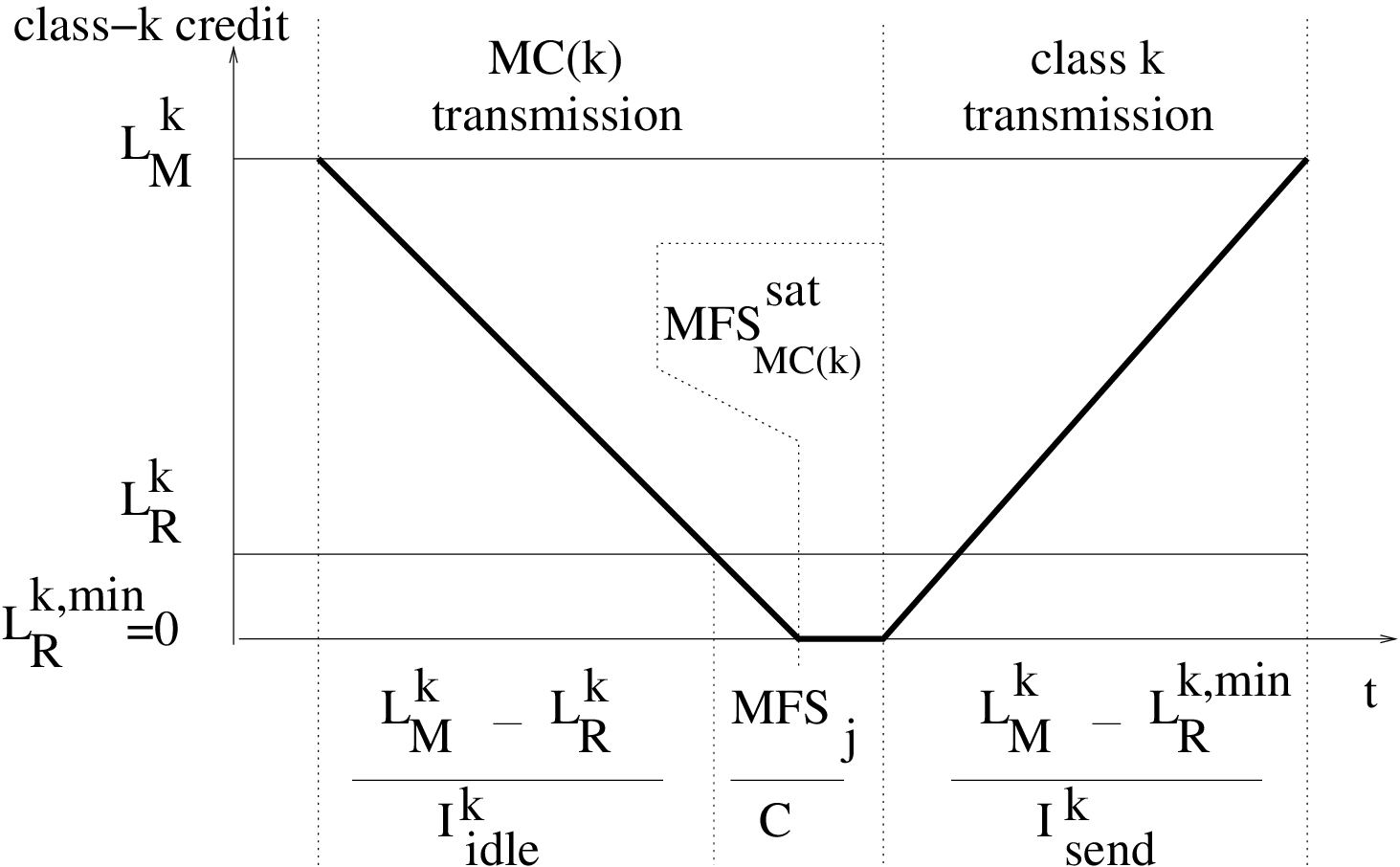}		
		\caption{Computing $\beta_{k}^{bls}(t)$}		
		\label{fig:minComputesketch1}
	\end{figure}
	
	After the first non-preempted MC(k) frame has been sent, the priority of the class $k$ queue is high. So, in presence of class $k$ traffic, no MC(k) traffic can be sent until a priority change: $L_M^k$ must be reached between two non-preempted MC(k) frames.	
	
	Thus, we study the intervals of time between the start of two transmissions of non-preempted MC(k) frames starting their transmission just before $L_R^k$ is reached. The smallest duration of such an interval is equal to the sum of 
	\begin{enumerate}
		\item the transmission time of the non-preempted MC(k) frame , such as at the end of the transmission the credit reaches $L_R^{k,min}=\max(L_R^k-\frac{\max_{j\in MC(k)}MFS_{j}}{C}\cdot I_{idle}^k,0)$;
		\item the duration $ \frac{L_M^k-L_R^{k,min}}{I_{send}^k}
		$ because  class-$k$ traffic has to be sent continuously in order for the credit to reach $L_M^k$ in the minimum duration ; 
		\item finally $\frac{L_M^k-L_R^k}{I_{idle}^k}$ because MC(k) traffic has to be sent continuously in order for the credit to return in the minimum duration to $L_R^k$. 
	\end{enumerate}
	
	In total, the minimum duration between the start of the transmission of two non-preempted MC(k) frames (each starting just before $L_R^k$ is reached), is $$\Delta^{k,\beta}_{inter}= \frac{\max_{j\in MC(k)}MFS_{j}}{C}+\frac{L_M^k-L_R^{k,min}}{I_{send}^k}+\frac{L_M^k-L_R^k}{I_{idle}^k}$$ 
	
	Thus during $\delta$, the number of time a non-preempted MC(k) frame can be sent is upper bounded by $\lceil\frac{\delta}{\Delta^{k,\beta}_{inter}}\rceil$.

	Secondly, we need to compute the maximum amount data sent while the credit remains at 0 during the transmission of one non-preempted maximum-sized MC(k) frame as illustrated in Fig. \ref{fig:minComputesketch1}.  This is equal to the maximum size of a MC(k) frame, minus the amount of data transmitted while the credit decreases from $L_R^k$ to 0: $$MFS_{MC(k)}^{sat}=\max(\max_{j\in MC(k)}MFS_{j}-\frac{C}{I_{idle}^k}\cdot L_R^k,0)$$


	\textit{Impact of HC(k) on $\Delta R^{*}_{0,sat}(\delta)$}
	
	The second way credit can saturate at 0 happens if traffic from HC(k) is sent while the credit remains at 0. We denoted $\alpha_h(t)$ the aggregate traffic of HC(k), arriving at a rate of $r_h$, with a burst $b_h$, such as $\alpha_h(t)=r_h\cdot \delta+b_h$.

	As a result, the amount of MC(k) and HC(k) traffic sent while the credit is saturated is such as:	
	\begin{eqnarray}
	\Delta R^{*}_{0,sat}(\delta)\leqslant \sum\limits_{h\in HC(k)} r_h\cdot \delta+b_h+ MFS_{MC(k)}^{sat}\cdot\lceil\frac{\delta}{\Delta^{k,\beta}_{inter}}\rceil\nonumber\\
	\leqslant \sum\limits_{h\in HC(k)} r_h\cdot \delta+b_h+
	MFS_{MC(k)}^{sat}\cdot\left(  \frac{\delta}{\Delta^{k,\beta}_{inter}}+1\right) \nonumber
	\end{eqnarray}

	%
\end{proof}
\begin{lemma}[credit saturation at $L_M^k$]\label{lemmasatlm}
	We consider a shaped class $k$, with the aggregate traffic of priority strictly higher than $p_H(k)$ $\alpha_h$-constrained with $\alpha_h=r_h\cdot t +b_h $. 
	
	$\forall \delta$, the amount of traffic sent while the traffic is saturated at $L_M^k$ is such as:
	\begin{eqnarray}
	0\leqslant\Delta R^{*}_{L_M^k, sat}(\delta)\leqslant
	MFS_{k}\cdot\left(  \frac{\delta}{\Delta^{k,\gamma}_{inter}}+1\right)	\nonumber
	\end{eqnarray} 
	with:
	$$\Delta^{k,\gamma}_{inter}= \frac{MFS_{k}}{C}+\frac{L_M^k-L_R^k}{I_{idle}^k}+\frac{L_M^k-L_R^k}{I_{send}^k}$$
	
\end{lemma}
\begin{proof}
	First, we know that $\Delta R^{*}_{L_M^k,sat}(\delta)\geqslant0$.
	Secondly for the upper bound, in the presence of MC(k) frames, the saturation of the credit at $L_M^k$ can only occur if an additional frame is sent while the credit is increasing and about to reach $L_M^k$. Due to non-preemption, the frame finishes its transmission even though the class $k$ priority is now lower.
	
	To be able to compute the largest impact of the non-preemption of class $k$ frames, we must find the highest number of non-preempted frames that can be sent during a time interval $\delta$. So we must compute the smallest duration between two occurrences of the phenomenon. Fig. \ref{fig:maxComputesketch1} illustrates the following explanation. After the first non-preempted class-$k$ frame has been sent, the priority of the class $k$ queue is low, so in presence of MC(k) traffic, no class $k$ traffic can be sent until a priority change: $L_R^k$ must be reached between two non-preempted class-$k$ frames.		
	\begin{figure}[h]
		\centering
		\includegraphics[width=0.6\linewidth]{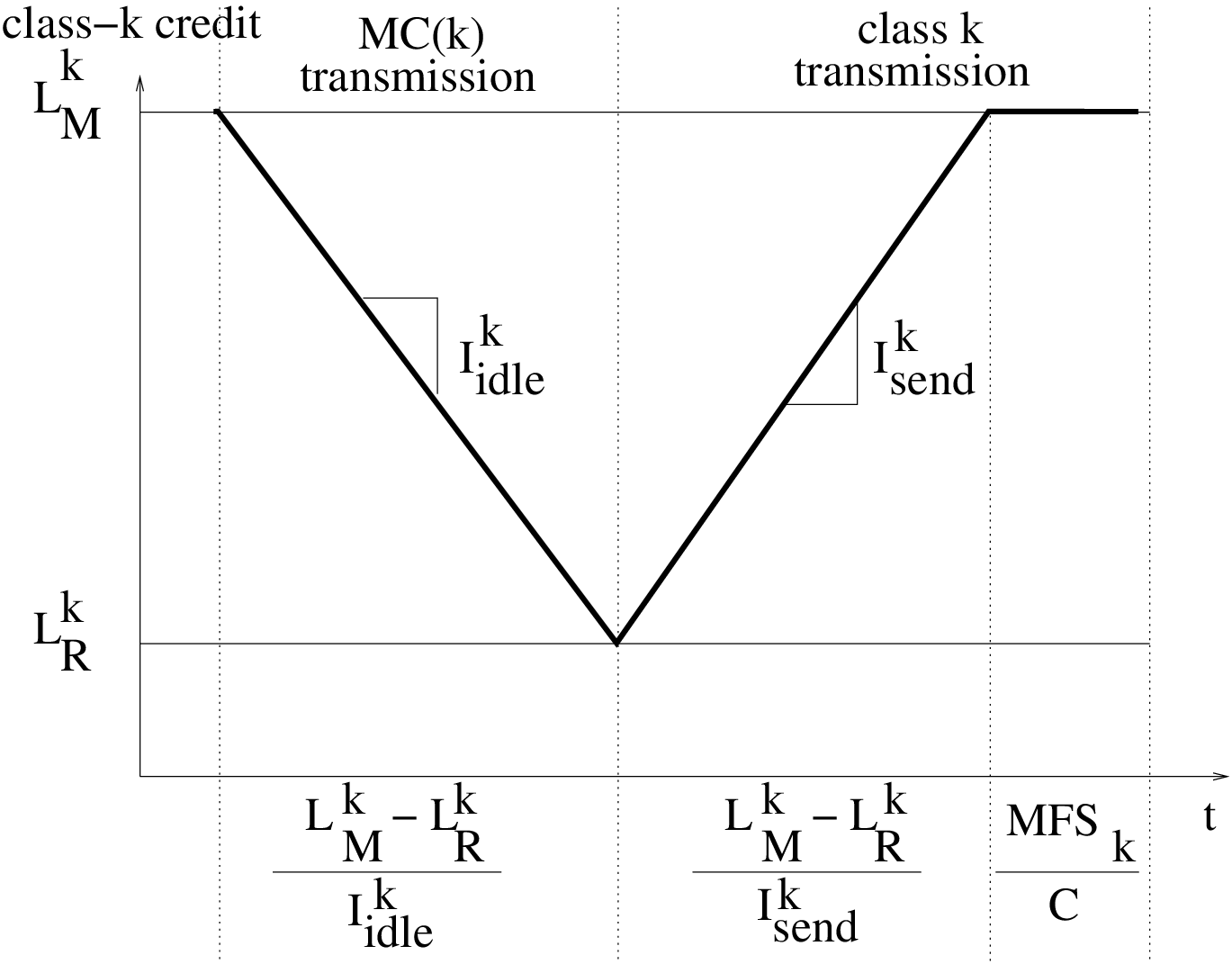}		
		\caption{Computing $\gamma_{k}^{bls}(t)$}		
		\label{fig:maxComputesketch1}
	\end{figure}
	Thus, we study the intervals of time between the start of two transmissions of non-preempted class-$k$ frames starting their transmission just before $L_M^k$ is reached. The smallest duration of such an interval is equal to the sum of:
	\begin{enumerate}
		\item the transmission time of the non-preempted class-$k$ frame (at the end of the transmission the credit is equal to $L_M^k$); 		
		\item the duration $ \frac{L_M^k-L_R^k}{I_{idle}^k}$ because  MC(k) traffic has to be sent continuously in order for the credit to reach $L_R^k$ in the minimum duration ;  		
		\item finally $\frac{L_M^k-L_R^k}{I_{send}^k}$ because class $k$ traffic has to be sent continuously in order for the credit to return in the minimum duration to $L_M^k$. 
	\end{enumerate}
	
	In total, the minimum duration between the start of the transmission of two non-preempted class-$k$ frames (each starting just before $L_M^k$ is reached), is $$\Delta^{k,\gamma}_{inter}= \frac{MFS_{k}}{C}+\frac{L_M^k-L_R^k}{I_{idle}^k}+\frac{L_M^k-L_R^k}{I_{send}^k}$$ 
	Thus during $\delta$, the number of time a non-preempted class-$k$ frame can be sent is upper bounded by $\lceil\frac{\delta}{\Delta^{k,\gamma}_{inter}}\rceil$.
	
	
	As a result, the amount class $k$ traffic sent while the credit is saturated is such as:	
	$$ \Delta R^{*}_{L_M^k, sat}(\delta)\leqslant MFS_{k}\cdot\lceil\frac{\delta}{\Delta^{k,\gamma}_{inter}}\rceil
	\leqslant
	MFS_{k}\cdot\left(  \frac{\delta}{\Delta^{k,\gamma}_{inter}}+1\right)$$		
	
\end{proof}

\subsection{Th.\ref{Th:k-minGene}: gCCbA strict minimum service curve}

\label{proofGenemin}

We search a strict minimum service curve offered to a class $k$ defined by a Rate-Latency curve, i.e.,  $\beta_{k}^{bls}(t)=\rho\cdot (t-\tau)^+$ with rate $\rho$ and initial latency $\tau$. 

The impact of other classes, are separated into three parts: the impact of $LC(k)$, $MC(k)$, $HC(k)$.


According to the definition of the strict minimum service curve, $\forall$ backlogged period $\delta$:

\begin{equation}
\Delta R^*_{k}(\delta) \geqslant \beta_{k}^{bls}(\delta) = \rho\cdot(\delta-\tau)^+	\label{propbeta}
\end{equation}

For any duration lower than $\tau$, the variation of the output is lower bounded by 0. 
$$\forall\delta\leqslant \tau,\Delta R^*_{k}(\delta)\geqslant 0$$
Thus, the best $\tau$ for our strict service curve is the largest duration during which no class $k$ traffic can be sent. So, when considering the impacts of the different classes we have:

\begin{enumerate}
	\item for traffic of Lower Classes $LC(k)$, the impact is the one computed with Static Priority: it is due to the non-preemption and is taken into account in the Static Priority model;	
	
	\item for traffic of Medium Classes $MC(k)$: the worst-case occurs if the credit starts at $L_M^k$, MC(k) frames are transmitted until $L_R^k$ is reached and due to non-preemption an additional MC(f) frame is sent. We denote this duration  $\Delta_{idle}^{k,\beta}$. So, we have: $$\Delta_{idle}^{k,\beta}=\frac{L_M^k-L_R^k}{I_{idle}^k}+ \frac{\max_{j\in MC(k)}MFS_{j}}{C}$$
	
	\item for  Higher Classes $HC(k)$, the impact is already computed with Static Priority and is not taken into account here.
\end{enumerate}

So finally, we have: 	
$$\tau=\Delta_{idle}^{k,\beta}$$	

Concerning the $\rho$, we search for a strictly positive rate. We use the definition of $\beta_{k}^{bls}$ as a Rate-Latency strict service curve and Eq. (\ref{propbeta}) to deduce a property of $\rho$.
We notice  the limit toward infinity of $\Delta R^*_{k}(\delta)$ over $\delta$ will be greater than $\rho$: $$\lim_{\delta\rightarrow +\infty} \frac{\Delta R^*_{k}(\delta)}{\delta}\geqslant \lim_{\delta\rightarrow +\infty}\rho\cdot \left( 1 - \frac{\tau}{\delta}  \right) =\rho.$$

So we look for a $x>0$ fulfilling the following condition: $$\lim_{\delta\rightarrow +\infty} \frac{\Delta R^*_{k}}{\delta} \geqslant x .$$

We now use the continuity property of the BLS credit to determine $x$.
From Lemma~\ref{minmaxsum}, we know that:
\begin{eqnarray}
\left( \begin{aligned}
\Delta R^*_{k}(\delta)-\frac{\Delta R^{*}_{L_M^k,sat}(\delta)}{C}\cdot I_{send}^k\\-(\delta - \frac{\Delta R^{*}_{0,sat}(\delta)}{C})\cdot I_{idle}^k
\end{aligned}\right) \nonumber  \geqslant  -L_M^k
\end{eqnarray}
Thus:	
$$\Delta R^*_{k}(\delta)\geqslant -L_M^k+\frac{\Delta R^{*}_{L_M^k,sat}(\delta)}{C}\cdot I_{send}^k+(\delta - \frac{\Delta R^{*}_{0,sat}(\delta)}{C})\cdot I_{idle}^k$$	

To find $x$, we must find a lower bound of $\lim\limits_{\delta\rightarrow+\infty}\frac{\Delta R^*_{k}(\delta)}{\delta}$, so we have:	
$$\frac{\Delta R^*_{k}(\delta)}{\delta}\geqslant \frac{-L_M^k}{\delta}+\frac{\Delta R^{*}_{L_M^k,sat}(\delta)}{\delta\cdot C}\cdot I_{send}^k+(1 - \frac{\Delta R^{*}_{0,sat}(\delta)}{\delta\cdot C})\cdot I_{idle}^k$$

\begin{eqnarray}\label{sumCredit1Gene}
\lim\limits_{\delta\rightarrow+\infty}\frac{\Delta R^*_{k}(\delta)}{\delta}\geqslant& \lim\limits_{\delta\rightarrow+\infty}&\frac{-L_M^k}{\delta}+\frac{\Delta R^{*}_{L_M^k,sat}(\delta)}{\delta\cdot C}\cdot I_{send}^k\nonumber\\&&+(1 - \frac{\Delta R^{*}_{0,sat}(\delta)}{\delta\cdot C})\cdot I_{idle}^k
\end{eqnarray}

We need the lower bound of $\Delta R^{*}_{L_M^k,sat}(\delta)$, and the upper bound of $\Delta R^{*}_{0,sat}(\delta)$. We use Lemmas~\ref{lemmasat0} and \ref{lemmasatlm} to compute the bounds. This gives:

\begin{equation}\label{deltaRLMsat}
\lim\limits_{\delta\rightarrow\infty}\frac{ \Delta R^{*,max}_{L_M^k,sat}(\delta)}{\delta}\geqslant 0
\end{equation}

\begin{equation}\label{deltaR0sat}
\lim\limits_{\delta\rightarrow\infty}\frac{ \Delta R^{*,max}_{0,sat}(\delta)}{\delta}\leqslant \sum\limits_{h\in HC(k)} r_h+ \frac{MFS_{MC(k)}^{sat}}{\Delta^{k,\beta}_{inter}}
\end{equation}

Thus, from Eq. (\ref{sumCredit1Gene}), Eq. (\ref{deltaRLMsat}), and Eq. (\ref{deltaR0sat}), we deduce:	
\begin{eqnarray}\label{sumCredit2Gene}
\lim\limits_{\delta\rightarrow+\infty}\frac{\Delta R^*_{k}(\delta)}{\delta}
&\geqslant& \lim\limits_{\delta\rightarrow+\infty}(1 - \frac{\Delta R^{*,max}_{0, sat}(\delta)}{\delta\cdot C})\cdot I_{idle}^k\nonumber=\left( C -\sum\limits_{h\in HC(k)} r_h-\frac{MFS_{MC(k)}^{sat}}{\Delta^{k,\beta}_{inter}}\right)\cdot \frac{I_{idle}^k}{C}\nonumber
\end{eqnarray}

Finally, we have found a  suitable $\rho$ such as: $\lim\limits_{\delta\rightarrow+\infty}\frac{\Delta R^*_{k}(\delta)}{\delta}\geqslant \rho$ with $$\rho=\left( C -\sum\limits_{h\in HC(k)} r_h-\frac{MFS_{MC(k)}^{sat}}{\Delta^{k,\beta}_{inter}}\right)\cdot \frac{I_{idle}^k}{C}.$$

\subsection{Th.\ref{Th:k-maxGene}: gCCbA maximum service curve}

\label{proofGenemax}
We search a maximum service curve offered to a class $k$ defined by a leaky-bucket curve, i.e.,  $\gamma_{k}^{bls}(t)=r\cdot t+b$ with rate $r$ and burst $b$.



From \cite{leboudecthiran12}, we know that $\Delta R^*_{k} (t-s)\leqslant B(s) +\gamma(t-s)$, with $B(s)$ the backlog at s. We search for $z$ such as $\Delta R^*_{k} (t-s)\leqslant z$. As $B(s)\geqslant 0$, we obtain:
$$\Delta R^*_{k} (t-s)\leqslant z \leqslant B(s)+z$$

Hence, with $\delta=t-s$, we select $\gamma(\delta)= z$, which gives:		

\begin{equation}\label{defgamma}
\Delta R^*_{k}(\delta) \leqslant \gamma_{k}^{bls}(\delta)
\end{equation}

In the absence of other traffic, class $k$ can use the full capacity of the link, so  $\Delta R^*_{k}(\delta) \leqslant C\cdot t$. Thus, we deduce that:		
$$\gamma_{k}^{bls} (t) = C\cdot t.$$		
In a MC(k) backlogged period, we use the definition of $\gamma_{k}^{bls}$ as a leaky-bucket maximum service curve to deduce a property of $r$ using Eq. (\ref{defgamma}).

\textit{Computing $r$}

We notice  the limit toward infinity of $\Delta R^*_{k}$ over $\delta$ will be lower than $r$: $$\lim_{\delta\rightarrow +\infty} \frac{\Delta R^*_{k}}{\delta}\leqslant \lim_{\delta\rightarrow +\infty}r+ \frac{b}{\delta}  =r$$ 

So we search for a strictly positive rate $x$, equal or lower than the link output rate $C$ fulfilling the following condition: $$\lim_{\delta\rightarrow +\infty} \frac{\Delta R^*_{k}}{\delta} \leqslant x $$

We use the continuity property of the BLS credit to determine $x$. From Lemma~\ref{minmaxsum}, we know that:
$$\Delta R^*_{k}(\delta)-\frac{\Delta R^{*}_{L_M^k,sat}(\delta)}{C}\cdot I_{send}^k-(\delta - \frac{\Delta R^{*}_{0,sat}(\delta)}{C})\cdot I_{idle}^k\leqslant L_M^k$$	

Thus,	
$$\Delta R^*_{k}(\delta)\leqslant L_M^k+\frac{\Delta R^{*}_{L_M^k,sat}(\delta)}{C}\cdot I_{send}^k+(\delta - \frac{\Delta R^{*}_{0,sat}(\delta)}{C})\cdot I_{idle}^k$$		
To find $x$, we must find a lower bound of $\lim\limits_{\delta\rightarrow+\infty}\frac{\Delta R^*_{k}(\delta)}{\delta}$, so we have:		
\begin{eqnarray}
\frac{\Delta R^*_{k}(\delta)}{\delta}\leqslant\nonumber \frac{L_M^k}{\delta}+\frac{\Delta R^{*}_{L_M^k,sat}(\delta)}{\delta\cdot C}\cdot I_{send}^k+(1 - \frac{\Delta R^{*}_{0,sat}(\delta)}{\delta\cdot C})\cdot I_{idle}^k\nonumber
\end{eqnarray}	
\begin{eqnarray}\label{sumCredit10Gene}
\lim\limits_{\delta\rightarrow+\infty}\frac{\Delta R^*_{k}(\delta)}{\delta}\leqslant \lim\limits_{\delta\rightarrow+\infty} \frac{L_M^k}{\delta}+\frac{\Delta R^{*}_{L_M^k,sat}(\delta)}{\delta\cdot C}\cdot I_{send}^k\nonumber+(1 - \frac{\Delta R^{*}_{0,sat}(\delta)}{\delta\cdot C})\cdot I_{idle}^k
\end{eqnarray}

We need the lower bound of $\Delta R^{*}_{0, sat}(\delta)$, and the upper bound of $\Delta R^{*}_{L_M^k, sat}(\delta)$.
We use Lemmas~\ref{lemmasat0} and \ref{lemmasatlm} to compute the bounds. This gives:	
\begin{equation}\nonumber
\lim\limits_{\delta\rightarrow\infty}\frac{ \Delta R^{*,max}_{L_M^k, sat}(\delta)}{\delta}\geqslant 0
\end{equation}
\begin{equation}\nonumber
\lim\limits_{\delta\rightarrow\infty}\frac{ \Delta R^{*,max}_{L_M^k, sat}(\delta)}{\delta}\leqslant \frac{MFS_{k}}{\Delta^{k,\gamma}_{inter}}
\end{equation}

Thus, from Eq. (\ref{sumCredit10Gene}), we deduce:	

\begin{eqnarray}\label{sumCredit20Gene}
\lim\limits_{\delta\rightarrow+\infty}\frac{\Delta R^*_{k}(\delta)}{\delta}\leqslant \lim\limits_{\delta\rightarrow+\infty}I_{idle}^k + \frac{\Delta R^{*,max}_{L_M^k, sat}(\delta)}{\delta\cdot C}\cdot I_{send}^k\nonumber= I_{idle}^k +\frac{MFS_{k}}{\Delta^{k,\gamma}_{inter}}\cdot \frac{I_{send}^k}{  C} \nonumber
\end{eqnarray}

We call $\Delta^{k,\gamma, send}_{inter}$ the interval during which class $k$ frames are sent, and $\Delta^{k,\gamma, idle}_{inter}$ the interval during which MC(k) frames are sent such as $\Delta^{k,\gamma}_{inter}=\Delta^{k,\gamma}_{send}+\Delta^{k,\gamma}_{idle}$.
$$\Delta^{k,\gamma}_{send}= \frac{MFS_{k}}{C}+\frac{L_M^k-L_R^k}{I_{send}^k}$$		
$$\Delta^{k,\gamma}_{idle}=\frac{L_M^k-L_R^k}{I_{idle}^k}$$

Using the definitions of the different expressions, we deduce that:
$$I_{idle}^k +\frac{MFS_{k}}{\Delta^{k,\gamma}_{inter}}\cdot \frac{I_{send }^k}{  C}=\frac{\Delta^{k,\gamma}_{send}}{\Delta^{k,\gamma}_{inter}}\cdot C <C$$
Finally, we have found a  suitable $r$,                                                                                                                                                                                      such as: $\lim\limits_{\delta\rightarrow+\infty}\frac{\Delta R^*_{k}(\delta)}{\delta}\leqslant r$ with $r=\frac{\Delta^{k,\gamma}_{send}}{\Delta^{k,\gamma}_{inter}}\cdot C$

Now that we have found $r$, we need to find $b$ such as $\forall$ MC(k) backlogged period $\delta$: $$\Delta{R^*_{k}}(\delta)\leqslant \frac{\Delta^{k,\gamma}_{send}}{\Delta^{k,\gamma}_{inter}}\cdot C\cdot \delta +b$$

\textit{Computing $b$}

We use the largest class $k$ burst that can be sent with the BLS.


In the presence of MC(k) traffic, the largest period of time during which class $k$ traffic can be sent continuously occurs if the credit started at 0. Then, class $k$ traffic is sent continuously until $L_M^k$ is reached and the priority is changed to its low value $p_L$. If a new class $k$ frame started its transmission just before the credit reached $L_M^k$ due to non-preemption, it will finish its transmission before the waiting MC(k) traffic can be sent. Thus, with a link capacity $C$ the largest class $k$ burst is  $b^{max}_{k}=\frac{C}{I_{send}^k}\cdot L_M^k+MFS_{k}$.
This gives:

$$\Delta R^*_{k}(\frac{b^{max}_{k}}{C}) \leqslant  b^{max}_{k} = \frac{\Delta^{k,\gamma}_{send}}{\Delta^{k,\gamma}_{inter}}\cdot
b^{max}_{k}+b$$

$$\Rightarrow b=b^{max}_{k} \cdot \frac{\Delta^{k,\gamma}_{idle}}{\Delta^{k,\gamma}_{inter}}$$

So this gives:
$\forall \delta\geqslant \frac{b^{max}_{k}}{C}$, $\Delta R^*_{k}(\delta)  \leqslant \frac{\Delta^{k,\gamma}_{send}}{\Delta^{k,\gamma}_{inter}}\cdot C \cdot\delta + b^{max}_{k} \cdot \frac{\Delta^{k,\gamma}_{idle}}{\Delta^{k,\gamma}_{inter}}$.

Additionally, we have:
$\forall \delta\leqslant \frac{b^{max}_{k}}{C}$, $\Delta R^*_{k}(\delta)  \leqslant C\cdot \delta\leqslant \frac{\Delta^{k,\gamma}_{send}}{\Delta^{k,\gamma}_{inter}}\cdot C \cdot\delta + b^{max}_{k} \cdot \frac{\Delta^{k,\gamma}_{idle}}{\Delta^{k,\gamma}_{inter}}$.

So, we have proved that $\forall \delta\in \mathds R^+$, $$\Delta R^*_{k}(\delta)  \leqslant \frac{\Delta^{k,\gamma}_{send}}{\Delta^{k,\gamma}_{inter}}\cdot C \cdot\delta + b^{max}_{k} \cdot\frac{\Delta^{k,\gamma}_{idle}}{\Delta^{k,\gamma}_{inter}}.$$

	\bibliographystyle{plain}
	\bibliography{JournalCCbA} 	
	
\end{document}

%% file: figures/discontinuity2.latex
\setlength{\unitlength}{4144sp}%
\begingroup\makeatletter\ifx\SetFigFont\undefined%
\gdef\SetFigFont#1#2#3#4#5{%
  \reset@font\fontsize{#1}{#2pt}%
  \fontfamily{#3}\fontseries{#4}\fontshape{#5}%
  \selectfont}%
\fi\endgroup%
\begin{picture}(9522,4801)(5026,-7550)
\thinlines
{\color[rgb]{0,0,0}\put(7651,-6811){\vector( 1, 0){6885}}
}%
\thicklines
{\color[rgb]{0,0,0}\put(7666,-6796){\line( 1, 2){1778.200}}
\put(9429,-3232){\line( 1,-2){1571}}
}%
\thinlines
{\color[rgb]{0,0,0}\multiput(11020,-5461)(0.00000,-90.00000){10}{\line( 0,-1){ 45.000}}
\put(11020,-6316){\vector( 0,-1){0}}
\put(11020,-5461){\vector( 0, 1){0}}
}%
{\color[rgb]{0,0,0}\multiput(9451,-3211)(0.00000,-9.00000){401}{\makebox(1.5875,11.1125){\small.}}
}%
{\color[rgb]{0,0,0}\multiput(11026,-6811)(0.00000,9.00000){51}{\makebox(1.5875,11.1125){\small.}}
}%
{\color[rgb]{0,0,0}\multiput(9451,-7036)(116.66667,0.00000){14}{\line( 1, 0){ 58.333}}
\put(11026,-7036){\vector( 1, 0){0}}
\put(9451,-7036){\vector(-1, 0){0}}
}%
{\color[rgb]{0,0,0}\put(7426,-6361){\line( 1, 0){450}}
}%
{\color[rgb]{0,0,0}\put(7651,-6811){\vector( 0, 1){4050}}
}%
{\color[rgb]{0,0,0}\put(7426,-3211){\line( 1, 0){450}}
}%
{\color[rgb]{0,0,0}\multiput(11026,-5461)(-9.00000,0.00000){376}{\makebox(1.5875,11.1125){\small.}}
}%
{\color[rgb]{0,0,0}\multiput(11071,-6361)(-9.00000,0.00000){376}{\makebox(1.5875,11.1125){\small.}}
}%
{\color[rgb]{0,0,0}\put(7381,-5461){\line( 1, 0){450}}
}%
\thicklines
{\color[rgb]{0,0,0}\put(11023,-5484){\line( 1, 2){1141.400}}
\put(12128,-3183){\line( 1,-2){1122}}
}%
\thinlines
{\color[rgb]{0,0,0}\multiput(12151,-3211)(0.00000,-9.00000){401}{\makebox(1.5875,11.1125){\small.}}
}%
{\color[rgb]{0,0,0}\multiput(11026,-7036)(118.42105,0.00000){10}{\line( 1, 0){ 59.211}}
\put(12151,-7036){\vector( 1, 0){0}}
\put(11026,-7036){\vector(-1, 0){0}}
}%
\put(14446,-7216){\makebox(0,0)[lb]{\smash{{\SetFigFont{20}{24.0}{\rmdefault}{\mddefault}{\updefault}{\color[rgb]{0,0,0}t}%
}}}}
\put(9901,-7486){\makebox(0,0)[lb]{\smash{{\SetFigFont{25}{30.0}{\rmdefault}{\mddefault}{\updefault}{\color[rgb]{0,0,0}$\Delta_{idle}^{k,max}$}%
}}}}
\put(11251,-7486){\makebox(0,0)[lb]{\smash{{\SetFigFont{25}{30.0}{\rmdefault}{\mddefault}{\updefault}{\color[rgb]{0,0,0}$\Delta_{send}^{k,min}$}%
}}}}
\put(6751,-3301){\makebox(0,0)[lb]{\smash{{\SetFigFont{20}{24.0}{\rmdefault}{\mddefault}{\updefault}{\color[rgb]{0,0,0}$L_M^k$}%
}}}}
\put(6751,-5551){\makebox(0,0)[lb]{\smash{{\SetFigFont{20}{24.0}{\rmdefault}{\mddefault}{\updefault}{\color[rgb]{0,0,0}$L_R^k$}%
}}}}
\put(8026,-3061){\makebox(0,0)[lb]{\smash{{\SetFigFont{20}{24.0}{\rmdefault}{\mddefault}{\updefault}{\color[rgb]{0,0,0}class-k credit}%
}}}}
\put(5041,-6406){\makebox(0,0)[lb]{\smash{{\SetFigFont{20}{24.0}{\rmdefault}{\mddefault}{\updefault}{\color[rgb]{0,0,0}$L_R^{k}-\frac{MFS}{C}\cdot I^k_{idle}$}%
}}}}
\end{picture}%